\documentclass[prb,superscriptaddress,twocolumn,aps,nofootinbib, amssymb,floatfix,a4paper]{revtex4-2}
\usepackage{amsmath,amsfonts,enumerate,hyperref,tikz,pgfplots,thmtools,amsthm, adjustbox}
\usepackage{mathtools, bbm, graphicx}
\usepackage{qcircuit}
\usepackage{comment}
\usepackage[most]{tcolorbox}   
\usepackage{enumitem}          
\usepackage{overpic}
\usepackage[normalem]{ulem} 
\usepackage{booktabs, multirow}

\newtcolorbox{summarybox}[1][]{
  enhanced,
  breakable,
  colback=gray!3, colframe=gray!50,
  boxrule=0.4pt, arc=2pt,
  left=6pt,right=6pt,top=6pt,bottom=6pt,
  fonttitle=\bfseries,
  fontupper=\small,
  title=Workflow summary,
  #1
}

\hypersetup{colorlinks=true,linkcolor=blue,citecolor=blue,urlcolor=blue}

\makeatletter
\newif\ifappendixonlytoc
\appendixonlytocfalse

\let\appendixonlytoc@addcontentsline\addcontentsline

\renewcommand{\addcontentsline}[3]{%
  \def\appendixonlytoc@file{#1}%
  \def\appendixonlytoc@toc{toc}%
  \ifx\appendixonlytoc@file\appendixonlytoc@toc
    \ifappendixonlytoc
      \appendixonlytoc@addcontentsline{#1}{#2}{#3}%
    \fi
  \else
    \appendixonlytoc@addcontentsline{#1}{#2}{#3}%
  \fi
}
\makeatother

\def\C{ {\cal C} }

\def\R{ \mathbb{R} }

\def\Ti{T_{\mathrm{ion}}}

\newcommand{\bra}[1]{\langle {#1} |}
\newcommand{\ket}[1]{| {#1} \rangle}
 
\newcommand{\ketbra}[2]{\ensuremath{\left|#1\right\rangle\!\!\left\langle#2\right|}}
\newcommand{\braket}[2]{\ensuremath{\!\left\langle#1|#2\right\rangle}\!}

\renewcommand{\v}[1]{\ensuremath{\boldsymbol #1}}

\newcommand{\m}[1]{\mathsf{#1}}
\newcommand{\mc}[1]{\mathcal{#1}}
\newcommand{\kb}{k_{\mathrm{B}}}

\newcommand{\ve}{\v{e}}
\newcommand{\vj}{\v{j}}

\newcommand{\vn}{\v{n}}
\newcommand{\vx}{\v{x}}
\newcommand{\vv}{\v{v}}

\newcommand{\val}{\v{\alpha}}

\newcommand{\sof}{s_{\mathrm{F}}}
\newcommand{\soh}{s_{\mathrm{H}}}

\newcommand{\dr}{d_{r}}
\newcommand{\dc}{d_{c}}

\newcommand{\NH}{N_{\mathrm{H}}}
\newcommand{\NF}{N_{\mathrm{F}}}
\newcommand{\NC}{N_{\mathrm{C}}}

\newcommand{\epsC}{\epsilon_{\mathrm{C}}}

\newcommand{\lde}{\lambda_{\mathrm{D},e}}
\newcommand{\ldion}{\lambda_{\mathrm{D}}}

\newcommand{\vth}{v_{\mathrm{th}}}

\newcommand{\TL}{\mathbb{T}_L}

\newcommand{\ee}{\mathrm{e}}

\theoremstyle{plain}
\newtheorem{thm}{Theorem}
\newtheorem{theorem}{Theorem}

\newtheorem{lemma}[thm]{Lemma}
\newtheorem{prop}[thm]{Proposition}
\newtheorem{proposition}[thm]{Proposition}

\theoremstyle{definition}

\newtheorem{rmk}[thm]{Remark}

\newtheorem{ass}{Assumption}
\newtheorem{problem}{Problem}

\definecolor{pptgreen}{RGB}{215,241,206}
\definecolor{pptyellow}{RGB}{255,255,122}
\definecolor{pptred}{RGB}{255,122,122}

\usetikzlibrary{arrows.meta}

\begin{document}

\title{An end-to-end quantum algorithm for weakly nonlinear plasma physics with superquadratic speedup}

        \author{Bjorn K. Berntson}
        \affiliation{PsiQuantum, 700 Hansen Way, Palo Alto, CA 94304, USA}
	\author{David Jennings}
        \affiliation{PsiQuantum, 700 Hansen Way, Palo Alto, CA 94304, USA}
	\author{Matteo Lostaglio}
        \affiliation{PsiQuantum, 700 Hansen Way, Palo Alto, CA 94304, USA}
        	\author{Scott Parker}
        \affiliation{Renewable and Sustainable Energy Institute, University of Colorado, Boulder, Colorado 80309, USA}

\date{\today}
\begin{abstract}
Nonlinear kinetic plasma simulation is high-dimensional and classically demanding, while quantum algorithms face different bottlenecks: embedding nonlinear dynamics into a linear computation, loading potentially dense field-interaction data, and efficiently extracting information. Here we present an end-to-end quantum algorithm, with rigorous convergence guarantees, for a weakly nonlinear kinetic plasma model. The system describes a three-dimensional electron--ion plasma with adiabatic electrons, kinetic ions, Debye screening, and Krook relaxation. After Fourier--Hermite truncation, the dynamics reduces to a high-dimensional quadratic ordinary differential equation. To tackle quantum bottlenecks we combine three key ingredients. First, we use a plasma free energy to identify a Lyapunov transform under which a Carleman linear embedding converges exponentially in the truncation order within a certified weakly nonlinear regime. Second, we develop a hierarchical block-encoding protocol for dense matrices, exploiting the spatial decay of the screened field to avoid polynomial overhead from direct sparse access encodings of dense matrices. Third, we introduce a subroutine for information extraction that exploits the nonlinear components encoded in the full Carleman history state to improve the estimation of linear observables. Combining these ingredients, we construct a quantum algorithm to estimate the spacetime-averaged kinetic energy using
$\widetilde{O}\!\left(
N_{\mathrm{F}} N_{\mathrm{H}}^{1/2}
\operatorname{polylog}\!\left(\frac{T}{\epsilon}\right)
\frac{1}{\epsilon}
\right)$ gates and $
\widetilde{O}\!\left(
\log\!\left(N_{\mathrm{F}} N_{\mathrm{H}}^{1/2}T\right)\log\!\left(\frac{1}{\epsilon}\right)
\right)$ qubits, where \(N_{\mathrm{F}} \) and \(N_{\mathrm{H}}\) are the Fourier and Hermite cutoffs respectively. Relative to a standard Fourier--Hermite spectral solver, this yields exponential memory savings and superquadratic improvements in time. Together, these results establish a controlled nonlinear plasma benchmark for quantum simulation.
\end{abstract}

	\maketitle

\renewcommand{\thefootnote}{\arabic{footnote}}
	\setcounter{footnote}{0}

    \section{Introduction and overview}
    
    Plasma simulation is a central computational challenge in physics. Kinetic plasma models describe the evolution of a particle distribution function on phase space, and are capable of capturing effects such as Landau damping, phase mixing, wave--particle resonance, kinetic instabilities, and velocity-space filamentation. These effects are essential in many regimes of interest, but they also make first-principles simulation expensive: even before discretization error is considered, the state variable is a function of three spatial and three velocity coordinates. A natural question is whether fault-tolerant quantum computers can mitigate this challenge~\cite{joseph2023quantum}.

    Early quantum algorithms for \emph{linear} plasma dynamics showed that such problems can sometimes be mapped to Hamiltonian simulation or quantum ODE solving~\cite{engel2019quantum}. Subsequent work~\cite{ameri2023quantum} revisited the same physical problem and showed that the choice of representation is critical: a Fourier--Hermite system is exponentially more compact than the velocity-grid formulation for the linear Landau damping benchmark,
    so a classical Hermite solver can match the performance of the earlier grid-based quantum approach. At the same time, they showed that a quantum algorithm for the Fourier--Hermite formulation still gives a quadratic speedup in the Hermite system size over classical algorithms for the same truncated linear system~\cite{ameri2023quantum}.

Linearized plasma models provide a starting point for understanding waves, instabilities, and phase mixing, but many physically relevant regimes are intrinsically nonlinear. In such settings, kinetic simulations must additionally resolve mode coupling, nonlocal field interactions, saturation mechanisms, and the transfer of free energy across phase-space scales~\cite{Mouhot_2011}. 
   
   In this work, we construct an end-to-end quantum algorithm with a priori convergence guarantees for a 3+3D \emph{nonlinear} kinetic plasma problem. Specifically, our work focuses on the nonlinear Vlasov-Poisson equations with collisions where we expand around a background. This  simplified model captures two essential challenges faced by quantum algorithms for plasma: the presence of nonlinearities and the problem of data-loading for dense electric field terms. Prior work~\cite{vaszary2025solving} constructed a quantum algorithm for the one-dimensional nonlinear Vlasov-Poisson equation under a direct Carleman embedding approach that outputs a coherent state. However,  
   it was found that convergence could only be certified in physically unrealistic regimes.
   Furthermore, the quantum algorithm displayed worse-than-classical complexity scaling, even before accounting for the question of data extraction from a quantum state.

   To overcome these limitations, we introduce a perturbative Carleman expansion, use a Fourier--Hermite representation, and incorporate a Lyapunov stability analysis. Using a Lyapunov $R$-number criterion~\cite{jennings2025quantum}, we obtain provable convergence guarantees for small, but non-negligible nonlinearities. We also address the dense field-interaction terms using hierarchical block-encoding techniques. Our work provides a rigorous end-to-end quantum algorithm for a nonlinear kinetic plasma model in a controlled perturbative regime. We prove superquadratic quantum speedup over standard Fourier--Hermite spectral methods. The a priori analytical guarantees only apply to nonlinearities that are at least two orders of magnitude smaller than what is typically desired in practice (see Fig.~\ref{fig:schematicconvergence}). However, we expect that the algorithm will converge beyond the regime where strict analytical arguments are available, and Koopman-von Neumann linearization and related approaches~\cite{novikau2025quantum, joseph2020koopman, bravyi2025quantum, bravyi2026quantum, jemcov2026unitary, may2025second, jennings2026quantum} may offer a more robust route to strongly nonlinear or turbulent plasma dynamics. For the ease of reading, please refer to Appendix~\ref{app:notation} for a table of notation.

   \begin{center}
\begin{figure*}
\resizebox{0.8\textwidth}{!}{
\begin{tikzpicture}[x=1in,y=1in]
  \tikzset{
    pptline/.style={draw=black,line width=1.5pt,line cap=butt},
    pptarrow/.style={pptline,-{Triangle[length=5pt,width=6pt]}},
    pptdoublearrow/.style={pptline,{Triangle[length=4.5pt,width=5.5pt]}-{Triangle[length=4.5pt,width=5.5pt]}},
    regionlabel/.style={align=center,inner sep=0pt,text=black,font=\sffamily\fontsize{10}{12}\selectfont},
    ticklabel/.style={align=center,inner sep=0pt,text=black,font=\sffamily\fontsize{14}{16}\selectfont},
    axislabel/.style={align=left,inner sep=0pt,text=black,font=\sffamily\fontsize{12}{14}\selectfont}
  }

  \def\xzero{0.0000}
  \def\xone{0.6958}
  \def\xtwo{1.3916}
  \def\xthree{2.0874}
  \def\xfour{2.7831}
  \def\xfive{3.4789}
  \def\xsix{4.1747}
  \def\xseven{4.8705}
  \def\xredstart{4.5355}
  \def\xredend{5.7645}

  \def\xgreenend{1.6129}

  \fill[pptgreen]  (0.0084,0.0113) rectangle (\xgreenend,1.2838);
  \fill[pptyellow] (\xgreenend,0.0113) rectangle (\xredstart,1.2838);
  \fill[pptred]    (\xredstart,0.0113) rectangle (\xredend,1.2838);

  \draw[pptarrow] (0,0) -- (6.2080,0);
  \foreach \x in {\xzero,\xone,\xtwo,\xthree,\xfour,\xfive,\xsix,\xseven}
    \draw[pptline] (\x,-0.1299) -- (\x,0.1299);

  \draw[pptline] (0.1559,-0.0627) -- (0.2652,0.0672);
  \draw[pptline] (0.2187,-0.0605) -- (0.3280,0.0694);

  \node[ticklabel,anchor=north] at (\xzero,-0.2050) {$0$};
  \node[ticklabel,anchor=north] at (\xone,-0.2050) {$10^{-6}$};
  \node[ticklabel,anchor=north] at (\xtwo,-0.2050) {$10^{-5}$};
  \node[ticklabel,anchor=north] at (\xthree,-0.2050) {$10^{-4}$};
  \node[ticklabel,anchor=north] at (\xfour,-0.2050) {$10^{-3}$};
  \node[ticklabel,anchor=north] at (\xfive,-0.2050) {$10^{-2}$};
  \node[ticklabel,anchor=north] at (\xsix,-0.2050) {$10^{-1}$};
  \node[ticklabel,anchor=north] at (\xseven,-0.1910) {$1$};

  \node[regionlabel,text width=1.50in] at (0.8100,0.6476)
    {\textbf{Very weak}\\[-1pt]\textbf{nonlinearities}\\[-1pt]Analytically proved convergence};

  \node[regionlabel,text width=2.55in] at (3.0500,0.6476)
    {\textbf{Weak nonlinearities}\\[-1pt]Convergence threshold still to be determined};

  \node[regionlabel,text width=1.02in] at (5.1500,0.6476)
    {\textbf{Model breaks}\\[-1pt]\textbf{down}};

  \draw[pptdoublearrow] (\xfour,0.2235) -- (\xfive,0.2235);
  \node[inner sep=0pt,align=center,font=\sffamily\fontsize{8}{9}\selectfont] at (3.1310,0.3310) {tokamak core};

  \node[axislabel,anchor=north west,text width=1.37in] at (5.1600,-0.0850)
    {Nonlinearity $\varphi_{\max}$};
\end{tikzpicture}
}
\caption{\textbf{Narrowing down the convergence regime.} Schematic illustration of convergence regimes for the weakly nonlinear Vlasov-Poisson model with collisions used to benchmark our quantum algorithm. As the nonlinearity strength $\varphi_{\max}$ increases,
the problem of simulating the plasma model defined by Eqs.~(\ref{eq:gnonlinearintro}--\ref{eq:phi_g_single_speciesintro}) passes from an analytically certified Carleman-convergence regime,
through a weakly nonlinear regime whose practical convergence threshold
remains to be determined numerically or by alternative methods. Finally we pass to the limits of the model \(\varphi_{\max}=O(1)\), where the weakly
nonlinear plasma model itself ceases to apply. The indicated values are computed for the Gaussian example in Fig.~\ref{fig:Gaussian_plot}. The certified boundary near \(10^{-5}\) is shown for the representative case
\(\bar{\nu}=10^{-2}\) at Debye-scale resolution \(\Delta x=\pi\);
at fixed discretization and initial data parameters, the corresponding certified
threshold value of \(\varphi_{\max}\) scales linearly
with \(\bar{\nu}\); see \eqref{eq:Gaussian_RP}.}
\label{fig:schematicconvergence}
\end{figure*}
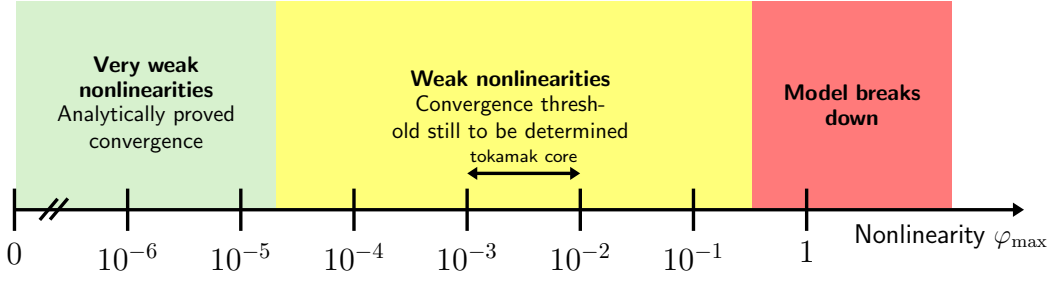
\end{center}
  
    \subsection{Overview of the plasma problem}
    
  We take a weakly nonlinear, electrostatic, unmagnetized, near-Maxwellian, collisional ion-electron plasma in three dimensions, with one thermal species and one kinetic species. The electrons are thermalized, according to the adiabatic response. The ions are kinetic (i.e., treated dynamically), and we perform an expansion about a Maxwell-Boltzmann equilibrium distribution. This idealized model retains the self-consistent electric-field nonlinearity characteristic of Vlasov--Poisson dynamics, whose nonlinear evolution, phase-space filamentation, and nonlocal field solve are central sources of computational difficulty in deterministic kinetic plasma simulation~\cite{ChengKnorr1976,FilbetSonnendruckerBertrand2001,ParkerDellar2015,AdkinsSchekochihin2018}. Related kinetic-ion/adiabatic-electron reductions are standard in treatments of ion-acoustic dynamics and in electrostatic gyrokinetic models of ion-temperature-gradient turbulence~\cite{FriedGould1961,LiuEtAl2023,ParkerLeeSantoro1993,DimitsEtAl2000}. As such, this model is an ideal testbed to investigate whether or not quantum computers can be utilized for nonlinear plasma simulation. 

A derivation of the plasma model is given in Sec.~\ref{sec:detailed}, but here is the upshot. The electrons are in local thermal equilibrium determined by the (time-varying) electrostatic potential $\phi$. 
A core quantity is the ratio of electrostatic potential energy to electron
thermal energy: 
\begin{align}
\label{eq:phibar}
\varphi(\v{x},t) \coloneqq {\lvert q_e\rvert \phi(\vx,t) \over k_B T_e},
\end{align}
where $q_e$ is the electron charge and $k_{\mathrm{B}}$ is Boltzmann's constant. This dimensionless potential is a physically meaningful measure of nonlinearity. For weak turbulence regimes, such as core tokamak turbulence and transport modeled using the gyrokinetic ordering, the normalized electrostatic potential \(\varphi(\v{x},t)\) is treated as a small parameter~\cite{BrizardHahm2007, Krommes2012}. For tokamak core parameters $\varphi(\vx,t)$
typically ranges between \(10^{-3}\) and \(10^{-2}\). For this reason we linearize the adiabatic
electron response.

The ions are taken to have charge equal to $|q_e|$. Their distribution function $f(\vx,\vv,t)$ is dynamic, and it satisfies the Vlasov-Poisson system. We also include a linear BGK-type collisional model given by the Krook operator. We assume a Maxwellian equilibrium for the ions 
\begin{align}
\label{eq:Maxwellian}
f_{0}(\v{v})\coloneqq n_{0} \bigg(\frac{m}{2\pi k_{\mathrm{B}}\Ti}\bigg)^{3/2}\exp\bigg(-\frac{m\lvert\v{v}\rvert^2}{2 k_{\mathrm{B}}\Ti} \bigg),
\end{align}
where $m$, $T_{\mathrm{ion}}$, $n_{0}$ are the ion mass, equilibrium temperature, and equilibrium density. We then derive evolution equations for the relative perturbation of the ion phase space distribution about the Maxwellian background
\begin{align}
\label{eq:perturbationg}
    g(\v{x}, \v{v}, t) \coloneqq \frac{f(\v{x}, \v{v}, t) - f_0(\v{v})}{f_0(\v{v})}.
\end{align} 
As detailed in Sec.~\ref{sec:detailed}, in natural plasma units we end up with the following equations of motion for  $g$:
\small
\begin{equation}\label{eq:gnonlinearintro}
\partial_t g=-\vv\cdot\nabla_x g-\vv\cdot\nabla_x\phi-(\vv\cdot \nabla_x \phi)g+\nabla_x\phi\cdot\nabla_v g-\bar{\nu} g,
\end{equation}
\normalsize
with $\bar{\nu}$ the (rescaled) collisional rate. Here $g$ is assumed to be square integrable with respect to the weight $f_0$, as expected from the fact that the perturbation should have finite free energy relative to equilibrium.

Due to the thermal background, we have \emph{electrostatic screening}, whereby the Coulomb electric field is deformed essentially to Yukawa form via the Helmholtz equation
\begin{equation}
\label{eq:phi_g_single_speciesintro}
\big(\nabla_x^2-\tau\big)\phi
=
-\frac{1}{(2\pi)^{3/2}}\int_{\R^3}g(\v{x},\v v,t)\ee^{-\lvert\vv\rvert^2/2}\,\mathrm d^3v,
\end{equation}
where  $\tau \coloneqq \Ti/T_e$ is the ratio between the ion and electron temperatures. 

Our benchmark problem is end-to-end, and therefore we must choose an observable of interest, which we choose to be the out of equilibrium kinetic energy. For the ion-acoustic damping problem considered here, the field energy is
transferred to particle degrees of freedom through phase mixing and nonlinear
wave--particle dynamics, while Krook collisions eventually damp the
perturbation. The spacetime-averaged kinetic-energy moment provides a simple
diagnostic of this transfer. In future applications to large-scale direct numerical simulation of turbulent transport, one could extract functionals of the local particle and energy fluxes, which are all quadratic in perturbation strength.

To assess prospects for quantum computers to provide speedups over classical methods we consider a spectral Galerkin approximation~\cite{canuto2006spectral} to the plasma system defined by Eqs.~\eqref{eq:gnonlinearintro} and~\eqref{eq:phi_g_single_speciesintro}, where we project onto $(2\NF+1)^3$ Fourier modes for the spatial degrees of freedom, and $(\NH+1)^3$ Hermite modes for the velocity degrees of freedom. For the spatial degrees of freedom we consider a torus $\mathbb{T}_L^3$, represented by a fundamental domain $[0,L]^3$ with periodic boundary conditions imposed on opposite faces. The integers $(\NF,\NH)$ are resolution parameters for the Galerkin approximation of the plasma system. Within this context, we shall consider the following problem.
\begin{problem}
\label{problem}
    Given $L>0$ and $\epsilon>0$, consider the Galerkin approximation of the plasma system \eqref{eq:gnonlinearintro} and~\eqref{eq:phi_g_single_speciesintro} on $\mathbb{T}_L^3 \times \mathbb{R}^3$ with resolution parameters $(\NF,\NH)$, and fix initial data $g(\v{x},\v{v},0)$ that is square-integrable with respect to the measure $f_0(\v{v})$, and has zero total mass:
    \begin{align}
         \frac{1}{{(2\pi)^{3/2}}} \int_{\mathbb T_L^3\times\mathbb{R}^3}
     g(\vx,\vv,0)\ee^{-|\v{v}|^2/2}\,\mathrm{d}^3v\,\mathrm{d}^3x =0.
    \end{align}
    For a given time $T>0$, output the time-averaged kinetic energy moment 
    \small
\begin{align}
\label{eq:kineticperturbation}
    \langle \mathcal{K} \rangle = \frac{1}{L^3 T} \int_0^T \mathrm{d}t
    \int_{\mathbb T_L^3\times\mathbb{R}^3}
   \! \! \! \!  \frac{|\vv|^2}{2}\,
    \frac{\ee^{-|\v{v}|^2/2}}{(2\pi)^{3/2}}\,
    g(\vx,\vv,t)\,\mathrm{d}^3v\,\mathrm{d}^3x ,
\end{align}
\normalsize
to $\epsilon$-additive error and constant success probability.
\end{problem}

\subsection{Overview of contributions}

We outline the challenges in simulating nonlinear plasma physics, describe our results in this context, and elaborate on the significance of these results. 

\subsubsection{The challenges}

If we were to set the nonlinear term in Problem~\ref{problem} to zero, and consider the $1$-dimensional rather than the $3$-dimensional problem, we would recover the linear Vlasov-Poisson equation with collisions. In this context, Ref.~\cite{ameri2023quantum} established the existence of a quadratic speedup in the number of Hermite modes $\NH$. However, handling nonlinearities in 3 dimensions is considerably more difficult. This is because here nonlinearities introduce two crucial new challenges:

\begin{enumerate}
    \item \emph{Faithful simulation:} Since quantum mechanics is fundamentally linear, the nonlinear plasma problem needs to be mapped onto a (large) linear problem before one can deploy quantum algorithms. Achieving this mapping with controlled errors is a central roadblock to a faithful quantum simulation~\cite{forets2017explicit,chen2024carleman,liu2021efficient,jennings2025quantum,jennings2025end,bravyi2025quantum}.

    \item \emph{Dense encodings:} The nonlinear term, which encapsulates the (screened) electrostatic interaction, is represented by dense matrices. This means that we need to contend with a quantum data loading problem, which can prohibit quantum speedup, as in Ref.~\cite{vaszary2025solving}.

\end{enumerate}

\subsubsection{Main result}

Our main contribution is to overcome these two core challenges in a simplified but representative nonlinear plasma system. The end-to-end Problem~\ref{problem} can be efficiently solved for a variety of initial conditions with sufficiently weak nonlinearity parameter~\eqref{eq:phibar} that can be efficiently prepared on a quantum computer. As an illustration:

\begin{theorem}
\label{theorem}
   Take Problem~\ref{problem} with single Fourier mode initial data: 
\begin{align}
\label{eq:initialconditiontheorem}
g(\vx,\vv,t=0)
=&\;
\frac{\varphi_{\max} }{\tau}\kappa_{\v{n}_0}\cos\bigg(\frac{2\pi}{L} \v{n}_0 \cdot \vx\bigg), \\
\phi(\vx,t=0)
=&\;
\frac{\varphi_{\max}}{\tau} \cos\bigg(\frac{2\pi}{L} \v{n}_0 \cdot \vx\bigg),
\end{align}
where $ \kappa_{\v{n}_0} =\frac{4\pi^2}{L^2}|\v{n}_0|^2 +\tau$ with $\v{n}_0$ a nonzero vector of integers with magnitudes smaller than $\NF$ and
\begin{align}
\label{eq:epsmax_Fourier_mode_explicitintro}
\varphi_{\max}
<
&\;
\frac{\sqrt{2}\tau\bar{\nu}}{\sqrt{3\kappa_{\vn_0}(1+\kappa_{\vn_0})}}
\nonumber \\
&\; \times\bigg(
\frac{1}{1+\tau}+\frac{13L^2}{2\pi^2}\NF
\bigg)^{-1/2} \!\!\!
\big(1+\sqrt{3\NH}\big)^{-1}.
\end{align}
Then Problem~\ref{problem} can be solved by a quantum algorithm involving
\begin{align}
   \mathcal{Q}_{\mathrm{mem}} =\tilde{O}\bigg(\log(\NF \NH^{1/2} T) \log\bigg(\frac{1}{\epsilon}\bigg)\bigg)
\end{align}
qubits and
\begin{align}
\label{eq:overallcostintro}
\mathcal{Q}_{\mathrm{op}} = \tilde{O}\left(\NF\NH^{1/2}\operatorname{polylog} \left(\frac{T}{ \epsilon}\right) \frac{1}{\epsilon}\right),
\end{align}
one and two-qubit gates and classical operations. 
\end{theorem}

We also test the same convergence criterion on a more localized family of initial data obtained from Gaussian density perturbations. In the large-box regime relevant for the plots, this family is naturally viewed through the radial approximation
\small
\begin{equation}
\label{eq:g_rad_intro}
g_{\mathrm{rad}}(\vx,\vv,0)
=
\frac{\varphi_{\max}}{\tau S_1}
\bigg(
\frac{1}{(2\pi)^{3/2}\sigma^3}
\exp\bigg(-\frac{\lvert\vx\rvert^2}{2\sigma^2}\bigg)
-
\frac{1}{L^3}
\bigg),
\end{equation}
\normalsize
where $g_{\mathrm{rad}}$ is independent of \(\vv\) and $S_1$ is a lattice sum defined by 
\begin{equation}
\label{eq:S1}
S_1\coloneqq \frac{1}{L^3}\sum_{\vn\in \mathbb{Z}^3\setminus\{0\}} \frac{\exp\big(-\frac{2 \pi^2\sigma^2}{L^2}\lvert \vn \rvert^2 \big)}{\frac{4\pi^2}{L^2}\lvert{\vn}\rvert^2+\tau}.
\end{equation}
The radial potential corresponding to \eqref{eq:g_rad_intro} is determined by the Helmholtz equation~\eqref{eq:phi_g_single_speciesintro}; see Eq.~\eqref{eq:phi_radial_approx} for the explicit expression. The parameter \(\sigma\) controls the width of the Gaussian and, correspondingly, how far the initial datum is from spatial equilibrium. In the exact periodized Gaussian family, increasing $\sigma$ makes the density profile increasingly uniform on the torus, so that the mean-subtracted perturbation approaches zero. Decreasing \(\sigma\) instead produces a more localized density perturbation, larger spatial gradients, and a more pronounced screened electrostatic response. Therefore, this one-parameter Gaussian family interpolates between equilibrium initial data and localized nonequilibrium initial data, complementing the single-mode benchmark of Theorem~\ref{theorem}.  

Again, for this setting a result fully analogous to  Theorem~\ref{theorem} holds, with the only variation being a modification of inequality~\eqref{eq:epsmax_Fourier_mode_explicitintro}. See Section~\ref{subsec:initial_data} and Appendix~\ref{subsec:Gaussian_proof} for further details.

\subsubsection{Regimes of a priori provable convergence}

The previous results show that the quantum algorithm is rigorously guaranteed to perform a faithful simulation if the condition~\eqref{eq:epsmax_Fourier_mode_explicitintro} (and a similar condition for the localized perturbation) is satisfied. Does this include physically relevant regimes?
To answer this question, we note that the parameter $\varphi_{\max} $ in both the Fourier perturbation (Eq.~\eqref{eq:initialconditiontheorem}) and the localized perturbation (Eq.~\eqref{eq:g_rad_intro}) coincides with the measure of nonlinearity in Eq.~\eqref{eq:phibar}, maximized over $\v{x}$ at $t=0$:
\begin{align}
    \varphi_{\max} = \max_{\v{x} \in \TL^3}\, \lvert \varphi(\v{x},0)\rvert.
\end{align}
We can therefore plot the largest $\varphi_{\max}$ for which the perturbation satisfies the assumptions required by the algorithm's a priori performance guarantees, as a function of resolution and collisional strength.
Fig.~\ref{fig:Fourier_plot}  shows this for the Fourier perturbation and Fig.~\ref{fig:Gaussian_plot} for the localized perturbation. The localized perturbation allows us to certify a larger \(\varphi_{\max}\) compared to the delocalized Fourier mode. In fact, the more the starting perturbation is localized, the higher the value of \(\varphi_{\max}\).   

The results in Fig.~\ref{fig:Fourier_plot} and Fig.~\ref{fig:Gaussian_plot} bound the initial state nonlinearities. Since the plasma dynamics has non-normality in it, we can have transient growth in the size of the perturbations. This growth in perturbations will in turn correspond to an increase in the nonlinearity parameter $\varphi$ and is also accommodated by our algorithm. However, the size of possible transient growth is relatively limited,\footnote{See Proposition~\ref{prop:u2Pbounds} for an upper bound on the transient growth of perturbations.} and whether it is realized for given initial data is a separate matter.
 
Even accounting for the above, the upshot is that the analytical convergence guarantees only apply to high collisional rates and small nonlinearities, below what is normally considered in most applications. However, we expect the algorithm to apply beyond the regime where our a priori analysis applies, which motivates future investigation.

The choice $\Delta x\simeq \pi$ corresponds to resolving the
screening scale in Fourier space for the $\tau=1$ cases plotted here. Indeed,
with modes $\vn\in\{-\NF,\ldots,\NF\}^3$, the largest coordinate wavenumber is
$\xi_{\max}\coloneqq 2\pi\NF/L$, while the associated collocation spacing is
$\Delta x=L/(2\NF+1)$. A screened response has Fourier multiplier controlled
by $(\lvert\v{\xi}_{\vn}\rvert^2+\tau)^{-1}$ (see Eq.~\ref{eq:phi_g_single_speciesintro}), so the spectral crossover occurs when $\lvert\v{\xi}_{\vn}\rvert^2\sim\tau$.
For $\tau=1$, this gives $\xi_{\max}\gtrsim1$, or equivalently
\begin{equation}
    \Delta x\lesssim \frac{2\pi\NF}{2\NF+1}\simeq\pi.
    \end{equation}
This is the usual Nyquist factor: a mode with wavenumber $\lvert\v{\xi}_{\vn}\rvert$ has wavelength
$2\pi/\lvert\v{\xi}_{\vn}\rvert$, so two samples per wavelength give spacing $\pi/\lvert\v{\xi}_{\vn}\rvert$.

\begin{figure}[t]
    \centering    
    \begin{overpic}[width=0.35\textwidth]{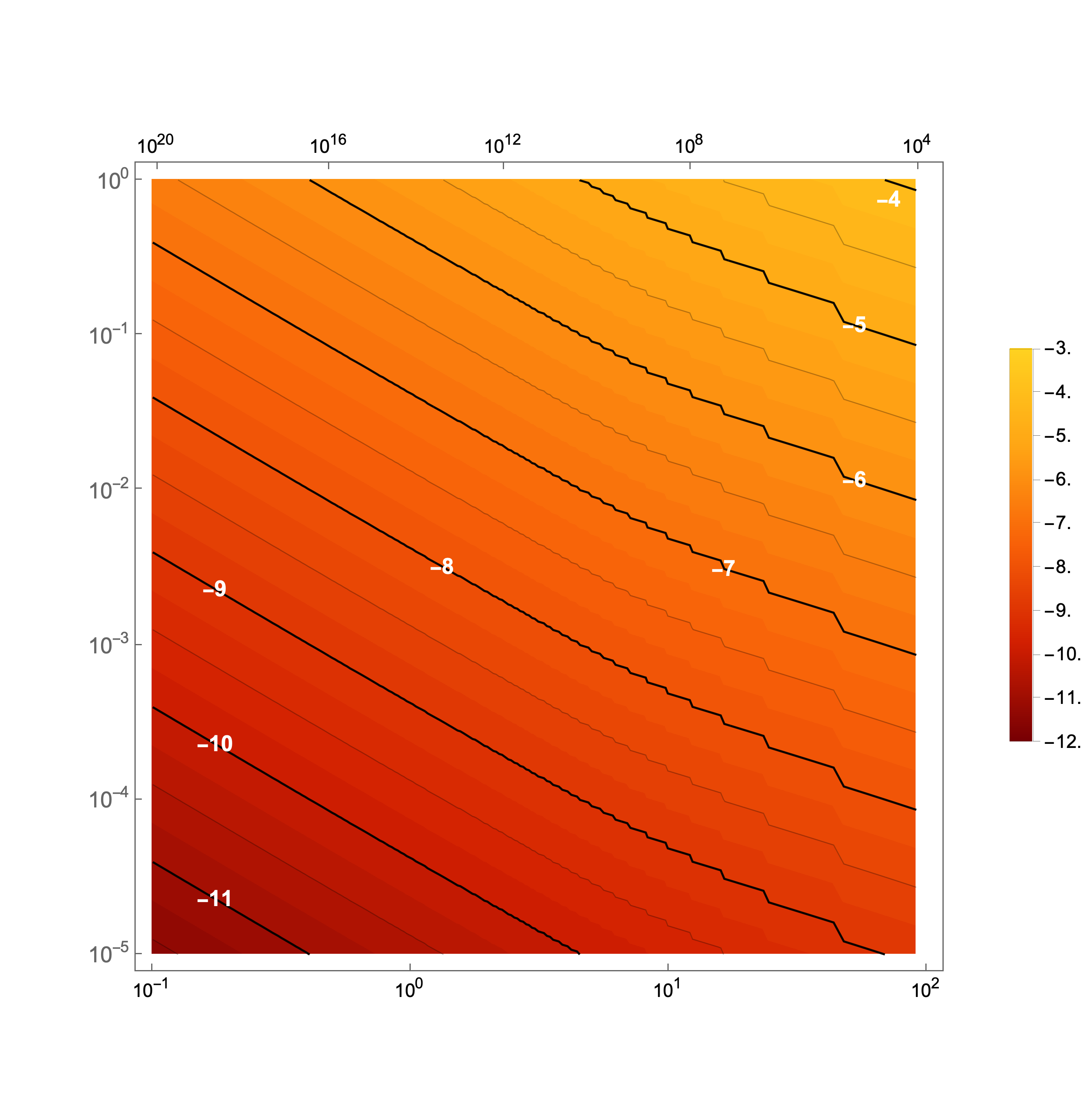}
    \put(0,48){$\bar{\nu}$}
    \put(30,3){$\Delta x=L/(2\NF+1)$}
    \put(85,72){${\log_{10}\varphi_{\max}}$}
    \put(25,92){$(2\NF+1)^3(\NH+1)^3$}
    \put(45,55){$\times$}
    \end{overpic}

    \caption{\textbf{Provably resolved plasma nonlinearities for plane-wave initial data.}
    Shown in the above contour plot is nonlinearity measure $\varphi_{\max}$ from Theorem~\ref{theorem} for which we can rigorously prove exponential convergence of the quantum algorithm. This places a simple analytic bound on the strength of nonlinearities that can be resolved. Further analysis is needed to determine if stronger nonlinearities can also be simulated.
    We set $\vn_0=(1,0,0)$, $\tau=1$, $L=10^3$, and $\NH=\lceil \NF/10\rceil$. The point $(\Delta x,\bar{\nu})=(\pi,10^{-2})$, on which $\varphi_{\max}=7.12\times 10^{-8}$, is indicated by the $\times$.} 
\label{fig:Fourier_plot}
\end{figure}

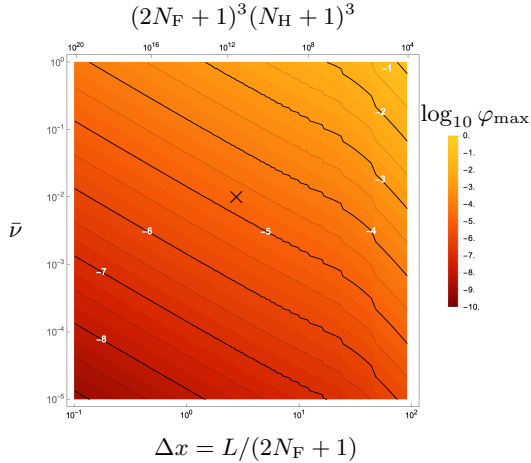
\begin{figure}[t]
    \centering    
    \begin{overpic}[width=0.35\textwidth]{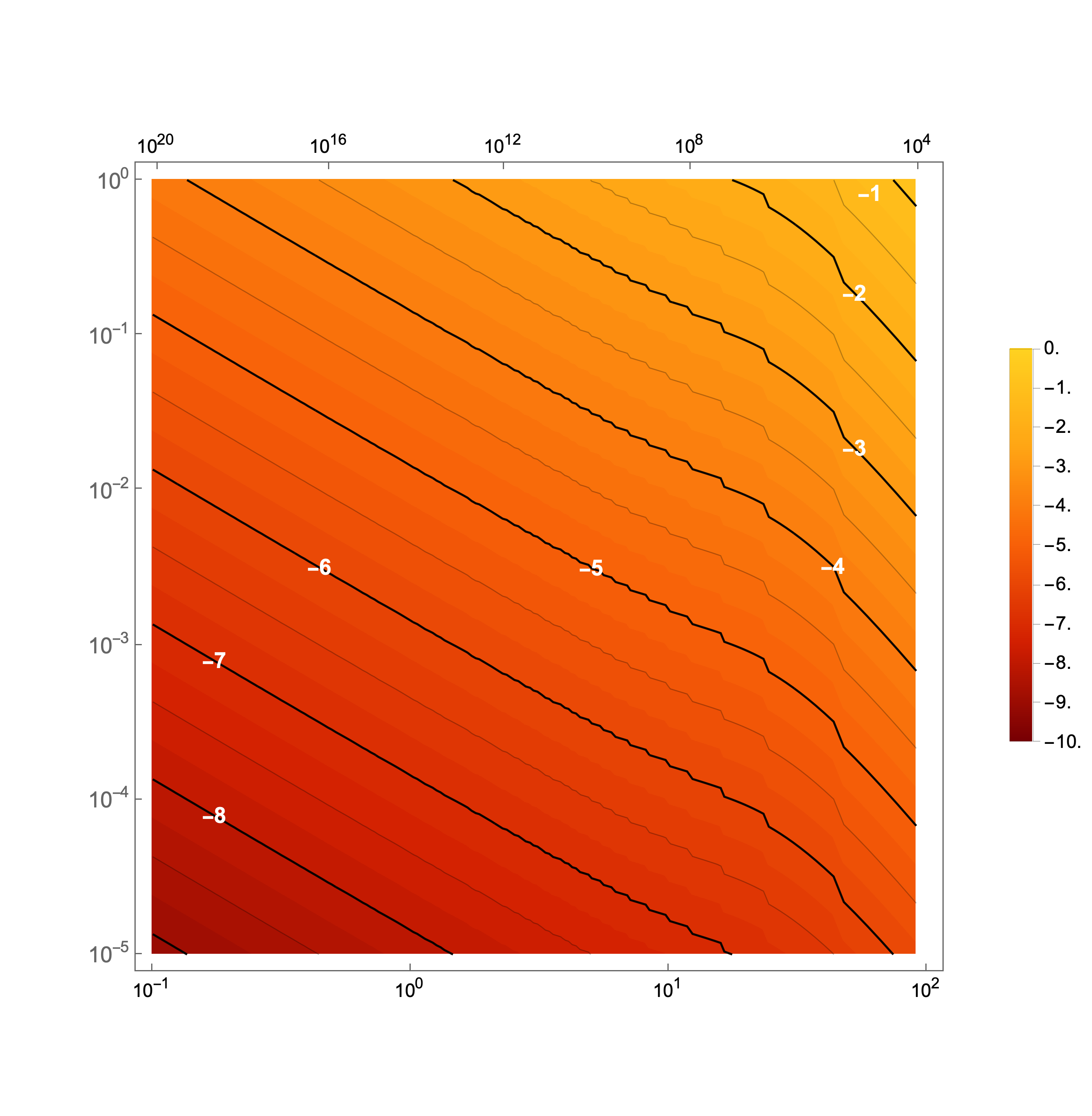}
    \put(0,48){$\bar{\nu}$}
    \put(30,3){$\Delta x=L/(2\NF+1)$}
    \put(85,72){${\small \log_{10}\varphi_{\max}}$}
    \put(25,92){$(2\NF+1)^3(\NH+1)^3$}
    \put(45,55){$\times$}
    \end{overpic}

    \caption{\textbf{Provably resolved plasma nonlinearities for Gaussian initial data.}
    Shown in the above contour plot is nonlinearity measure $\varphi_{\max}$ corresponding to the Gaussian initial data \eqref{eq:g_rad_intro} for which we can rigorously prove exponential convergence of the quantum algorithm.
    We set $\sigma=10$, $\tau=1$, $L=10^3$, and $\NH=\lceil \NF/10\rceil$. The point $(\Delta x,\bar{\nu})=(\pi,10^{-2})$, on which $\varphi_{\max}=2.08\times 10^{-5}$, is indicated by the $\times$.
    See Section~\ref{subsec:initial_data} for details. }
\label{fig:Gaussian_plot}
\end{figure}

\subsubsection{Quantum speedup relative to classical spectral methods}

What quantum speedup is entailed by Theorem~\ref{theorem}, using a Fourier-Hermite classical solver as a benchmark? Classically, the dominant complexity of a single step of the computation involves evaluating the nonlinear term. Naively, the cost would scale as $O(\NH^3 \NF^6)$, but via the Fast Fourier Transform this can be reduced to $\tilde{O}(\NH^3 \NF^3)$, where the tilde indicates that we suppress logarithmic factors. Since the norm of the linear term scales as $O(\NH^{1/2} \NF)$ (proof in Sec.~\ref{sec:normestimates}), for stability reasons the timestep of an explicit time-stepping method must scale with the inverse of that norm. Specifically, a $p$-order method requires $O(T \NH^{1/2} \NF/\epsilon^{1/p})$ timesteps to simulate up to time $T$, leading to an overall cost in terms of elementary operations
\begin{align}
\mathcal{C}_{\mathrm{op}} = \tilde{O}(T \NH^{7/2} \NF^4/\epsilon^{1/p}), 
\end{align}
and a memory cost \begin{align}
    \C_{\mathrm{mem}} = O(\NH^3 \NF^3).
\end{align}

Hence, the quantum algorithm has an exponential saving in memory cost,
and using Eq.~\eqref{eq:overallcostintro} the improvement in operations is given by
\begin{equation}
    \frac{\mathcal{C}_{\mathrm{op}}}{\mathcal{Q}_{\mathrm{op}}} = \tilde{O}(T \NH^{3} \NF^3\epsilon^{1-1/p}).
\end{equation}
Formally, the quantum algorithm displays a $7$th order $(\NH^{7/2} \rightarrow \NH^{1/2}$) and a quartic  ($\NF^4\rightarrow \NF$) improvement in the scaling of the number of operations with $\NH$ and $\NF$, respectively, but a worse scaling with $\epsilon$. It also has an exponential speedup in~$T$, but the significance of this speedup in the present model must be interpreted with care, since the perturbation decays to a scale $\epsilon$ in constant time.

In terms of Sobolev regularity parameters~\cite{adams2003sobolev} $s_{\m{F}}$ and $s_{\m{H}}$ one typically has $\NF = O(\epsilon^{-1/s_{\m{F}}})$ and $\NH = O(\epsilon^{-1/s_{\m{H}}})$, hence
\begin{equation}
    \frac{\mathcal{C}_{\mathrm{op}}}{\mathcal{Q}_{\mathrm{op}}} = \tilde{O}(T \epsilon^{1 -3/s_{\m{F}}-3/s_{\m{H}}-1/p}).
\end{equation}
If we assume high spatial regularity we can take \mbox{$s_{\m{F}}\rightarrow \infty$}. The Hermite regularity is more likely the bottleneck asymptotically. For example, with $s_{\m{H}} = 1$, one finds
\begin{align}
    \mathcal{Q}_{\mathrm{op}} = \tilde{O}(\epsilon^{-3/2}).
\end{align}
Taking the impractical asymptotic scenario $p = \infty$ for the classical method, one obtains
\begin{align}
\mathcal{C}_{\mathrm{op}} = \tilde{O}(\epsilon^{-7/2}) = \tilde{O}(\mathcal{Q}^{7/3}_{\mathrm{op}}).
\end{align}
For the more realistic case of fourth-order methods, e.g. Runge--Kutta, we have $p=4$ and one obtains
\begin{align}
\mathcal{C}_{\mathrm{op}} = \tilde{O}(\epsilon^{-15/4}) = \tilde{O}(\mathcal{Q}^{5/2}_{\mathrm{op}}).
\end{align}
Both cases give a \emph{superquadratic speedup} relative to the classical spectral method. 

The analysis applies to a simplified plasma model. Establishing this superquadratic advantage in a nonlinear regime is a stepping stone towards meaningful impact on plasma simulations. For more complex versions of the plasma problem, especially regimes with stronger phase mixing, filamentation, or turbulent phase-space cascades, where the effective Hermite regularity can be lower than in the smooth weakly nonlinear benchmark considered here, the corresponding Sobolev regularity parameter may be even smaller, e.g. $\soh \in (0,1)$. This could further increase the asymptotic separation between quantum and classical methods, provided that block-encoding and information-extraction costs are not adversely affected. 

\subsubsection{Main technical contributions}

We complete this overview of results by presenting the main technical novelties that were used in our algorithmic pipeline.

1. \emph{Stability analysis.} How did we address the issue of finding a faithful mapping of the nonlinear problem into a large linear one? First, we move from a problem for the distribution function $f$ to a problem for the perturbation $g$ in Eq.~\eqref{eq:perturbationg}. When we apply the Carleman linear embedding, which involves constructing a hierarchy of the dynamical variable and its tensor powers, $g$ is a more appropriate expansion point from the linear case than $f$ is. Secondly, we leverage the stability theory developed in Ref.~\cite{jennings2025quantum}. In that work, convergence criteria for the linear embedding are given whenever a Lyapunov observable (a functional that monotonically decreases along solutions) can be constructed for the linearized problem. In Ref.~\cite{jennings2025quantum} the Lyapunov $R$-number criterion is just an abstract condition, however in the context of the present plasma setting we show that it becomes a natural statement of the thermodynamic monotonicity of the plasma free energy
\small
\begin{align}
   \mathcal{F}[f,\phi] \coloneqq S(f\| f_0) + \frac{1}{2} \int_{\TL^3}  \big(\lvert\nabla \phi\rvert^2+\tau\phi^2\big)\, \mathrm{d}^3x + O(g^3)
    \label{eq:freenergycontinuum}
\end{align}
\normalsize
as a Lyapunov observable. The first term is the Kullback–Leibler relative entropy, a natural measure of the non-equilibrium ion free energy, relative to a Maxwell-Boltzmann equilibrium, and the second term is the screened electrostatic field energy. Leveraging results from Ref.~\cite{jennings2025quantum}, we show that if the free energy~\eqref{eq:freenergycontinuum} of the perturbation at the initial time is small enough, then a faithful linear embedding can be constructed with only a small $\Theta(\log(1/\epsilon))$ overhead in the qubit count. The condition in Eq.~\eqref{eq:epsmax_Fourier_mode_explicitintro} is just a special case of this broader construction for the case of a Fourier mode initial condition. 

2. \emph{Hierarchical data loading.} How did we address the issue of loading dense dynamical data? 
The main technical result we provide is a block-encoding theorem for a class of dense rectangular matrices that arise from
quadratic interactions.  Although we apply it to the Vlasov--Poisson nonlinearity, the result is not specific to that setting and is of more general use. In particular, in Sec.~\ref{app:block-encoding-rectangular} we show how a pure Coulomb kernel, the Coulomb plasma interaction term, and the present screened plasma nonlinearity can all be optimally block-encoded as special cases of the general result below.

The abstract object is any dense rectangular matrix with one output spatial index
and two input spatial indices,
\begin{align}
    \m K_{\v j,\v\alpha \mid \v j',\v\alpha';\v j'',\v\alpha''}
    =
    \sigma_{\v j,\v j''}
    \sum_{a=1}^3
    K^a_{\v j,\v j'}
    W^a_{\v\alpha\mid\v\alpha',\v\alpha''}.
    \label{eq:Kmatrixintro}
\end{align}
Here \(\v j,\v j',\v j''\) lie on a three-dimensional grid of
\(N_x^3\) points on the torus, and the matrix $\m{K}$ can be interpreted as describing a general 2--point interaction operator with additional degrees of freedom, specified by $W^a$. The vertical bar is used for the rectangular matrices to separate row and column indices for the sake of clarity.

In our setting $\m{K}$ is the nonlinear plasma interaction, the terms $W^a$ describe the velocity degrees of freedom (in terms of Hermite modes), while the spatial kernels $\sigma$ and $K^a$ describe the electrostatic field in a Lyapunov-transformed frame. In order to block-encode the nonlinear term we do not discretize in real-space, but instead do Nyquist sampling~\cite{shannon1949communication, nyquist1928certain} of continuum kernels at  lattice points, which in turn determines the Fourier data via a discrete Fourier transform. The kernels \(\sigma\) and \(K^a\) are dense, but are
assumed to decay with the spatial separation of their arguments.  The tensor
\(W^a\) contains the non-spatial degrees of freedom and is assumed to be
sparse. The task at hand is to construct a unitary that encodes $\m K$ as one of its blocks, with a rescaling constant $\alpha_{\m K}$ as small as possible, since the latter typically enters linearly in the cost of quantum algorithms. 

The difficulty is that a direct sparse block-encoding treats the spatial
kernels as dense matrices and therefore produces a normalization that can
grow like a large power of \(N_x\).  The hierarchical block-encoding avoids
this by decomposing each spatial kernel into \(O(\log N_x)\) pieces according
to distance scale.
Nearby interactions are kept in a genuinely sparse
``adjacent'' matrix, while far-away interactions are grouped into larger
blocks whose entries are small because of kernel decay. We provide a cleaner derivation of hierarchical block-encodings than the one presented in Ref.~\cite{nguyen2022block}, extended to three-dimensional settings and the non-square generalized kernel interactions in Eq.~\eqref{eq:Kmatrixintro}. We obtain the following general result.
\begin{theorem}[Hierarchical block-encoding for two polynomially decaying kernels -- Informal]
\label{thm:informal-two-kernel-hbe}
Let \(N_x=2^{\ell_{\max}}\), and consider a rectangular matrix $
    \m K
    \in
    \mathbb C^{N_x^3|\mathcal A|
    \times
    N_x^6|\mathcal A'||\mathcal A''|}$
with entries of the form~\eqref{eq:Kmatrixintro}. Assume that the non-spatial tensors \(W^a\) have \(N_x\)-independent max norm,
row sparsity, and column sparsity.  Assume also that the two spatial kernels
are upper-bounded by polynomial off-diagonal decay,
\begin{align}
    \small
    |\sigma_{\v j,\v k}|
    \le
    \frac{C_\sigma}{
    (1+\|\v j-\v k\|)^{p_\sigma}}, \; \; 
    |\m K^a_{\v j,\v k}|
    \le
    \frac{C_{K,a}}{
    (1+\|\v j-\v k\|)^{p_{K,a}}} .
\end{align}
for constants $C_\sigma, C_{K,a}$. One can then construct hierarchical decompositions
\[
    \sigma = \sum_{\lambda=1}^{\ell_{\max}} \sigma^{(\lambda)},
    \qquad
    K^a = \sum_{\ell=1}^{\ell_{\max}} K^{a,(\ell)} .
\]
In each decomposition, the  \(\ell,\lambda=1\) components contain the same-point and
nearest-neighbor points, while the components \(m=\lambda \geq 2\) or
\(m=\ell \geq 2\) group nonlocal interactions between points whose separation is of
order $
    2^{\ell_{\max}-m}$
in the \(L^\infty(\mathbb{N})\) norm, up to constant factors.

Assume access to unitaries enumerating the nonzero entries in $\sigma^{(\lambda)}$, $K^{a, (\ell)}$, $W^a$, both along columns and along rows; unitaries returning \(b\)-bit approximations to
    \(\sigma_{\v j,\v k}\), \(K^a_{\v j,\v k}\), and
    \(W^a_{\v\alpha\mid\v\alpha',\v\alpha''}\) given the indices; and a unitary routine that prepares a superposition over the hierarchy terms
\((a,\lambda,\ell)\), with amplitudes proportional to
the product of the largest entry of \(\sigma^{(\lambda)}\), the largest entry
of \(K^{a,(\ell)}\), and the relevant square-root sparsity factor. Then \(\m K\) admits a rectangular block-encoding whose normalization satisfies
\begin{align}
    \alpha_{\m K}
    =&\;
    \widetilde O\!\left(
    N_x^{r_1} + N_x^{r_2} + N_x^{r_3}
    \right), \nonumber \\
    r_a
    =&\;
    \max\left\{
        0,\,
        \frac32-p_\sigma,\,
        \frac32-p_{K,a},\,
        \frac92-p_\sigma-p_{K,a}
    \right\}. 
\end{align}
\end{theorem}
The formal statement and proof are provided in Appendix~\ref{app:HBEs}.

The result shows that sufficiently rapid spatial decay offsets the formal density of the matrix.  In the symmetric case \(p_\sigma=p_{K,a}=p\), one obtains
\[
    \alpha_{\m K}
    =
    \begin{cases}
    O(1), & p>9/4,\\
    O(\operatorname{polylog} N_x), & p=9/4,\\
    \widetilde O(N_x^{9/2-2p}), & p<9/4.
    \end{cases}
\]
The power $p=9/4$ is a threshold for the 2-kernel interaction -- in the sense of separating $\mathrm{poly}(N_x)$ scaling of $\alpha_{\m{K}}$ from $O(1)$ scaling. In Appendix~\ref{app:block-encoding-rectangular} we show that if $\sigma_{\v{j},\v{j}''}= \delta_{\v{j},\v{j}''}$ then the resulting electric field is obtained via a convolution (e.g. such as appears in the nonlinear Vlasov-Poisson equation), and this has the threshold decay $p = 3/2$. Thus, for the Coulomb interaction in Vlasov-Poisson we can construct an $\tilde{O}(1)$ block-encoding of the non-linear interaction. Finally, we can also specialize the above result to the square kernel of the form $C/(1+\|\v{j}-\v{j}'\|)^p$, and recover a threshold decay of $p=3$ found in Ref.~\cite{nguyen2022block}. In particular, the block-encoding of Coulomb matrix has scale factor $\tilde{O}(N_x)$, which is again optimal~\cite{babbush2019quantumsublinear}.

3. \emph{Information extraction.} Our algorithmic pipeline outputs an approximate encoding of the plasma data, and tensor powers of the plasma, over time in the form a quantum state. From this `history state' we show how to estimate the space-time averaged non-equilibrium kinetic energy $\langle \mc{K} \rangle$ of the plasma system. Instead of a uniform sampling over time, we use an `integrator' method at the level of the Carleman embedding that converges super-exponentially in the order of the approximate integrator. Notably, within the Carleman embedding scheme this method leverages non-linear data encoded in the history state to improve the accuracy of the linear observable estimation.

\subsection{Overview of algorithm}

At a high level, our approach is the following:

\textbf{Step 1: Truncated system.} The core equations are~\eqref{eq:gnonlinearintro}-\eqref{eq:phi_g_single_speciesintro}. After truncating at $2\NF+1$ Fourier modes and $\NH+1$ Hermite modes per direction, we obtain a quadratic ordinary differential equation (ODE) of the form
\begin{equation}\label{eq:quadraticODEintro}
\dot{\m{u}}=\m{F}_1\m{u}+\m{F}_2(\m{u} \otimes \m{u}), \quad \m{u}(0)=\m{u}_0,
\end{equation}
where $\m{u}_0$ encodes the coefficients of the Fourier-Hermite expansion of $g(\v{x}, \v{v},t=0)$. The weighted square integrability of $g(\vx,\vv,0)$ implies 
\begin{align}
\label{eq:u0bound}
    \lVert\m{u}_0\rVert_2=O(1);
\end{align}  
see Appendix~\ref{app:regularity} for the precise statement. We assume that $\lVert\m{u}_0\rVert_2$ is either known or upper-bounded. 

\textbf{Step 2: Lyapunov transformed system.} An explicit, analytical construction is given of a symmetric positive definite matrix $\m{P}$, such that the function
\begin{align}
\label{eq:discretefreeenergy}
    \mathcal{F}[\m{u}]:= \frac{1}{2}\m{u}^\dag \m{P} \m{u}
\end{align}
is strictly decreasing along any nontrivial solution of the linearized equations. This is the free energy in Eq.~\eqref{eq:freenergycontinuum} for perturbations with modes in the truncation window, see Appendix~\ref{subsec:quadratic_Lyapunov}. 

Given a constant $R_{\mathrm{thr}}<1$, we introduce a corresponding \emph{Lyapunov transform} 
\begin{align}
\label{eq:ubar}
\bar{\m{u}} \coloneqq \frac{1+R_{\mathrm{thr}}}{2 \|\tilde{\m{u}}_0\|_2} \tilde{\m{u}}
\end{align} 
obeying the quadratic ODE system 
\begin{equation}
\label{eq:rescaledsystemintro}
\dot{\bar{\m{u}}}=\bar{\m{F}}_1\bar{\m{u}}+\bar{\m{F}}_2(\bar{\m{u}} \otimes \bar{\m{u}}), \quad \bar{\m{u}}(0)= \bar{\m{u}}_0.
\end{equation}
The transform~\eqref{eq:ubar} can be seen as a generalization of the rescaling considered in prior Carleman approaches, e.g., in Ref.~\cite{liu2021efficient, krovi2022improved, costa2023further}. With this Lyapunov transform, we find that
\begin{enumerate}
    \item The logarithmic norm of the linear component equals to the dimensionless dissipation rate, 
    $$\mu_2(\bar{\m{F}}_1)= - \bar{\nu},$$
    where $\mu_2(\m{X})$ is the largest eigenvalue of $\tfrac12(\m{X}+\m{X}^{\dag})$. This means the transformed system is dissipative. 
    \item The initial norm obeys 
    \begin{align}
    \label{eq:rescalednorm}
       0<  \| \bar{\m{u}}_0\| \leq \tfrac12(1+R_{\mathrm{thr}}) <1.
    \end{align}
\end{enumerate}
The first property allows us to get a priori estimates of the linear embedding error, discussed in Step~3, and it also underpins the complexity analysis of the quantum algorithm, discussed in Step~5. The second property is also required for the error analysis, as well as to guarantee efficient information extraction in the quantum algorithm. The fact that $\lVert\m{u}_0\rVert_2=O(1)$ and $\| \m{P}\|_2 = O(1)$ in our construction allows us to avoid complexity overheads, side-stepping a potential source of inefficiency.

\textbf{Step 3: Carleman linear embedding and a priori convergence guarantees.} In the rescaled coordinates, we obtain a hierarchy of ODE equations for the time-derivative of the variables $$\bar{\m{z}} = (\bar{\m{u}}, \bar{\m{u}}^{\otimes 2}, \dots, \bar{\m{u}}^{\otimes N_\mathrm{C}})^{\mathsf{T}}.$$ The equation for $\m{u}^{\otimes N_\mathrm{C}}$ depends on $\m{u}^{\otimes (\NC+1)}$, but we close the system by setting the latter to zero. This gives rise to an $O\left((\NH \NF)^{3N_\mathrm{C}}\right)$-dimensional truncated \emph{Carleman ODE system}
\begin{align}
\label{eq:Carlemansystem}
   \frac{\mathrm{d}}{\mathrm{d}t}\bar{\m{z}} = \bar{\m{A}} \bar{\m{z}}, \quad \bar{\m{z}}(0)= [ \bar{\m{u}}_0, \bar{\m{u}}^{\otimes 2}_0, \dots, \bar{\m{u}}^{\otimes N_\mathrm{C}}_0]^{\mathsf{T}},
\end{align}
for an appropriate \emph{Carleman matrix} $\bar{\m{A}}$. The linear embedding error is defined as
\begin{align}
    \epsilon_{\mathrm{C}} := \max_{t \geq 0} \| \bar{\m{z}}(t) - \bar{\m{w}}(t)\|,
\end{align}
where 
\begin{align}
\label{eq:wbar}
\m{w}(t) = [\bar{\m{u}}(t), \bar{\m{u}}(t)^{\otimes 2}, \dots, \bar{\m{u}}(t)^{\otimes \NC}] ^{\mathsf{T} }  
\end{align}
and $\bar{\m{u}}(t)$ is the solution to Eq.~\eqref{eq:rescaledsystemintro}. With our constructions and leveraging the analysis of Ref.~\cite{jennings2025quantum}, we can prove exponential decay of the error with the truncation parameter $\NC$
\begin{align}
    \NC = O(\log(1/\epsC)),
\end{align}
as long as the  Lyapunov $R$-number condition is satisfied:
\begin{align}
\label{eq:RPintro}
    R_{\m{P}}:= \frac{1}{-\mu(\bar{\m{F}}_1)} \| \bar{\m{F}}_2\|_2 \|{\bar{\m{u}}}_0\|_2 \leq R_{\mathrm{thr}}.
\end{align}

For the plasma system, it turns out that this takes the more physically intuitive formulation in terms of the free energy $\mathcal{F}$ in Eq.~\eqref{eq:discretefreeenergy}. The convergence condition can be phrased as a free energy bound at \mbox{$t=0$}:
\begin{align}
\label{eq:freenergycondition}
    \mathcal{F}(t=0) \leq \mathcal{F}_{\max},
\end{align}
where the threshold value is
\begin{equation}
\label{eq:Fmax_threshold}
\mathcal{F}_{\max}
=
\frac{\bar{\nu}^2 L^3}{6}
\frac{1}{
\big(
\frac{1}{1+\tau}
+
\frac{13L^2}{2\pi^2}\NF
\big)
\big(1+\sqrt{3\NH}\big)^2
}.
\end{equation}
Note that $\mathcal{F}_{\max}$ decreases as $(\NF\NH)^{-1}$. However we recall that the above is only a \emph{sufficient}, but not necessary, condition for convergence. We test these conditions for Fourier mode initial data in Eq.~\eqref{eq:initialconditiontheorem} and the Gaussian perturbation in Eq.~\eqref{eq:g_rad_intro}, obtaining the results we presented in Figs.~\ref{fig:Fourier_plot}-\ref{fig:Gaussian_plot}.

\textbf{Step 4: Quantum solver.} We run a quantum ODE solver on the Carleman ODE system~\eqref{eq:Carlemansystem}. In particular, we use the solver from Ref.~\cite{berry2022quantum} with no `idling' steps and the additive error scheme from Ref.~\cite{jennings2023cost}. To encode the required dense interaction matrices, we leverage the hierarchical block-encoding techniques of Theorem~\ref{thm:informal-two-kernel-hbe}.

This outputs an $\epsilon$-approximation to a quantum state encoding an $O(\epsilon)$-approximation to the solution and its tensor powers up to $\NC$ along the whole trajectory:
\begin{align}
\label{eq:powerhistorystate}
\ket{\chi} = \frac{1}{\mathcal{N}_\chi} \sum_{s=0}^{M-1} \|\bar{\m{w}}(s \Delta t)\| \ket{\bar{\m{w}}(s \Delta t)} \ket{s},
\end{align}
where $\ket{\bar{\m{w}}(t)}$ is a quantum encoding of $\m{w}(t)$ in Eq.~\eqref{eq:wbar}. This is a \emph{Carleman history state}. We prove that outputting this state has a gate cost
\begin{align}
\label{eq:complexitylinearsolver}
\tilde{O}\left(\NF\NH^{1/2}\operatorname{polylog} \left(\frac{T}{ \epsilon}\right) \operatorname{polylog}\left(\frac{1}{\epsilon}\right)\right),
\end{align}
and a qubit cost
\begin{align}
    \tilde{O}\bigg(\log(\NF \NH^{1/2} T) \log\bigg(\frac{1}{\epsilon}\bigg)\bigg).
\end{align}
The logarithmic scaling in $T$ follows from the stability properties of the nonlinear system~\eqref{eq:rescaledsystemintro}, and is in line with recent fast-forwarding results~\cite{an2026fast}. The scaling with $\NF$, $\NH$ follows from our data-encoding constructions.

\textbf{Step 5: Information extraction.} To solve the end-to-end Problem~\ref{problem}, we need to efficiently estimate a spacetime averaged kinetic moment of the perturbation in Eq.~\eqref{eq:kineticperturbation}, up to $\epsilon$ additive error. At a high level, we do as follows. 
We construct a high-order quadrature approximation to the time average directly from the Carleman history state~\eqref{eq:powerhistorystate}. Rather than estimating the observable only at the stored time points, we use the Carleman generator $\bar{\m{A}}$ to approximate the integral over each interval $[s\Delta t,(s+1)\Delta t]$ by applying the truncated operator
\begin{equation}
\m{R}_k(\bar{\m{A}} h):=\sum_{q=0}^{k}\frac{(\bar{\m{A}} h)^q}{(q+1)!}.
\end{equation}
This converts the desired time average into an overlap between the Carleman history state, including all powers of the data, and a state proportional to $\m{R}_k(\bar{\m{A}} h)^\dagger\bar\ell_K$, where $\bar\ell_K$ is the kinetic-energy functional in the Lyapunov-transformed Carleman space. The overlap is estimated by amplitude estimation~\cite{brassard2000quantum}, while the history-state normalization is estimated using the normalization-estimation procedure of Ref.~\cite{dalzell2024shortcut}. The Taylor quadrature error decreases factorially in $k$, so choosing $k=\Theta(\log(1/\epsilon)/\log\log(1/\epsilon))$ suffices. The resulting information-extraction cost is $\widetilde O(1/\epsilon)$, giving the overall dependence on $(\NF,\NH)$ and $\epsilon$ stated in Theorem~\ref{theorem}.

\section{Detailed analysis}
\label{sec:detailed}

We now present the detailed analysis of the plasma system and quantum algorithm. We shall assume that we have an electron-ion plasma in $3+3$ phase space dimensions. The extension to multi-species plasmas can also be done, however it turns out that the stability analysis can become involved when the collisional rates of the multiple species vary substantially.

    \subsection{Details of the plasma model}

We derive the weakly nonlinear plasma model, nondimensionalize it, and put the resulting equations in Fourier--Hermite form. We then introduce the symmetrizing change of variables later used in the stability analysis.
    
    \label{sec:detailedplasma}

    \subsubsection{Weakly nonlinear electron-ion plasma}

To demonstrate our quantum algorithm for solving a weakly nonlinear Vlasov-Poisson system, we consider an electron-ion plasma where the ions are treated as kinetic and the electrons are assumed to be adiabatic. This type of kinetic problem is common in plasma physics, such as the investigation of ion-temperature-gradient turbulence where ions are treated as gyrokinetic and the electrons are often assumed to be adiabatic, or the electron-temperature-gradient turbulence where the electrons are gyrokinetic and the ions are often adiabatic.

The thermalized electrons are assumed to be in a local Maxwell-Boltzmann equilibrium,
\begin{equation}
\label{eq:thermal_ansatz}
n_e(\v{x},t)=n_{e,0}\exp\bigg(-\frac{q_{e}\phi(\v{x},t)}{k_{\mathrm{B}}T_{e}} \bigg),
\end{equation}
where $q_e$, $T_e$, $n_{e,0}$ are the charge, equilibrium temperature, and equilibrium density of the electrons, respectively, and $k_{\mathrm{B}}$ is Boltzmann's constant. We will study the \emph{weakly nonlinear} regime, where \begin{align}
\label{eq:weaklynonlinear}
\varphi(\vx,t) \coloneqq {\lvert q_e\rvert \phi(\vx,t) \over k_B T_e}\ll 1. 
\end{align}

The kinetic species evolves according to the Vlasov equation for $f(\v{x},\v{v},t)$ in the $6$-dimensional phase space, and we allow for weak collisions modeled by a collision operator $\Omega$,  
\begin{equation}
\partial_t f =-\v{v}\cdot \nabla_x f-\frac{\lvert q_e\rvert }{m}\v{E}\cdot\nabla_v f+\Omega(f),
\label{eq:kinetic}
\end{equation}
with $|q_e|$ the ion charge.

For simplicity, we take the spatial domain to be
 a cubic box of length $L$ with periodic boundary conditions, denoted by $\mathbb{T}_L^3$. The collision operator 
$\Omega(f)$ is taken to have the form \cite{bhatagnar1954model}, 
\begin{equation}\label{eq:BGK}
    \Omega(f) = - \nu (f - f_{0}).
\end{equation}
This is a simplified model, and is primarily for technical convenience. In particular, it ensures stability of the plasma system, however, recent results~\cite{jennings2025quantum} have also shown that dynamics with conserved quantities can also be handled within the Carleman framework. It would be of interest to also explore other collisional operators with such properties. 

The kinetic species is assumed to be described by perturbations around a Maxwellian background:
\begin{equation}
\label{eq:kinetic_ansatz}
f(\v{x},\v{v},t)=f_{0}(\v{v})+h(\v{x},\v{v},t),
\end{equation}
where $h$ is assumed small relative to $f_{0}$.  Substituting~\eqref{eq:kinetic_ansatz} into ~\eqref{eq:kinetic}, the kinetic equation becomes
\begin{align}
\label{eq:kinetich}
\partial_t h=&\; -\v{v}\cdot \nabla_x h -\frac{\lvert q_e\rvert}{m}\v{E}\cdot \nabla_v f_{0} \nonumber \\
&\; -\frac{\lvert q_e\rvert}{m}\v{E}\cdot \nabla_v h - \nu h.
\end{align}

Naturally, we assume the perturbation does not change the total number of particles relative to the background. Set
\begin{equation}
\label{eq:Mh}
M(t)\coloneqq \int_{\TL^3\times \R^3} h(\vx,\vv,t)\,\mathrm{d}^3x\,\mathrm{d}^3 v. 
\end{equation}
If $M(0)=0$, then $M(t)=0$ for $t>0$; see 
Appendix~\ref{subsec:mass_conservation}.

The two species are coupled through the Poisson equation:
\begin{equation}
\begin{split}
\label{eq:Poisson}
\nabla^2_x \phi =&\; \frac{\lvert q_e\rvert}{\epsilon_0}(n_e-n), \\
\end{split}
\end{equation}
where 
\begin{equation}
 n(\v{x},t) \coloneqq \int_{\mathbb{R}^3} f(\v{x},\v{v},t)\,\mathrm{d}^3 v   
\end{equation} and $\epsilon_0$ is the permittivity of free space. We assume background charge neutrality,
\begin{equation}
\label{ass:charge_neutrality}
n_{e,0}=n_{0},
\end{equation}
where $n_0$ is the background ion density from Eq.~\eqref{eq:Maxwellian}.

\subsubsection{Screening}

The presence of the thermalized electrons generates a screening potential seen by the ions.  From Eq.~\eqref{eq:thermal_ansatz} and the assumption~\eqref{eq:weaklynonlinear}, we can expand the Boltzmann distribution of thermalized species: 
\begin{equation}
\label{eq:nsexp}
n_{e}(\v{x},t)
= n_{e,0}+ \delta n_e(\v{x},t) +O\big(\big(\tfrac{q_e\phi}{k_{\mathrm{B}} T_e} \big)^2\big),
\end{equation}
where $
\delta n_e(\v{x},t)\coloneqq -n_{e,0}\frac{q_e\phi(\v{x},t)}{k_{\mathrm{B}}T_e}$. With the expanded Boltzmann response for the thermal species, Poisson's equation (Eq.~\eqref{eq:Poisson}), now takes the form of a Helmholtz equation
\begin{equation}
\label{eq:Helmholtz}
\bigg( \nabla^2_x-\frac{1}{\lambda_{\mathrm{D},e}^2} \bigg)\phi =-\frac{\lvert q_e\rvert}{\epsilon_0} \int_{\R^3} h\,\mathrm{d}^3 v, 
\end{equation}
where $\nabla_x^2$ is the Laplacian on the configuration space $\TL^3$ and we have defined the \emph{Debye length}
\begin{equation}\label{eq:lambdadefns}
\lambda_{\mathrm{D},e}\coloneqq \sqrt{\frac{\epsilon_0 k_{\mathrm{B}}T_e}{n_{e,0}q_e^2}}.
\end{equation}

Eq.~\eqref{eq:Helmholtz} shows that the kinetic particles move in an effective screened potential due to Debye shielding. More specifically, the potential is generated by the Green's function
of the Helmholtz operator $\nabla_x^2-\lambda_{\mathrm{D},e}^{-2}$. On $\R^3$, this is precisely the Yukawa potential
\begin{equation}
\label{eq:Green_function}
G_{\mu}(\vx)
\coloneqq 
\frac{\ee^{-\mu\lvert \vx\rvert}}{4\pi\lvert\vx\rvert};
\end{equation}
with $\mu=1/\lde$ the corresponding Green's function on $\TL^3$ is obtained by convolving \eqref{eq:Green_function} with a periodic Dirac delta.

The weakly nonlinear Vlasov-Poisson system we solve is described by the Vlasov equation for the 6D perturbed distribution function $h$, Eq.~\eqref{eq:kinetich}, together with the Helmholtz equation \eqref{eq:Helmholtz} and initial conditions. 

\subsubsection{Natural plasma units}
\label{subsec:electron_ion_plasma}

Define the ion thermal velocity
\begin{equation}
\vth\coloneqq \sqrt{\frac{\kb T_{\mathrm{ion}}}{m}}.
\end{equation}
We move to natural dimensionless units:
\begin{equation}
\label{eq:natural_units}
\begin{split}
&\bar{\vx}\coloneqq \frac{1}{\ldion}\vx, \quad
\bar{\vv}\coloneqq \frac{1}{\vth}\vv, \quad
\bar{t}\coloneqq \frac{\vth}{\ldion}t, \\
&\bar{f}\coloneqq \frac{\vth^3}{n_0}f, \quad
\bar{\phi}\coloneqq \frac{\lvert q_e\rvert}{\kb T_{\mathrm{ion}}}\phi=\frac{1}{\tau}\varphi,
\end{split}
\end{equation}
where
\begin{equation}
\ldion\coloneqq \sqrt{\frac{\epsilon_0\kb T_{\mathrm{ion}}}{n_{0}q_e^2}}
\end{equation}
is the ion Debye length. Note that \eqref{eq:natural_units} induces a rescaling of the torus length parameter, to the dimensionless quantity
\begin{equation}
\bar{L}\coloneqq \frac{1}{\ldion}L.
\end{equation}
From this point onward, we drop the bars and understand all variables and parameters to be dimensionless. \\

Define
\begin{equation}\label{eq:honf}
g(\v{x},\v{v},t)\coloneqq \frac{h(\v{x},\v{v},t)}{w(\v{v})},
\end{equation}
where $w$ is the Gaussian weight function
\begin{equation}
\label{eq:w_def}
w(\v{v})\coloneqq \frac{1}{(2\pi)^{3/2}}\exp\bigg(-\frac12\lvert \v{v}\rvert^2\bigg).
\end{equation}
Note that $w$ coincides with $f_0$ under the change of variables \eqref{eq:natural_units}.
Moreover, let
\begin{equation}
\tau\coloneqq \frac{T_{\mathrm{ion}}}{T_e},\quad \bar{\nu}\coloneqq \frac{\nu\ldion}{\vth}.
\end{equation}
Then the kinetic equation~\eqref{eq:kinetich}, written in natural units and in terms of the scaled variable \eqref{eq:honf}, reads
\small
\begin{equation}\label{eq:gnonlinear}
g_t
=
-\vv\cdot\nabla_x g
-\vv\cdot\nabla_x\phi
-(\vv\cdot \nabla_x \phi)g
+\nabla_x\phi\cdot\nabla_v g
-\bar{\nu} g.
\end{equation}
\normalsize
It is coupled to the Helmholtz equation~\eqref{eq:Helmholtz}, which now reads
\begin{equation}
\label{eq:phi_g_single_species}
\big(\nabla_x^2-\tau\big)\phi
=
-\int_{\mathbb{R}^3}gw\,\mathrm d^3v.
\end{equation}

\subsubsection{Fourier-Hermite expansion}

For each Fourier mode $\vn\in \mathbb Z^3$, define the corresponding wavevector
\begin{equation}
\v{\xi}_{\vn}\coloneqq \frac{2\pi}{L}\vn.
\end{equation}

We Fourier transform the spatial variable $\vx$. For each Fourier mode
$\vn\in \mathbb Z^3$, we write
\begin{equation}
\hat g_{\vn}(\v v,t)
=
\frac{1}{L^{3/2}}
\int_{\mathbb T_L^3}
g(\vx,\vv,t)\ee^{-i\v\xi_{\vn}\cdot \v x}\,\mathrm d^3x.
\end{equation}

Furthermore, we expand the velocity dependence in the normalized Hermite tensor basis, orthonormal with respect to the Gaussian weight $w$:
\begin{equation}
\label{eq:Hermite_g}
\hat g_{\vn}(\vv,t)
=
\sum_{\v\alpha\in\mathbb N_0^3}
g_{\vn,\v\alpha}(t){H}_{\v\alpha}(\vv),
\end{equation}
where 
\begin{align}
\label{eq:hermite_basis}
{H}_\alpha(u)\coloneqq &\; \frac{\mathrm{He}_\alpha(u)}{\sqrt{\alpha!}}
\quad (\alpha\in \mathbb{N}_0), \\
{H}_{\v{\alpha}}(\v{u})\coloneqq &\;
\prod_{a=1}^3 {H}_{\alpha_a}(u_a)
\quad (\v{\alpha}\in \mathbb{N}_0^3),
\end{align}
with $\mathrm{He}_\alpha$ the probabilists' Hermite polynomials,
\begin{equation}
\mathrm{He}_{\alpha}(u)\coloneqq (-1)^{\alpha}\ee^{u^2/2}\frac{\mathrm{d}^\alpha}{\mathrm{d}u^\alpha}\ee^{-u^2/2} \quad (\alpha\in\mathbb{N}_0).
\end{equation}

We obtain (see Appendix~\ref{app:FHbasis})
\small
\begin{align}
\label{eq:g_mode_hermite}
\dot g_{\vn,\v\alpha}
={}&
-i
\sum_{a=1}^3
\xi_{\vn,a}
\big(
\sqrt{\alpha_a+1}g_{\vn,\v\alpha+\v e_a}
+
\sqrt{\alpha_a}g_{\vn,\v\alpha-\v e_a}
\big)
\nonumber\\
&
-i\frac{1}{\kappa_{\vn}}
\Bigg(
\sum_{a=1}^3 \xi_{\vn,a}\delta_{\v\alpha,\v e_a}
\Bigg)
g_{\vn,\v 0}
-\bar{\nu} g_{\vn,\v\alpha} \nonumber \\
& -\frac{i}{L^{3/2}}
\sum_{\substack{\vn',\vn''\in\mathbb{Z}^3\\ \vn'+\vn''=\vn}}
\frac{1}{\kappa_{\vn'}}g_{\vn',\v 0}
\sum_{a=1}^3 \xi_{\vn',a}\sqrt{\alpha_a}g_{\vn'',\v{\alpha}-\v e_a},
\end{align}
\normalsize
with the convention
\begin{equation}
g_{\vn,\v\alpha-\v e_a}=0
\qquad\text{if }\alpha_a=0 \quad (a=1,2,3)
\end{equation}
and
\begin{equation}
\label{eq:kappa_n_def}
\kappa_{\vn}\coloneqq |\v\xi_{\vn}|^2+\tau.
\end{equation}
Equation~(\ref{eq:g_mode_hermite}) defines the central set of nonlinear equations of interest here. However, there are two obstacles to a direct quantum simulation: firstly, the system is unbounded and must be suitably truncated, and secondly the system is nonlinear. We will apply a Carleman embedding, but the stability properties of the resultant linear ODE system require care. 

\subsubsection{Symmetrization}

The stability properties of the linear part of Eq.~\eqref{eq:g_mode_hermite} become transparent after an appropriate coordinate rescaling, which we will refer to in the next section as the Lyapunov transform.

The only asymmetry in Eq.~\eqref{eq:g_mode_hermite} occurs in the coupling
between the density mode $\v\alpha=\v 0$ and the first Hermite shell
$\v\alpha=\v e_a$. We therefore introduce the non-unitary rescaling
\begin{equation}\label{eq:rescaling}
\tilde g_{\vn,\v\alpha}
\coloneqq
\begin{dcases}
\sqrt{\frac{1+\kappa_{\vn}}{\kappa_{\vn}}}g_{\vn,\v 0} & \v\alpha=\v 0\\
g_{\vn,\v\alpha} & \v\alpha\neq \v 0.
\end{dcases}
\end{equation}

In these variables, the equations of motion take the following form. For the density mode $\v\alpha=\v 0$,
\begin{equation}
\label{eq:gtilde_density}
\dot{\tilde g}_{\vn,\v 0}
=
-i\sqrt{\frac{1+\kappa_{\vn}}{\kappa_{\vn}}}
\sum_{a=1}^3 \xi_{\vn,a}\tilde g_{\vn,\v e_a}
-\bar{\nu}\tilde g_{\vn,\v 0}.
\end{equation}
For the first Hermite shell, $\v{\alpha}=\v e_a$, $a=1,2,3$,
\small
\begin{align}
\dot{\tilde g}_{\vn,\v e_a}
=&\;
-i\sqrt{\frac{1+\kappa_{\vn}}{\kappa_{\vn}}}
\xi_{\vn,a}\tilde g_{\vn,\v 0}
\nonumber\\
&\;
-i\sum_{b=1}^3
\xi_{\vn,b}\sqrt{1+\delta_{a,b}}\,
\tilde g_{\vn,\v e_a+\v e_b}
-\bar{\nu} \tilde g_{\vn,\v e_a}
\nonumber\\
&\;
-\frac{i}{L^{3/2}}
\sum_{\substack{\vn',\vn''\in \mathbb{Z}^3\\ \vn'+\vn''=\vn}}
\frac{\sqrt{\kappa_{\vn''}}\,\xi_{\vn',a}}
{\sqrt{\kappa_{\vn'}(1+\kappa_{\vn'})(1+\kappa_{\vn''})}}\,
\tilde g_{\vn',\v 0}\tilde g_{\vn'',\v 0}.
\label{eq:gtilde_firstshell}
\end{align}
\normalsize
For all other Hermite modes,
\small
\begin{align}
\label{eq:gtilde_higher}
\dot{\tilde g}_{\vn,\v\alpha}
=&
-i
\sum_{a=1}^3
\xi_{\vn,a}
\big(
\sqrt{\alpha_a+1}\,\tilde g_{\vn,\v\alpha+\v e_a}
+
\sqrt{\alpha_a}\,\tilde g_{\vn,\v\alpha-\v e_a}
\big)
-\bar{\nu} \tilde g_{\vn,\v\alpha}
\nonumber\\
&
\! \! \!\! \! \!\! \! \! \! \! \!\! \! \! -\frac{i}{L^{3/2}}
\sum_{\substack{\vn',\vn''\in \mathbb{Z}^3\\ \vn'+\vn''=\vn}}
\frac{1}{\sqrt{\kappa_{\vn'}(1+\kappa_{\vn'})}}\,
\tilde g_{\vn',\v 0}
\sum_{a=1}^3 \xi_{\vn',a}\sqrt{\alpha_a}\,
\tilde g_{\vn'',\v{\alpha}-\v e_a}.
\end{align}
\normalsize

The crucial property of this transformation is that the linear component becomes the sum of a skew-Hermitian operator and the diagonal damping term $-\bar{\nu} \m{I}$, and is compatible with the one-dimensional analysis of Ref.~\cite{ameri2023quantum}. The reason for performing this transformation is to allow a simple stability analysis of the embedding Carleman ODE system.

\subsubsection{Electric field in the Lyapunov picture}
The above transformed system defines a nonlinear ODE system in Fourier-Hermite space, with a slightly modified nonlinear term compared to the screened/Yukawa setting. This in turn corresponds to a different spatial profile in real-space that is slightly more complex. In Appendix~\ref{app:BE-of-nonlinear-term} we present an explicit analysis of this Lyapunov-transformed field, and how it distorts the nonlinear interaction term. In particular, we find that the nonlinear term has two decaying kernels appearing in it, denoted by $\phi_L(\v{x})$ and $\sigma_L(\v{x})$. These can be viewed as distortions of the Yukawa potential. In the $L\rightarrow \infty$ limit we show that they take on the simple form
\begin{align}
    \lim_{L\rightarrow \infty} \phi_L (\v{x}) &= \frac{1}{\pi}\int_0^1 \frac{1}{\sqrt{s(1-s)}}  \frac{\ee^{-\mu(s)\lvert\v{r}\rvert }}{4\pi\lvert\v{r}\rvert}\,\mathrm{d}s, \nonumber \\
  \lim_{L\rightarrow \infty} \sigma_L (\v{x}) &=   \delta^{(3)}(\v{r})-\frac{1}{\pi}\int_0^1 \sqrt{\frac{s}{1-s}}\,\frac{\ee^{-\mu(s)\lvert\v{r}\rvert }}{4\pi\lvert\v{r}\rvert}\,\mathrm{d}s. 
\end{align}
Therefore, in this unbounded limit the effect of the Lyapunov transform is to generate a mixture of Yukawa potentials. This introduces some additional technicalities, but the situation is largely the same as the original screened field. Finally, the finite $L$ torus case can be obtained from the unbounded situation by the method of periodization by images, which is also discussed in Appendix~\ref{app:BE-of-nonlinear-term}.

\subsection{Galerkin approximation and Lyapunov stability analysis}

We now truncate the Fourier--Hermite system to a finite Galerkin ODE and use the symmetrization from the previous subsection to identify the Lyapunov structure of the linearized dynamics.

\label{sec:galerkin}
\subsubsection{Symmetrization as a Lyapunov transform}

From now on, we choose Fourier and Hermite cutoffs $\NF,\NH\in \mathbb{N}$, respectively~\cite{canuto2006spectral}. Define
\begin{equation}
\label{eq:Truncation}
\begin{split}
\mc{Z}_{\NF}\coloneqq &\;  \{-\NF,\ldots,\NF\},\\
\mc{N}_{\NH}\coloneqq &\; \{0,\ldots,\NH\}.
\end{split}
\end{equation}
We then work with the product sets $\mc{Z}_{\NF}^{ 3}$ and $\mc{N}_{\NH}^{3}$.

We may write the nonlinear system \eqref{eq:g_mode_hermite} in vector form after truncation as
\begin{equation}\label{eq:quadraticODE}
\dot{\m{u}}=\m{F}_1\m{u}+\m{F}_2(\m{u}\otimes \m{u}),
\end{equation}
where the components of $\m{u}$ are $\m{u}_{\vn,\v{\alpha}}\coloneqq g_{\vn,\v{\alpha}}$, an $N$-dimensional vector of size
\begin{align}
    N = (2\NF+1)^3(\NH+1)^3.
\end{align}

We shall now analyze the properties of the linearized dynamics (where we temporarily set $\m{F}_2=\m{0}$) through the lens of Lyapunov stability theory. The results will be leveraged to extract a priori guarantees for the Carleman linear embedding method, using the results of Ref.~\cite{jennings2025quantum}.

The vector form of the rescaled equations \eqref{eq:gtilde_density}, \eqref{eq:gtilde_firstshell}, \eqref{eq:gtilde_higher} after truncation is
\begin{equation}
\label{eq:quadraticODErescaled}
\dot{\tilde{\m{u}}}
=
\tilde{\m{F}}_1\tilde{\m u}
+
\tilde{\m{F}}_2 (\tilde{\m u} \otimes \tilde{\m u}),
\end{equation}
where the components of $\tilde{\m{u}}$ are
$\tilde{\m{u}}_{\vn,\v{\alpha}}\coloneqq \tilde{g}_{\vn,\v{\alpha}}$.

We can understand the relation between~\eqref{eq:quadraticODE} and \eqref{eq:quadraticODErescaled} via the rescaling \eqref{eq:rescaling} as a transformation defined by the quadratic form
\begin{equation}
\m{P}\coloneqq \bigoplus_{\vn\in\mc{Z}_{\NF}^3} \m{P}_{\vn},
\qquad
\m{P}_{\vn}\coloneqq \m{I}+\frac{1}{\kappa_{\vn}}\Pi_{\v 0},
\end{equation}
where $\Pi_{\v 0}$ is the orthogonal projector onto the $\v \alpha=\v 0$ Hermite mode,
\begin{equation}
(\Pi_{\v 0})_{\v \alpha|\v \alpha'}\coloneqq
\delta_{\v\alpha,\v 0}\delta_{\v \alpha',\v 0}.
\end{equation}
The unique symmetric positive-definite square root of $\m{P}$ is constructed from the unique symmetric positive-definite square root of each $\m{P}_{\vn}$,
\begin{equation}
\begin{split}
\m{P}^{\pm 1/2}
=&\; 
\bigoplus_{\vn\in\mc{Z}_{\NF}^3} \m{P}_{\vn}^{\pm 1/2},
 \\
\m P_{\vn}^{\pm 1/2}
=&\;
\m I+\bigg(
\left(\frac{1+\kappa_{\vn}}{\kappa_{\vn}}\right)^{\pm\frac{1}{2}}-1
\bigg)\Pi_{\v 0}. 
\end{split}
\end{equation}

With this notation, the rescaling~\eqref{eq:rescaling} reads
\begin{align}
\label{eq:utilde}
\tilde{\m{u}}
&\coloneqq \m{P}^{1/2}\m{u},
\qquad
\tilde{\m{F}}_1\coloneqq \m{P}^{1/2}\m{F}_1\m{P}^{-1/2},
\nonumber\\
\tilde{\m{F}}_2
&\coloneqq \m{P}^{1/2}\m{F}_2(\m{P}^{-1/2}\otimes \m{P}^{-1/2}).
\end{align}
The explicit forms of the matrices $\tilde{\m{F}}_1$ and $\tilde{\m{F}}_2$ are given in Appendix~\ref{app:quadraticODErescaled}. Now, note that $\m{P}\succ 0$ satisfies the Lyapunov inequality
\begin{align}
\label{eq:lyapunovmatrix}
    \m{P} \m{F}_1 + \m{F}_1^\dag \m{P} = - 2\nu  \m{P} \prec 0,
\end{align}
which is equivalent to the stability of the linearized problem~\cite{plischke2005transient}. Hence, the symmetrization~\eqref{eq:rescaling} is a Lyapunov transformation.

Recall that the logarithmic norm, relative to the spectral norm, of a matrix $\m{X}$ is defined as
\begin{equation}
\mu_2(\m{X})\coloneqq \lim_{\varepsilon\searrow 0} \frac{\lVert\m{I}+\varepsilon\m{X} \rVert_2-1}{\varepsilon},
\end{equation}
which also equals the largest eigenvalue of $\tfrac12(\m{X} + \m{X}^\dag)$. From the properties of $\tilde{\m{F}}_1$, it follows immediately that
\begin{equation}\label{eq:mu2F1}
\mu_2(\tilde{\m{F}}_{1})=-\bar{\nu}. 
\end{equation}

Equivalently~\cite{plischke2005transient}, in the symmetrized coordinates we have that the squared Euclidean norm of the solution is a Lyapunov functional 
\begin{equation}
\label{eq:free_energy_tilde}
\mc{F} \coloneqq \frac12 \lVert \tilde{\m{u}}\rVert_2^2
=
\frac12\sum_{\substack{\vn\in \mc{Z}_{\NF}^3 \\ \v{\alpha}\in \mc{N}_{\NH}^3}}  \lvert \tilde{g}_{\vn,\v{\alpha}} \rvert^2
\end{equation}
satisfies
\begin{equation}
\frac{\mathrm{d}}{\mathrm{d}t} \mc{F} \leq -2\bar{\nu} \mc{F},
\end{equation}
and hence decays exponentially along the trajectory.
Expressed in terms of the original coordinates, the Lyapunov functional is the $\m{P}$-squared norm of $\m{u}$:
\begin{align}\label{eq:Lyapunov_pullback}
\mc{F}[\m{u}]\coloneqq &\; \frac12\lVert \m{u}\rVert_{2,\m{P}}^2 \nonumber \\
:=&\; \frac12\sum_{\substack{\vn\in \mc{Z}_{\NF}^3 \\ \v{\alpha},\v{\alpha}'\in \mc{N}_{\NH}^3}}  g_{\vn,\v{\alpha}}^* (\m{P}_{\vn})_{\v{\alpha},\v{\alpha}'}g_{\vn,\v{\alpha}'} \nonumber \\
=&\; \frac12 \sum_{\vn\in \mc{Z}_{\NF}^3} \frac{1}{\kappa_{\vn}}\lvert g_{\vn,\v{0}} \rvert^2 + \frac12 \sum_{\substack{\vn\in \mc{Z}_{\NF}^3 \\ \v{\alpha}\in \mc{N}_{\NH}^3}} \lvert g_{\vn,\v{\alpha}} \rvert^2.
\end{align}
Finally, in terms of the physics, this Lyapunov functional is the free energy~\eqref{eq:freenergycontinuum} of the (band-limited) plasma data -- namely the free energy of plasma data with Fourier modes restricted to $\mc{Z}$ and Hermite modes restricted to $\mc{N}$, see Appendix~\ref{subsec:quadratic_Lyapunov} for details. 

In conclusion, the non-unitary rescaling~\eqref{eq:rescaling} has the matrix representation~\eqref{eq:utilde}, with $\m{P}$ a Lyapunov matrix~\eqref{eq:lyapunovmatrix}. In the rescaled coordinates the linear part is purely dissipative, in the sense that the Euclidean norm decays along the solution, which corresponds to a $\m{P}$-norm decay in the original coordinates. This is a Lyapunov functional, which exists due to the underlying stability of the linearized dynamics. Physically, it coincides with the free energy~\eqref{eq:freenergycontinuum} after the imposition of appropriate cutoffs.

\subsubsection{Norm estimates}
\label{sec:normestimates}

We establish a set of norm bounds required to evaluate the convergence of the algorithm. For a matrix~$\m{X}$, we define $\rVert \m{X}\rVert_2$ as the operator norm with respect to the Euclidean norm $\lVert \cdot \rVert_2$. First, we establish some basic properties of the matrix $\m{P}^{1/2}$.

\begin{proposition}[Properties of $\m{P}^{1/2}$]
\label{prop:P12norm}
The matrix $\m{P}^{1/2}$ satisfies
\begin{align}\lVert \m{P}^{1/2} \rVert_2=&\; \sqrt{1+\frac{1}{\tau}}, \\
\lVert \m{P}^{-1/2} \rVert_2=&\; 1, \\
\kappa_2(\m{P}^{1/2})=&\; \sqrt{1+\frac{1}{\tau}}.
\label{kappaPbound}
\end{align}
\end{proposition}

\begin{proposition}[Euclidean norm of $\m{u}$]
\label{prop:u2Pbounds}
Let $\m{u}\in \mathbb{C}^{\mc{Z}_{\NF}^3\times \mc{N}_{\NH}^3}$. Then,
\begin{equation}\label{eq:u2Pbounds}
\lVert \m{u}\rVert_2 \leq \lVert \m{P}^{1/2}\m{u}\rVert_2\leq  \sqrt{1+\frac{1}{\tau}}  \lVert \m{u}\rVert_2. 
\end{equation}
Also, let $\bar{\m{u}}(t)$, $\m{u}(t)$ be solutions to Eq.~\eqref{eq:rescaledsystemintro}, Eq.~\eqref{eq:quadraticODE}, respectively. Suppose that $R_{\m{P}}<1$ and $\mu(\bar{\m{F}}_1)<0$. Then, 
\begin{equation}
\label{eq:ubarnorminequality}
    \| \bar{\m{u}}(t)\|_2 \leq \| \bar{\m{u}}(0)\|_2
\end{equation}
and
\begin{equation}
\label{eq:unormlinequality}
    \| \m{u}(t)\|_2 \leq \kappa_2(\m{P}^{1/2}) \| \m{u}(0)\|_2.
\end{equation}
\end{proposition}

\begin{proposition}[Spectral norm of $\tilde{\m{F}}_1$]
\label{prop:F1tildenorm}
The matrix $\tilde{\m{F}}_1$, with the cutoff~\eqref{eq:Truncation}, satisfies
\begin{equation}\label{eq:F1tildebound}
\lVert \tilde{\m{F}}_1\rVert_2 \leq \bigg(\bigg(
\frac{12\pi}{L}\NF\sqrt{\NH}+\frac{1}{4\sqrt{\tau}}
\bigg)^2+\bar{\nu}^2\bigg)^{1/2}
\end{equation}
and
\begin{align}
\label{eq:F1tildelowerbound}
    \lVert \tilde{\m{F}}_1\rVert_2 \geq \bigg(  \bigg(\frac{2\sqrt{3}\pi}{L}\NF\sqrt{\NH}\bigg)^2+\bar{\nu}^2\bigg)^{1/2}.
\end{align}
In particular,
\begin{equation}
    \lVert \tilde{\m{F}}_1\rVert_2=\Theta\big(\NF\sqrt{\NH}\big).
\end{equation}
\end{proposition}

\begin{proposition}[Spectral norm of $\tilde{\m{F}}_2$]
\label{prop:F2tildenorm}
The matrix $\tilde{\m{F}}_2$, with the cutoff~\eqref{eq:Truncation}, satisfies
\begin{align}\label{eq:F2tildespectralnorm}
\lVert \tilde{\m{F}}_2\rVert_2 \leq
&\; \frac{\sqrt{3}}{L^{3/2}}\bigg(\frac{1}{1+\tau}+\frac{13L^2}{2\pi^2}\NF \bigg)^{1/2}\big( 1+\sqrt{3\NH}\big).
\end{align}
\end{proposition}
In particular,
\begin{equation}
\label{eq:F2tildenormscaling}
\lVert \tilde{\m{F}}_2\rVert_2=O(\sqrt{\NF \NH}).
\end{equation}

The proofs of these are provided in Appendix~\ref{app:normestimates}.

\subsection{Carleman linear embedding}
\label{sec:Carleman}

We next embed the stabilized quadratic Galerkin system into a finite linear Carleman system and state the convergence criterion used by the quantum solver.

\subsubsection{Setup}

The quantum simulation will be carried out in a further rescaled version of the tilde coordinates:
\begin{align}
\label{eq:ubarappedix}
 \bar{\m{u}} \coloneqq \frac{1+R_{\mathrm{thr}}}{2 \|\tilde{\m{u}}_0\|_2} \tilde{\m{u}},    
\end{align}
so that
\begin{align}
\label{eq:Lyapunovandrescaledsystem}
\dot{\bar{\m{u}}}=\bar{\m{F}}_1\bar{\m{u}}+\bar{\m{F}}_2\bar{\m{u}} \otimes \bar{\m{u}}, \quad \bar{\m{u}}(0)= \bar{\m{u}}_0
\end{align}
with
\begin{align}
\label{eq:Frescalings}
    \bar{\m{F}}_1 = \tilde{\m{F}}_1, \quad \bar{\m{F}}_2 = \frac{2 \| \tilde{\m{u}}_0\|_2}{1+ R_{\mathrm{thr}}} \tilde{\m{F}}_2.
\end{align}
We have
\begin{align}
\label{eq:uobar}
    \|\bar{\m{u}}_0\|_2 = \frac{1+ R_{\mathrm{thr}}}{2}<1.
\end{align}

The linear representation of the nonlinear dynamics is obtained by introducing powers of $\bar{\m{u}}\in \mathbb{C}^N$ as further variables,  where $N=  (2\NF+1)^3(\NH+1)^3$. We define $\bar{\m{w}}_j := \bar{\m{u}}^{\otimes j}$ for any integer $j=1,2,\dots \infty$ (we set $\bar{\m{w}}_0 :=1$). We also define $N^j \times N^k$ matrices~$\bar{\m{A}}^j_k$ 
\begin{align}
    \bar{\m{A}}^j_j  &:= \bar{\m{F}}_1 \otimes \m{I}^{\otimes (j-1)}  + \mbox{shifts},  \nonumber\\
    \bar{\m{A}}^j_{j+1}  &:=\bar{\m{F}}_2 \otimes \m{I}^{\otimes (j-1)}+ \mbox{shifts},
\end{align}
where ``$\mbox{shifts}$" denotes that the other terms that are obtained by shifting one after the other single copies of $\m{I}$ past the $\bar{\m{F}}_k$ term, from the right-hand side to the left-hand side. For example, 
$$\bar{\m{A}}_{3}^{3} = \bar{\m{F}}_1 \otimes \m{I} \otimes \m{I} + \m{I} \otimes \bar{\m{F}}_1 \otimes \m{I} +  \m{I} \otimes \m{I} \otimes \bar{\m{F}}_1.$$

By a direct calculation (see, e.g., \cite{liu2021efficient}) the infinite set of variables $\bar{\m{w}}_j$ satisfy an infinite set of ODEs: 
\begin{align}
 \dot{\bar{\m{w}}}_j &=  \bar{\m{A}}^j_j\bar{\m{w}}_{ j}+ \bar{\m{A}}^j_{j+1} \bar{\m{w}}_{j+1} \mbox{ for all }j\in \mathbb{N} \nonumber \\
 \bar{\m{w}}_0(t) &:=1 \nonumber\\
 \bar{\m{w}}_j (0) &= \bar{\m{u}}_{0}^{\otimes j}\mbox{ for all }j\in \mathbb{N}.
\end{align}
By construction, the solution to this system takes the form $\bar{\m{w}}_j(t) = \bar{\m{u}}(t)^{\otimes j}$ for all $j \in \mathbb{N}$, where $\bar{\m{u}}(t)$ is the exact solution at time $t \ge 0$. Thus $\bar{\m{w}}(t) = [\bar{\m{w}}_1(t), \bar{\m{w}}_2(t), \dots ]$ is an infinite dimensional vector.

We now turn to the dynamics of a finite, truncated Carleman system. To distinguish this from the infinite case we use $\bar{\m{z}}_j$ as the variables for the finite system. If we truncate to a level $j=\NC$ and set $\bar{\m{A}}^{\NC}_{\NC+1} =0$ for the closure, we obtain a finite ODE system
\begin{align}
 \dot{\bar{\m{z}}}_j =&\;  \bar{\m{A}}^j_j\bar{\m{z}}_{ j}+ \bar{\m{A}}^j_{j+1} \bar{\m{z}}_{j+1} \quad \quad j=1,\dots \NC -1 \nonumber \\
  \dot{\bar{\m{z}}}_{\NC} =&\; \bar{\m{A}}^{\NC}_{\NC}\bar{\m{z}}_{\NC} \nonumber \\
 \bar{\m{z}}_0(t) =&\;1 \nonumber\\
 \bar{\m{z}}_j (0) =&\;  \bar{\m{u}}_{0}^{\otimes j} \quad j =1,\dots \NC.
 \label{eq:Carlemaninfinite}
\end{align}
We denote the solution to the truncated system by $\bar{\m{z}}_j(t)$ for all $j=1,\dots, \NC$. By writing $\bar{\m{z}} = [\bar{\m{z}}_1, \bar{\m{z}}_2, \dots, \bar{\m{z}}_{\NC}]$ we therefore get that $
    \bar{\m{z}}(t) \in \mathbb{C}^{\sum_{k=1}^{\NC} N^k}$,
and so the dimension of $\bar{\m{z}}(t)$ is $ O(N^{\NC})$.
This vector is subject to the system of equations on $\bigoplus_{j=1}^{\NC} \mathbb{C}^{N^j}$:
\begin{align}
    \dot{\bar{\m{z}}} &= \bar{\m{A}}\bar{\m{z}}, \nonumber
    \\
    \bar{\m{z}}(0) &= [ \bar{\m{u}}^0,\bar{\m{u}}^{0^{\otimes 2}},\bar{\m{u}}^{0^{\otimes 3}}, \dots, \bar{\m{u}}^{0^{\otimes \NC}}  ].
    \label{eq:barproblem}
\end{align}
We shall refer to this as the \emph{Carleman ODE system}. 

\subsubsection{Encoding}
\label{sec:encoding}

For the later block-encoding constructions it is convenient to realize the
truncated Carleman system inside the fixed tensor-product Hilbert space
\begin{equation}
\label{eq:H_definition}
\mathcal H
\coloneqq
\mathrm{span}\{\ket{1},\dots,\ket{\NC}\}
\otimes
(\mathbb{C}^{N})^{\otimes \NC}.
\end{equation}
Within $\mathcal H$, the $j$th Carleman sector is identified with the subspace 
\begin{equation}
\label{eq:Vj_definition}
V_j
\coloneqq
\mathrm{span}
\Bigl\{
\ket{j}\underbrace{\ket{i_1}\cdots\ket{i_j}\ket{0}^{\otimes (\NC-j)}}_{\mathrm{data \; registers}}
\Bigr\},
\end{equation}
and the full working space is
\begin{equation}
\label{eq:V_definition}
V\coloneqq \bigoplus_{j=1}^{\NC} V_j \subseteq \mathcal H.
\end{equation}
Under this identification, $\v z_j$ is stored in the first $j$ data registers,
with the remaining $\NC-j$ registers padded by~$\ket{0}$. Specifically, after a square matrix embedding, $\bar{\m{A}}$ is a block matrix with block-diagonal components $\bar{\m{A}}^j_j$ and superdiagonal blocks $\bar{\m{A}}^j_{j+1}$. This gives the ODE system
    \begin{align}
    \nonumber
    \dot{\bar{\m{z}}}(t) &= \bar{\mathcal{A}}\bar{\m{z}}(t)
    \\
    \bar{\m{z}}(0) &=\bar{\m{z}}^0= [ \bar{\m{u}}^0,\bar{\m{u}}^{0^{\otimes 2}},\bar{\m{u}}^{0^{\otimes 3}}, \dots, \bar{\m{u}}^{0^{\otimes \NC}}  ],
    \label{eq:CarlemanODE}
\end{align}
for the embedding $\bar{\mathcal{A}}$ of the Carleman matrix $\bar{\m{A}}$, which is a matrix on $V$. This defines the dynamics that will be solved on a quantum computer.

Using Eq.~\eqref{eq:RPintro} (which also holds in the tilde coordinates) as well as Eq.~\eqref{eq:Frescalings}, we have that:
\begin{align}
\label{eq:rescalingchoiceworks}
    \mu(\bar{\m{F}}_1) + \| \bar{\m{F}}_2 \|_2 \leq  \mu(\bar{\m{F}}_1) \frac{1-R_{\mathrm{thr}}}{1+R_{\mathrm{thr}}} < 0. 
\end{align}
From the fact that our choice of rescaling satisfies Eq.~\eqref{eq:rescalingchoiceworks}, together with \cite[Eq.~B.39]{jennings2025quantum} that the log-norm of the rescaled Carleman system satisfies
\begin{align}
\label{eq:Carlemanlognorm}
    \mu(\bar{\m{A}}) < \mu(\bar{\m{F}}_1)/2 = - \bar{\nu}/2 <0.
\end{align} 
In other words, $ \mu(\bar{\m{A}}) $ is bounded away from zero by a constant. In particular the condition $R_{\m{P}}<1$ together with $\mu(\bar{\m{F}}_1)<0$ imply \cite[Lemma~15]{krovi2022improved},   \cite[Lemma~3.3]{jennings2025quantum}
\begin{align}
\label{eq:normdecreases}
    \|\bar{\m{u}}(t)\|_2 \leq \| \bar{\m{u}}(0)\|_2 \quad  (t \geq 0)
\end{align}

The strategy is to carry out the simulation in the above `bar' coordinates and then transform back to the original coordinates at the stage of information extraction.

\subsubsection{Linear embedding error bound and free energy threshold}

There are several approaches to bounding the total error
\begin{align}
\label{eq:totalerror}
    \epsC = \max_{t \geq 0} \| \bar{\m{z}}(t) - \bar{\m{w}}(t)\|_2,
\end{align}
where with a slight abuse of notation we indicated by 
\begin{align}
\label{eq:w}
\bar{\m{w}}(t) = [\bar{\m{u}}(t), \bar{\m{u}}(t)^{\otimes 2}, \dots, \bar{\m{u}}(t)^{\otimes \NC}]    
\end{align} 
the first $\NC$ components of the solution to the infinite system in Eq.~\eqref{eq:Carlemaninfinite}.

The construction of a Lyapunov functional from Sec.~\ref{sec:galerkin} and the results from Ref.~\cite{jennings2023cost} [Theorem 3.6] can be leveraged to find such error bound. In particular, since $\mu(\bar{\m{F}}_1) = \mu(\tilde{\m{F}}_1)  <0$ (Eq.~ \eqref{eq:mu2F1}), and $\|\bar{\m{u}}(0)\|_2<1$ (Eq.~\eqref{eq:uobar}) we have that under the condition
\begin{align}
\label{eq:RP}
    R_{\m{P}}:=  \frac{1}{-\mu(\bar{\m{F}}_1)} \| \bar{\m{F}}_2\|_2 \|{\bar{\m{u}}}_0\|_2  = \frac{1}{-\mu(\tilde{\m{F}}_1)} \| \tilde{\m{F}}_2\|_2 \| {\tilde{\m{u}}}_0\|_2 <1,
\end{align}
the error decays exponentially with $\NC$, i.e., if we take
\begin{align}
\label{eq:Carlemantruncationerror}
    \NC = O\left (\log\frac{1}{\epsC} \right),
\end{align}
we can ensure the linear embedding error is at most $\epsilon_{\mathrm{C}}$.\footnote{To see this, simply note that with the definitions of Ref.~\cite{jennings2025quantum} one has $\mu(\tilde{\m{F}}_1) = \mu_{2,\m{P}}(\m{F}_1)$, $\| \tilde{\m{F}_2}\| = \|\m{F}_2\|_{\m{P}}$, $\| \tilde{\m{u}}\|_2 = \| \m{u}\|_{2,\m{P}}$. } In summary, the convergence of the Carleman linear embedding of the system \eqref{eq:quadraticODE} is determined by the Lyapunov $R$-number condition \eqref{eq:RP}. 

From Proposition~\ref{prop:u2Pbounds} and the assumption in Eq.~\eqref{eq:u0bound}, $\| \tilde{\m{u}}_0\|_2 = O(1)$, we have that $\| \tilde{\m{F}}_2\|_2 = O(\sqrt{\NF \NH})$ and so $R_{\m{P}} = O(\sqrt{\NF \NH})$, which sets an upper limit to the finest resolution for which we can guarantee convergence under this particular stability method:

\begin{proposition}[Upper bound on $R_{\m{P}}$]
\label{prop:RP_upper_bound}
\begin{align}
\label{eq:RP_upper_bound}
R_{\m{P}}\leq &\; \frac{\sqrt{3}}{\bar{\nu}L^{3/2}}
\bigg(\frac{1}{1+\tau}+\frac{13L^2}{2\pi^2}\NF\bigg)^{1/2}
\nonumber \\
&\; \times\big(1+\sqrt{3\NH}\big)\lVert \m{u}(0)\rVert_{2,\m{P}}.
\end{align}
\end{proposition}

\begin{proof}
By putting \eqref{eq:mu2F1} and \eqref{eq:F2tildespectralnorm} into the definition of $R_{\m{P}}$ \eqref{eq:RP}, we obtain \eqref{eq:RP_upper_bound}.
\end{proof}

A simple algebraic manipulation allows one to rephrase the condition $R_{\m{P}}<1$ as in Eq.~\eqref{eq:freenergycondition}.

\subsubsection{$R_{\m{P}}$-number condition and plasma measures of nonlinearity}

In plasma physics, the local nonlinearity strength at time~$t$ is characterized by the quantity \eqref{eq:phibar}, which becomes
\begin{equation}
\label{eq:eta_transformed}
\varphi(\vx,t)=\tau \phi(\vx,t)
\end{equation}
under the transformations \eqref{eq:natural_units}, where as before we drop the bars after performing the transformation to dimensionless quantities. To obtain global measures of nonlinearity strength for the problem we consider, we take norms of \eqref{eq:eta_transformed} at $t=0$ to define
\begin{equation}
\label{eq:eta_avg}
\varphi_{\mathrm{avg}}\coloneqq \frac{1}{L^3}\lVert \varphi(\cdot,0)\rVert_{1,\TL^3}=\frac{\tau}{L^3}\int_{\TL^3} \lvert\phi(\vx,0)\rvert\,\mathrm{d}^3x
\end{equation}
and 
\begin{equation}
\label{eq:eta_max}
\varphi_{\max}\coloneqq \lVert \varphi(\cdot,0)\rVert_{\infty,\TL^3}=\tau\max_{\vx\in\TL^3}\lvert\phi(\vx,0)\rvert.
\end{equation}
In the following proposition, all quantities are assumed to be band-limited functions on $\TL^3$ with Fourier cutoff $\NF$.

\begin{proposition}[Necessary averaged-nonlinearity bound from $R_{\m P}<1$]
\label{prop:eta_avg_RP_necessary}
Suppose that the Lyapunov \(R\)-number satisfies $R_{\m{P}}<1$.
Then,
\begin{equation}
\varphi_{\mathrm{avg}}
<
\frac{\bar{\nu}\sqrt{\tau}}
{L^{3/2}\lVert \tilde{\m F}_2\rVert_2}.
\end{equation}
\end{proposition}

A proof is given in Appendix~\ref{subsec:eta_avg_proof}. We highlight that this bound is a very blunt estimate of what nonlinearities can be simulated with exponential precision. Determining the true scope requires future analysis and numerical simulations.

\subsubsection{Convergence guarantees for illustrative initial data}
\label{subsec:initial_data}

We consider specific initial data $(g(\vx,\vv,0),\phi(\vx,0))$ for the non-dimensionalized Vlasov--Poisson system \eqref{eq:gnonlinear}--\eqref{eq:phi_g_single_species}, and analyze the available a priori convergence guarantees through Proposition~\ref{prop:RP_upper_bound}.

\begin{proposition}[Sufficient $\varphi_{\max}$ for real Fourier mode initial data]
\label{prop:Fourier_initial_data}
Consider the initial data \eqref{eq:initialconditiontheorem}. Under the assumptions of Theorem~\ref{theorem}, we have
\begin{equation}
\label{eq:u0P_Fourier_mode}
\lVert \m{u}(0)\rVert_{2,\m{P}}
=
\frac{\varphi_{\max}}{\tau}
\bigg(
\frac{L^3}{2}\kappa_{\vn_0}(1+\kappa_{\vn_0})
\bigg)^{1/2}.
\end{equation}
and a sufficient condition for $R_{\m P}<1$ is
\begin{align}
\label{eq:varphimax_Fourier_mode}
\varphi_{\max}
<
&\;
\frac{\sqrt{2}\tau\bar{\nu}}{\sqrt{3\kappa_{\vn_0}(1+\kappa_{\vn_0})}}
\nonumber\\
&\;\times
\bigg(
\bigg(
\frac{1}{1+\tau}+\frac{13L^2}{2\pi^2}\NF
\bigg)^{1/2}
\big(1+\sqrt{3\NH}\big)
\bigg)^{-1}. 
\end{align}
Moreover, $\varphi_{\max}=O\big((\NF\NH)^{-1/2}\big)$ as $\NF,\NH\to\infty$ and under Debye-resolution scaling $\NF=\Theta(L)$, $\varphi_{\max}=O\big((L\NH)^{-1/2}\big)$.
\end{proposition}

A proof is given in Appendix~\ref{subsec:Fourier_proof}.

We now consider Gaussian initial data. Define the periodized Gaussian
\begin{equation}\label{eq:Gm}
\mc{G}(\vx;L,\sigma)
\coloneqq
\frac{1}{(2\pi)^{3/2}\sigma^3}\sum_{\vn\in \mathbb{Z}^3}
\exp\bigg(
-\frac{\lvert \vx+\vn L\rvert^2}{2\sigma^2}
\bigg).
\end{equation}
and consider the initial data 
\begin{align}
g(\vx,\vv,0)=&\; C\bigg(\mc{G}(\vx;L,\sigma)-\frac{1}{L^3} \bigg), \label{eq:Gaussian_g0} \\
\phi(\vx,0)=&\; \frac{C}{L^3}\sum_{\vn\in \mathbb{Z}^3\setminus\{0\}} \frac{\exp\big(-\frac{\sigma^2}{2}\lvert\v{\xi}_{\vn}\rvert^2 \big)}{\lvert\v{\xi}_{\vn}\rvert^2+\tau}\exp\big(i\v{\xi}_{\vn}\cdot \vx\big)\label{eq:Gaussian_phi0}.
\end{align}
Observe that $M(0)=0$ and  \eqref{eq:phi_g_single_species} and satisfied, due to the exact integral
\begin{equation}
\int_{\TL^3} \mc{G}(\vx;L,\sigma)\,\mathrm{d}^3x =1.
\end{equation}
Moreover, $\lim_{\sigma\to\infty} g(\vx,\vv,0)=0$, so that the initial datum $g(\vx,\vv,0)$ interpolates between a  sharply peaked Gaussian and zero perturbation.

We employ the lattice sums $S_1(L,\sigma)$ \eqref{eq:S1} and
\begin{equation}\label{eq:Gaussian_S2_def}
S_2(\NF,L,\sigma)\coloneqq
\frac{1}{L^3}\sum_{\vn\in \mc{Z}_{\NF}^3\setminus\{\v 0\}}
\bigg(
1+\frac{1}{\kappa_{\vn}}
\bigg)
\exp\big(
-\sigma^2\lvert\v\xi_{\vn}\rvert^2
\big).
\end{equation}
in the following proposition.

\begin{proposition}[Sufficient $\varphi_{\max}$ for Gaussian initial data]
\label{prop:Gaussian}
Consider the initial data \eqref{eq:Gaussian_g0}--\eqref{eq:Gaussian_phi0}. We have 
\begin{equation}\label{eq:Gaussian_u0P_exact}
\lVert \m u_0\rVert_{2,\m P}
=
\lvert C\rvert S_2(\NF, L,\sigma)^{1/2},
\end{equation}
and a sufficient condition for $R_{\m{P}}<1$ is 
\begin{align}
\label{eq:Gaussian_RP}
\varphi_{\max} < &\; \frac{\tau\bar{\nu}L^{3/2} S_1(L,\sigma)}{\sqrt{3}S_2(\NF,L,\sigma)^{1/2}} \nonumber \\
&\; \times \bigg(
\frac{1}{1+\tau}
+
\frac{13L^2}{2\pi^2}\NF
\bigg)^{-1/2}
\big(1+\sqrt{3\NH}\big)^{-1}. 
\end{align}
Moreover, suppose that \(\sigma=cL\) for fixed \(c\in(0,1)\). Then, $\varphi_{\max}
    =
    O\big((\NF\NH)^{-1/2}
    \big), $
as \(\NF,\NH\to\infty\) and under
Debye-resolution scaling \(\NF=\Theta(L)\), $
    \varphi_{\max}
    =
    O\big(
    (L\NH)^{-1/2}\big)$.
\end{proposition}

A proof is given in Appendix~\ref{subsec:Gaussian_proof}.

\begin{rmk}[Radial approximants for the Gaussian profiles]
We may approximate the periodized Gaussian initial data
\eqref{eq:Gaussian_g0}--\eqref{eq:Gaussian_phi0}
by the radial functions
\begin{align}
g_{\mathrm{rad}}(\vx,\vv,0)
&\coloneqq
C\bigg(
\frac{1}{(2\pi)^{3/2}\sigma^3}
\exp\bigg(
-\frac{\lvert \vx\rvert^2}{2\sigma^2}
\bigg)
-\frac{1}{L^3}
\bigg),
\label{eq:g_radial_approx}
\end{align}
and
{\small
\begin{align}
\phi_{\mathrm{rad}}(\vx,\vv,0)\coloneqq & \,
C\frac{
\exp\big(-\lvert \vx\rvert^2/(2\sigma^2)\big)
}{
8\pi \lvert\vx\rvert
} \nonumber\\
& \!\! \times 
\bigg(\!\!
\operatorname{erfcx}
\bigg(
\frac{\sqrt{\tau}\sigma^2-\lvert \vx\rvert}{\sqrt{2}\sigma}
\! \bigg)
-
\operatorname{erfcx}
\bigg(
\frac{\sqrt{\tau}\sigma^2+\lvert \vx\rvert}{\sqrt{2}\sigma}
\bigg)
\! \! \bigg) \nonumber \\
& \!
-\frac{C}{\tau L^3},
\label{eq:phi_radial_approx}
\end{align}}
where $\operatorname{erfcx}(z)\coloneqq\exp(z^2)\operatorname{erfc}(z)$. The function \eqref{eq:phi_radial_approx} is defined for $\vx\neq \v{0}$ and extended to $\vx=\v{0}$ by continuity. Observe that the pair
$(g_{\mathrm{rad}},\phi_{\mathrm{rad}})$ verifies
\eqref{eq:phi_g_single_species} in the radial approximation.
\end{rmk}

\subsection{Quantum algorithm}

We now describe the quantum implementation of the Carleman system, including block-encodings, state preparation, history-state generation, and information extraction.

\label{sec:block-encodingsF1F2}
\subsubsection{Block-encoding of $\m{F}_1$, $\m{F}_2$ and $\m{P}^{-1/2}$}

The plasma problem is specified by the matrices $\m{F}_1$, $\m{F}_2$, and the initial datum $g(\v{x},\v{v},0)$. In order to have an efficient quantum algorithm we must be able to efficiently encode these on a quantum computer. This is done through unitary block-encodings of the structured data.

Recall that a $(\alpha_\m{M}, a, \epsilon)$ block-encoding of a matrix $\m{M}$ is a unitary $\m{U}_\m{M}$ such that
\begin{align}
    \m{U}_\m{M} \ket{0^a} \ket{\psi} = \ket{0^a} \frac{\m{M}_\epsilon}{\alpha_\m{M}} \ket{\psi} + \ket{\perp}, 
\end{align}
where 
\begin{align}
    (\bra{0^a} \otimes \m{I})\ket{\perp} = 0, \quad \| \m{M}_\epsilon - \m{M}\|_2 \leq \epsilon.
\end{align}
Various methods exist for explicitly constructing such a block-encoding. In particular, when the matrix elements have certain sparsity features and are specified by explicit functions it is possible to construct block-encodings. More precisely, if $\m{M}$ is $\dr(\m{M})$ row-sparse and $\dc(\m{M})$ column-sparse then we can construct a block-encoding $\m{U}_\m{M}$ with scale factor $\alpha_\m{M} = \|\m{M}\|_{\max}\sqrt{\dc(\m{M})\dr(\m{M})}$, where $\|\m{M}\|_{\max}:= \max_{i,j} |\m{M}_{ij}|$~\cite{gilyen2018quantum}. This basic construction can be directly applied to the sparse matrices in our problem.

The matrix $\m{P} =\oplus_{\v{n}} \m{P}_{\v{n}}>0$ is diagonal in the Fourier index $\v{n}$, and as shown above  
\small
\begin{align}
\nonumber
(\m{P}^{\pm \frac{1}{2}})_{\vn,\v{\alpha}|\vn',\v{\alpha}'}
& =\;
\delta_{\vn,\vn'}\delta_{\v{\alpha},\v{\alpha}'} \\ & +\bigg(
\left(\frac{1+\kappa_{\vn}}{\kappa_{\vn}}\right)^{\pm \frac{1}{2}}\! \! \! \! \! \! -1
\bigg)
\delta_{\vn,\vn'}\delta_{\v{\alpha},\v{0}}\delta_{\v{\alpha}',\v{0}}.
\end{align}
\normalsize
These are diagonal operators whose entries are simple functions, with size ranging in $[1, \sqrt{ 1 + \tau^{-1}}]$ and $[1/\sqrt{ 1 + \tau^{-1}},1]$, respectively. The constant $\sqrt{ 1 + \tau^{-1}}$ typically does not vary much over plasmas (e.g. for tokamak core plasmas $1/2 \lesssim T_i/T_e \lesssim 2 $). Therefore, we can efficiently implement the Quantum Lyapunov transform in either direction.

For $\tilde{\m{F}}_1$ the situation is also simple. It can be shown (see Appendix~\ref{app:quadraticODErescaled} for details) that it is $7$-sparse, and has 
\begin{equation}
    \|\tilde{\m{F}}_1\|_{\max} = \left (\frac{\NF}{L} \sqrt{\NH} + \bar{\nu} \right ),
\end{equation}
with components that can readily computed coherently. Using a sparse block-encoding we can therefore construct a block-encoding of $\tilde{\m{F}}_1$ with scale factor 
\begin{equation}
    \alpha_{\tilde{\m{F}}_1} = 7 \|\tilde{\m{F}}_1\|_{\max}= O\left (\frac{\NF}{L} \sqrt{\NH} + \bar{\nu} \right ).
\end{equation}

However, $\tilde{\m{F}}_2$ is more involved. We find that $\tilde{\m{F}}_2$ is $\dc=3$ column-sparse, however the row sparsity is $\dr=3 (2\NF+1)^3$ (see Appendix~\ref{app:quadraticODErescaled} for details). It is also readily computed that
\scriptsize
\begin{align}
    \|\tilde{\m{F}}_2\|_{\max}     & \le \frac{1}{L^{3/2}(\sqrt{\tau}+\sqrt{1+\tau})} \left ( \sqrt{\NH+1} +  \sqrt{\frac{ 3\NF^2k^2+\tau}{1+3\NF^2k^2+\tau}} \right ),
\end{align}
\normalsize
where $k=2\pi/L$. Thus, for a fixed minimal real-space resolving limit (corresponding to $\NF k$ being a constant), we see that $\|\tilde{\m{F}}_2\|_{\max} = O\left(\frac{1}{\sqrt{\tau}} \sqrt{\NH/L^3} \right)$. However, we can do better than applying the sparse block-encoding method to $\tilde{\m{F}}_2$, which has the costly $3(2\NF+1)^3$ for the row-sparsity. More precisely, we can exploit the fact that its real-space profile is rapidly decaying and leverage hierarchical block-encoding methods. The spectrum of $\tilde{\m{F}}_2$ decays like $O(1/|\v{n}|)$ in Fourier space, however it also needs regulation of the singularity at $\v{x}=\v{0}$. In real-space the electrostatic field decays much faster and we leverage this to construct a more efficient block-encoding. Specifically we sample the field in real-space and by doing so we can neatly define (a) a regulation of the singularity, and (b) a truncated interaction $\tilde{\m{G}}_2$ restricted to the finite subspace of Fourier-Hermite modes. This $\tilde{\m{G}}_2$ can be taken to define the finite-dimensional model of the plasma nonlinearity, and indeed as $\NF$ increases any aliasing errors decrease and $\tilde{\m{G}}_2$ converges rapidly to $\tilde{\m{F}}_2$. Since we are here interested in asymptotics (and indeed it turns out the linear term $\tilde{\m{F}}_1$ is anyways the algorithm bottleneck) we show how to obtain an optimal block encoding of the nonlinear term via this real-space sampling.

To this end, in Appendix~\ref{app:HBEs} we derive general theory on block-encoding dense matrices, such as the ones arising here, and which may be of independent interest. Specifically we establish Theorem~\ref{thm:two-kernel-hbe} that provides a unitary block-encoding of an arbitrary, dense rectangular matrix
\begin{equation}
\label{eq:arbitrary-2-kernel}
    \m K_{\v j,\v\alpha \mid \v j',\v\alpha';\v j'',\v\alpha''}= \sigma_{\v j,\v j''} \sum_{a=1}^3 K^a_{\v j,\v j'}\, W^a_{\v\alpha\mid \v\alpha',\v\alpha''},
\end{equation}
where both $\sigma$ and $K^a$ are assumed to have decaying off-diagonal components, while $W^a$ are generic additional degrees of freedom. 

The difficulty is that the spatial matrices \(\sigma\) and \(K^a\) may be
dense, even though their entries decay with spatial separation.  To avoid this we develop Theorem~\ref{thm:two-kernel-hbe}, which handles this by applying a hierarchical decomposition to each spatial kernel:
near-field interactions are kept in sparse adjacent matrices, while far-field
interactions are grouped into levels.  Each level is then
block-encoded using a sparse-access construction, and the full operator is assembled by an LCU over the levels of the hierarchy. The hierarchical construction weights each level by the largest interaction strength on that level, so distant blocks contribute less when the physical interaction decays.  Thus the normalization can be polylogarithmic, or at worst a much smaller power of the grid size, depending on the decay rates of \(\sigma\) and \(K^a\). 

The theorem also unifies several special cases: setting one kernel to a Kronecker delta recovers local Coulomb interactions in plasmas, while fixing an auxiliary input register recovers the square-kernel hierarchical block-encoding of Ref.~\cite{nguyen2022block}.

We find that for a matrix of the form given in Eq.~\eqref{eq:arbitrary-2-kernel} that we may construct a unitary block-encoding with scale factor
\small
\begin{equation}
    \alpha_{\m K}=\sum_{a=1}^3\sum_{\lambda=1}^{\ell_{\max}}\sum_{\ell=1}^{\ell_{\max}}\bar\sigma_\lambda\,\bar K_{a,\ell}\,w_a\,\sqrt{D_\lambda D_\ell\,\Gamma^a_{\lambda,\ell}\,\dr(W^a)\dc(W^a)},
\end{equation}
\normalsize
where $\bar\sigma_\lambda,\bar K_{a,\ell}$ are max element norms for $\ell_{\max}$ geometric contributions to $\sigma$ and $ K^a$, respectively, and $w_a:= \| W^a\|_{\max}$. The number of terms $\ell_{\max}$ scales logarithmically with the matrix size. The terms $D_\lambda, D_\ell$ are combinatorial factors for the ordered geometric components, $\Gamma^a_{\lambda,\ell}$ is an intersection number for the supports of the two kernels, and $d_{R,C}(W^a)$ denotes the row/column-sparsity of $W^a$. 

With this result we can recover a range of prior results, and determine the scaling of $\alpha_{\m K}$ with the decay profiles of the two kernels. In particular in Appendix~\ref{app:HBEs}, we find for our present model that
\begin{equation}
    \alpha_{\tilde{\m{F}}_2} = O \left ( \sqrt{\NF \NH} \right ),
\end{equation}
which is a substantial improvement over the direct sparse block-encoding of $\tilde{\m{F}}_2$, and since $\|\tilde{\m{F}}_2\|= O(\sqrt{\NF \NH})$, the hierarchical method gives an asymptotically optimal block-encoding of the nonlinear term.

\subsubsection{Oracle-constructions for the ODE solver}

We run the quantum ODE solver algorithm from Ref.~\cite{berry2017quantum,jennings2023cost}, as described in more detail in the next section. The quantum algorithm requires access to two unitaries:
\begin{enumerate}
    \item A $(\alpha_{\bar{\m{A}}}, a_{\bar{\m{A}}}, \epsilon_{\bar{\m{A}}})$ block-encoding of the matrix $\bar{\m{A}}$.
    \item A unitary $\m{U}_{\m{z}_0}$ preparing the vector $\ket{\m{z}_0} \propto [ \bar{\m{u}}_0, \bar{\m{u}}^{\otimes 2}_0, \dots, \bar{\m{u}}_0^{\otimes \NC}]$
\end{enumerate}
We discuss each in turn. 

\emph{Block-encoding of $\bar{\m{A}}$.} This construction is standard~\cite{liu2021efficient}, so we shall only sketch it here. First, one embeds the direct sum space $V$, defined in Eq.~\eqref{eq:V_definition}, into an $\NC N^{\NC}$ dimensional space $\mathcal{H}$, with $N:=(\NH+1)^3(2\NF+1)^3$, as described in Sec.~\ref{sec:encoding}.  Then, we shall use the  $(\alpha_{\bar{\m{F}}_1}, a_1, \epsilon)$ block-encoding $\m{U}_{\bar{\m{F}}_1} $ of $\bar{\m{F}}_1$, and the $(\alpha_{\bar{\m{F}}_2}, a_2, \epsilon)$ block-encoding $\m{U}_{\bar{\m{F}}_2}$ of $\bar{\m{F}}_2$ constructed in the previous section. The matrix $\bar{\m{A}}$ can be decomposed as $\bar{\m{A}} = \bar{\m{A}}_1 + \bar{\m{A}}_2$. While $\bar{\m{A}}_1$ has the form
\begin{align}
  \bar{\m{A}}_1 = \sum_{i=1}^{\NC} \ketbra{i}{i} \otimes \mathcal{F}_1^i   
\end{align}
where
$$\mathcal{F}_1^i = \underbrace{\bar{\m{F}}_1 \otimes \m{I} \otimes \dots \otimes \m{I}}_{i \textrm{ terms}} + \mathrm{shifts}.$$
 For each block, the term  $\mathcal{F}_1^i$ can be block-encoded as $\m{U}_{\bar{\m{F}}_1} \otimes \m{I} \otimes \dots \otimes \m{I}$, and the shifted terms by permutations of this. Block $i$ can then be block-encoded by an LCU, where the PREP is a uniform state preparation over $i$ indexes, and the SELECT involves controlled permutations of $\m{U}_{\bar{\m{F}}_1} \otimes \m{I} \otimes \dots \otimes \m{I}$. This unitary block-encoding is further wrapped in a `block-diagonal LCU', see e.g. Lemma~2 \cite{jennings2025end}, which introduces a control over the label $i$ and simple $1$-qubit state preparations. Note that $\m{U}_{\bar{\m{F}}_1} \otimes \m{I} \otimes \dots \otimes \m{I}$ remains fixed and uncontrolled. This construction allows one to construct a block-encoding of $\bar{\m{A}}_1$ with prefactor $\NC \alpha_{\bar{\m{F}}_1}$, using $a_1+ O( \log \NC)$ ancillary qubits and a single application of $\m{U}_{{\bar{\m{F}}_1}}$. A similar construction allows one to block-encode $\bar{\m{A}}_2$, and then $\bar{\m{A}}_1 + \bar{\m{A}}_2$ are summed via a final LCU. The final result is that we can construct a block-encoding of $\bar{\m{A}}$ with
 \begin{align}
 \nonumber
     \alpha_{\bar{\m{A}}} =&\; \NC(\alpha_{\bar{\m{F}}_1} + \alpha_{\bar{\m{F}}_2}) = O\big(\log(1/\epsilon) N_{\m{F}} \sqrt{N_\m{H}}\big), \\
     \nonumber
     a_{\bar{\m{A}}}  =&\; \max\{a_1,a_2\} + O(\log(\NC)) = O(\log \log(1/\epsilon)),\\
     \epsilon_{\bar{\m{A}}} =&\; O(\epsilon),
     \label{eq:barAcost}
 \end{align}
 using a single application of $\m{U}_{\bar{\m{F}}_1}$, $\m{U}_{\bar{\m{F}}_2}$ and $\operatorname{polylog}(N)$ extra gates. 

Owing to the block-encoding constructions in Sec.~\ref{sec:block-encodingsF1F2}, this block-encoding of $\bar{\m{A}}$ is asymptotically optimal up to $\log(1/\epsilon)$ factors. To see this, recall that $\alpha_{\bar{\m{A}}} \geq \|\bar{\m{A}}\|$ by definition of block-encoding. Furthermore, since $\bar{\m{F}}_1$ is the top-left block of $\bar{\m{A}}$, Eqs.~\eqref{eq:F1tildelowerbound}-\eqref{eq:Frescalings} imply
\begin{align}
\alpha_{\bar{\m{A}}} \geq \|\bar{\m{A}}\| \geq \| \bar{\m{F}}_1\| = \|\tilde{\m{F}}_1\| = \Omega(\NF\sqrt{\NH}).     
\end{align} 
It follows that any block-encoding construction has a rescaling lower bounded as $\Omega(N_{\m{F}} \sqrt{N_\m{H}})$. Therefore our block-encoding of the Carleman matrix has asymptotically optimal scaling up to $\log(1/\epsilon)$ factors, and in particular it has optimal scaling in the resolution parameters $(\NF,\NH)$.

\emph{State preparation for $\ket{\m{z}_0}$}. We assume knowledge of~$\| \m{u}_0\|$, as well as a unitary involving $\operatorname{polylog}(N)$ gates preparing the band-limited initial condition
\begin{align}
\label{eq:Uu0}
   \m{U}_{\m{u}_0} \ket{0} = \ket{\m{u}_0},
\end{align}
where $\ket{\m{u}_0}$ denotes the quantum state whose amplitudes in the computational basis encode the entries of $\m{u}_0/\| \m{u}_0\|$. These correspond to the normalized Fourier-Hermite expansion of the initial condition. For our illustrative state preparations, the Fourier mode preparation is trivial, while the localized initial condition is a Gaussian profile (see Appendix~\ref{appendix:initialdata}), which can be efficiently prepared, for example, via rejection sampling techniques~\cite{lemieux2024quantum}.

Having access to the unitary in Eq.~\eqref{eq:Uu0}, we can apply the construction in the proof of Lemma~5 in Ref.~\cite{liu2021efficient} to show that the unitary $\m{U}_{\ket{\m{z}_0}}$ can be realized via $O(\NC \operatorname{polylog}(N))$ gates, which means $O(\log(1/\epsilon) \operatorname{polylog}(N))$ gates. 

\subsubsection{Coherent history-state output}

We showed how to rephrase the nonlinear Vlasov-Poisson equation as a linear set of ODEs \eqref{eq:CarlemanODE} in a larger space while incurring a controlled error. The next step is to construct a coherent encoding of the plasma dynamics over the time-interval $[0,T]$. Specifically, we output an $O(\epsilon)$--approximation to the Carleman history state
\begin{align}
\label{eq:historywitherror}
\ket{\chi} = \frac{1}{\mathcal{N}_\chi} \sum_{s=0}^{M-1} \|\bar{\m{z}}(s \Delta t)\| \ket{\bar{\m{z}}(s \Delta t)} \ket{s},
\end{align}
where the last register is a clock register, and $\mathcal{N}_\chi$ is a normalization factor. The Carleman history state encodes discrete snapshots of the plasma data (i.e., the plasma state in Fourier-Hermite basis, and its tensor powers up to~$\NC$) at $M$ times $t=0, t=\Delta t, \dots, t=M\Delta t = T$.  This state will allow us to extract an estimate of the target plasma observable for Problem~\ref{problem}.

The Carleman history state is obtained from the solution of a linear system of equations, derived from the linear Carleman ODE. This linear system of equations is obtained as follows. The solution to Eq.~\eqref{eq:CarlemanODE} can be formally written as 
\begin{equation}
\label{eq:ODEsolution}
    \bar{\m{z}}(t) = \ee^{\mathcal{\bar{\mathcal{A}}} t} \bar{\m{z}}^0.
\end{equation}
By discretizing time, we can recursively approximate the solution at time $t=m\Delta t$ advanced from that at time $t=(m-1)\Delta t$ via a Taylor series truncated at degree $k$:
\begin{equation}
\label{eq:recursive}
    \bar{\m{z}}(m\Delta t) \approx \bar{\m{z}}^m := \m{T}_k(\bar{\mathcal{A}}\Delta t) \bar{\m{z}}^{m-1}, \quad m=1, \ldots, M,
\end{equation}
where 
\begin{align}
\label{eq:truncation}
    \m{T}_k(x) = \sum_{j=0}^k \frac{x^j}{j!}.
\end{align} 
The above recursive relations can be encoded as a linear system of equations~\cite{berry2022quantum},
\begin{align}
\label{eq:linearsystemembedded}
    \m{L} \bar{\m{y}} = \bar{\m{b}},
\end{align}
where the matrix $\m{L}$ is given by
\begin{align}
\label{eq:L}
    \m{L} = \m{I} -\sum_{m=0}^{M-1} \ketbra{m+1}{m} \otimes \m{T}_k(\bar{\mathcal{A}}\Delta t).
\end{align}
The solution vector $\bar{\m{y}}$ and $\bar{b}$ have the structure
\begin{align}
\label{eq:Lextra}
\bar{\m{y}} = [ \bar{\m{z}}^0, \bar{\m{z}}^1, \dots, \bar{\m{z}}^m, \dots, \bar{\m{z}}^M], \quad 
    \bar{\m{b}} = [ \bar{\m{z}}^0, \underbrace{\m{0}, \dots,\m{0}}_{M \textrm{ times}}]. 
\end{align}
The quantities $\Delta t$ and $k$ are free parameters, and we now choose
\begin{enumerate}
    \item A timestep $\Delta t \leq 1/\alpha_{\bar{\m{A}}} \le 1/\|\bar{\m{A}}\|$.
    \item A Taylor series truncation $k$ that satisfies
    \begin{align}
    (k+1)! \geq \frac{M \ee^3 \|\bar{\m{z}}_0\|}{\epsilon}.
\end{align}
\end{enumerate}
By using Lemma 3.3 in Ref.~\cite{jennings2025quantum}, we have that $\|\bar{\m{z}}_0\| \geq \sup_{t \in [0,T]} \| \bar{\m{z}}(t)\|$. Hence, an application of  Lemma 3 in Ref.~\cite{jennings2023cost} gives that the chosen $k$ implies 
\begin{align}
\label{eq:timediscretization}
    \| \bar{\m{z}}^m - \bar{\m{z}}(m\Delta t) \| \leq \epsilon \quad \textrm{for every } m=0,\dots, M-1.
\end{align}
Therefore, the solution vector $\bar{\m{y}}$ provides an $O(\epsilon)$ accurate approximation to the history state $\ket{\chi}$ in Eq.~\eqref{eq:historywitherror}, which encodes the plasma dynamics.

The above linear system can be solved on a quantum computer using the linear solvers~\cite{costa2022optimal, jennings2025randomized, dalzell2024shortcut}, that output a state $O(\epsilon)$ close to the history state $\ket{\chi}$ at a cost of 
\begin{align}
\label{eq:linearsolvercost}
    \tilde{O}(\alpha_{\m{L}} \| \m{L}\|\| \m{L}^{-1}\| \log 1/\epsilon)
\end{align}
applications of (controlled) $\epsilon$-block encodings of the matrix $\m{L}$ with scaling factor $\alpha_{\m{L}}$, as well as (controlled) state preparations for~$\ket{\bar{\m{b}}}$. Also note that one can replace $\m{L}(\bar{\mathcal{A}})$ in Eq.~\eqref{eq:linearsolvercost} with the non-embedded linear system $\m{L}(\bar{\m{A}})$), see Appendix~\ref{sec:quantumODEsolver}. It remains to determine the gate-complexity of each component of the algorithm.

\emph{State preparation of $\ket{\bar{\m{b}}}$}: The state preparation boils down to the preparation of $\ket{\bar{\m{z}}_0}$. This involves $\operatorname{polylog}(N) \log\tfrac{1}{\epsilon}$ gates, as was discussed in the last section. 

\emph{Block-encoding of $\m{L}$}: The matrix $\m{L}$ can be constructed as a linear combination of unitaries:
\begin{align}
    \m{U}_{\m{L}} = \m{I} - \m{S} \otimes \m{U}_{\m{T}_k(\bar{\mathcal{A}}\Delta t)}.
\end{align}
Here $\m{S}$ is a quantum adder, adding $+1$ on a single qubit ancilla initialized in zero and the $\lceil \log(M)\rceil$ qubits of the clock register. This block-encodes $\sum_{m=0}^{M-1} \ketbra{m+1}{m}$ with rescaling  factor $1$, using a single ancilla qubit.  $\m{U}_{\m{T}_k(\bar{\mathcal{A}}\Delta t)}$ is instead a block-encoding of $\m{T}_k(\bar{\mathcal{A}}\Delta t)$, which can be realized via LCU using $\lceil \log k \rceil = \lceil \log \log (\alpha_{\bar{\m{A}}}T/ \epsilon) \rceil$ auxiliary qubits. The LCU involves 
\begin{align}
\label{eq:koverhead}
    k = O(\log (\alpha_{\bar{\m{A}}}T/ \epsilon) )
\end{align} 
applications of a block-encoding of $\bar{\mathcal{A}}$, which we have access to as shown in the last section. Note that the corresponding PREP of a state with amplitudes proportional to $1/j!$ can be realized with at most $O(\log (\alpha_{\bar{\m{A}}}T/ \epsilon))$ gates. The block-encoding prefactor of $\m{L}$ in this construction is 
\begin{align}
\label{eq:omegaL}
\alpha_{\m{L}} =O(\ee^{\|\bar{\mathcal{A}}\| \Delta t}) = O(1)
\end{align}

\emph{Bounding $\| \m{L}^{-1}\|$}: We show in Appendix~\ref{sec:quantumODEsolver} that, as long as $\mu(\bar{\m{A}})<0$ we have
\begin{align}
\label{eq:Linversebound}
  \| \m{L}^{-1}\| = O\left(\frac{\alpha_{\bar{\m{A}}}}{-\mu(\bar{\m{A}})}\right).   
\end{align}
Note the absence of a dependence on $T$, as expected from fast-forwarding results for stable systems~\cite{jennings2023cost, an2026fast}. Also note that we showed in Eq.~\eqref{eq:Carlemanlognorm} that $-\mu(\bar{\m{A}})$ is positive and bounded away from zero under our $R_{\m{P}}$-number condition. 

Now combine the cost in Eq.~\eqref{eq:linearsolvercost} with the estimates~\eqref{eq:omegaL}-\eqref{eq:Linversebound} and the overhead~\eqref{eq:koverhead}. We get that the number of times we need to apply the block-encodings of $\bar{\mathcal{A}}$ and the unitary preparing $\ket{{\bar{\m{z}}_0}}$ is 
\begin{align}
    \tilde{O}\left(\alpha_{\bar{\m{A}}} \log \left(\frac{\alpha_{\bar{\m{A}}}T}{ \epsilon}\right) \log \frac{1}{\epsilon} \right).
\end{align}
Using the results in Eq.~\eqref{eq:barAcost}, and including the gate overheads~\cite{berry2022quantum}, the gate complexity is
\begin{align}
    \tilde{O}\left(\NF\NH^{1/2}\operatorname{polylog} \left(\frac{T}{ \epsilon}\right) \operatorname{polylog}\left(\frac{1}{\epsilon}\right)\right).
\end{align}
The total number of qubits required is dominated by the encoding of the linear system to be solved
\begin{align}
    \tilde{O}(\log(\NF \NH^{1/2} T) \log(1/\epsilon)).
\end{align}

\subsubsection{Information-extraction from the history state}

To provide a solution to the benchmark Problem~\ref{problem}, we now estimate the time-averaged kinetic energy $\langle \mc{K} \rangle$ over the time interval $[0,T]$. The quantum algorithm outputs an approximate history state, which encodes discrete data sampled at a finite number $M$ of time-steps, where $T=M\Delta t$, and the time-step $\Delta t$ obeys $\alpha_{\bar{\m{A}}}\Delta t <1$. Taking a uniform sample of the kinetic energy over these $M$ points provides a convergent estimate of $\langle \mc{K} \rangle$ as $M\rightarrow \infty$. However, this first-order quadrature incurs an additional $O(1/\epsilon)$ cost. To avoid this, we instead use a local Taylor approximation for the integral of the full Carleman ODE solution over a small interval. This avoids any additional factors of $1/\epsilon$ in the overall algorithm complexity and gives an approximation $\overline{\mc{K}}_k$ for any order--$k$ Taylor truncation. 

In Appendix~\ref{app:quadrature-error} we show that for an order--$k$ Taylor truncation, the quadrature $\overline{\mc{K}}_k$ approximates the target quantity with error
\begin{equation}
    \left|\langle \mc{K} \rangle-\overline{\mc{K}}_{k}\right| =  O(\epsilon),
\end{equation}
when $k = \Theta ( \log (1/\epsilon)/\log\log(1/\epsilon))$. This scaling coincides with the constraints of the quantum subroutine to generate the history state. It remains to be shown how to output an estimate of $\overline{\mc{K}}_k$ from the coherent quantum output. 

The quantum algorithm outputs an $\epsilon$--approximation to $\ket{\chi}$ in Eq.~\eqref{eq:historywitherror}, which in turn encodes $\bar{\m{z}}(t)$ over the $M$ discrete times in $[0,T]$. The data $\bar{\m{z}}(t)$ in turn provide an $\epsilon$--approximation to the exact dynamics of the Galerkin system, and from this we wish to estimate an observable defined in the original plasma variables. To realize this, we transform the observable into the Lyapunov frame (via the construction from Sec.~\ref{sec:block-encodingsF1F2}). However, we do not post-select on the single-copy data, but instead work on the full history state $\ket{\chi}$, which provides a sharper estimate of the integrated single-copy observable.

In Appendix~\ref{app:quadrature-error}, we show that for any observable $\mc{K}(t)$ that can be expressed as a linear functional of the vector $\bar{\m{z}}(t)$ in the form $\bar{\ell}_{\mc{K}}^\dagger \bar{\m{z}}(t)$, we may define a quadrature approximation of the time average of $\mc{K}(t)$ over $[0,T]$ of the form
\begin{equation}
    \overline{\mc{K}}_k := \bar{\ell}_{\mc{K},k}^\dagger \left (\frac{1}{M}\sum_{s=0}^{M-1}\bar{\m{z}}(sh) \right ),
\end{equation}
where the vector $\bar{\ell}_{\mc{K},k}$ is given by
\begin{equation}
    \bar{\ell}_{\mc{K},k} \coloneqq \m{R}_k(\bar{\m{A}}\Delta t)^\dagger \bar{\ell}_{\mc{K}},
\end{equation}
and 
\begin{equation}
    \m{R}_k(\bar{\m{A}}\Delta t) := \sum_{q=0}^{k} \frac{ (\bar{\m{A}}\Delta t)^q}{(q+1)!}
\end{equation}
is a truncated Taylor term that approximates the integrator on a single time-step of size $\Delta t$. In particular, in the context of the Carleman embedding the approximate integrator does not just have support on the single-copy sector but couples this sector to nonlinear terms. Because of these nonlinear couplings we improve over sampling simply on the single-copy sector, and $\m{R}_k$ rapidly converges to the correct integrator as we increase the truncation scale. We obtain an $O(\epsilon)$ error when $k=\tilde{\Theta}(\log (1/\epsilon))$.

In the case of the kinetic energy of the plasma perturbation, the details of the vector $\bar{\ell}_{\mc{K},k}$ are as follows. The Lyapunov-transform for the exact dynamics is defined as
\begin{equation}
    \bar{\m u}(t)
    =
    \gamma \m{P}^{1/2}\m u(t),
    \qquad
    \gamma
    :=
    \frac{1+R_{\rm thr}}
    {2\|\m{u}_0\|_2\sqrt{1+1/\tau}}.
\end{equation}
Therefore, we have $\m{u} =\gamma^{-1} \m{P}^{-1/2}\bar{\m u}$, and in what follows we assume $\m{P}^{-1/2}$ is extended to the full Carleman space, where it only acts non-trivially on the single-copy sector. Thus the kinetic-energy functional on the full Carleman space is given by the vector
\begin{equation}
    \bar{\ell}_{\mc K} = \gamma^{-1}\, \m{P}^{-1/2}\ell_{\mc K}.
\end{equation}
and in the single-copy sector the vector $\ell_{\mc K}$ is given by (see Appendix~\ref{app:average-kinetic-energy})
\begin{equation}
     (\ell_{\mc{K}})_{\v{n},\v{\alpha}}=\frac{1}{\sqrt{2}L^{3/2}}\v 1_{\{\v{n}=\v{0}\}}\sum_{a=1}^3\v 1_{\{\v{\alpha}=2\v{e}_a\}} .
\end{equation}
Since \(P^{-1/2}\) acts nontrivially only on the density Hermite modes
\(\v\alpha=\v 0\), it acts trivially on this vector. Therefore,
\[
    \|\bar{\ell}_{\mc K}\|
    =
    \frac{\sqrt3}{ \sqrt{2} \gamma L^{3/2}} .
\]

To obtain $\overline{\mc{K}}_k$ we estimate a state overlap between $\ket{\chi}$ and a state $\ket{\overline{\mc{K}}_k}$, defined as
\begin{equation}
    \ket{\overline{\mc{K}}_k} := \ket{\bar{\ell}_{\mc{K},k}} \otimes \frac{1}{\sqrt{M}}\sum_{s=0}^{M-1} \ket{s}
\end{equation}
where $\ket{\bar{\ell}_{\mc{K},k}} = \bar{\ell}_{\mc{K},k} / \|\bar{\ell}_{\mc{K},k}\|$. Therefore, 
\begin{align}
    \braket{\overline{\mc{K}}_k}{\chi} = \frac{\sqrt{M}}{\|\bar{\ell}_{\mc{K},k}\| \mathcal{N}_\chi} \overline{\mc{K}}_k.
\end{align}
This implies that an estimate of the overlap $\braket{\overline{\mc{K}}_k}{\chi }$ and the normalization $\mathcal{N}_\chi$ together yield an estimate of $\overline{\mc{K}}_k$, and thus $\langle \mc{K} \rangle$. 

We already have an efficient quantum circuit $\m{U}_\chi$ to output an $O(\epsilon)$ approximation to the coherent state $\ket{\chi}$, and so it remains to efficiently prepare the state $\ket{\overline{K}_k}$. The vector $\bar{\ell}_{\mc{K}}$ is sparse and its components readily computed with coherent arithmetic to yield a state $\ket{\bar{\ell}_{\mc{K}}} = \bar{\ell}_{\mc{K}} / \|\bar{\ell}_{\mc{K}}\|$. We construct $\ket{\bar{\ell}_{\mc{K},k}}$ by implementing a unitary block-encoding of $\m{R}_k(\bar{\m{A}}\Delta t)^\dagger$ on $\ket{\bar{\ell}_{\mc{K}}}$ to give us $\ket{\bar{\ell}_{\mc{K},k}}$ on post-selection of success and $\ket{\overline{\mc{K}}_k}$ follows with an additional uniform state preparation on the time register. To achieve this, in Appendix~\ref{app:state-prep-ell} we show that a unitary block-encoding for $\m{R}_k(\bar{\m{A}}\Delta t)$ can be constructed with $\tilde{O}(\log(1/\epsilon))$ calls to the block-encoding for $\bar{\m{A}}$. We then show that the state preparation of $\ket{\overline{\mc{K}}_k}$ can be achieved with success probability $1-\delta$, using fixed-point amplitude amplification~\cite{yoder2014fixed} and $O(\log(1/\delta))$ calls to the block-encoding for $\m{R}_k(\bar{\m{A}}\Delta t)$. This provides reliable state preparation of $\ket{\overline{\mc{K}}_k}$ via a circuit $\m{U}_{\overline{\mc{K}}}$.

The state overlap $\braket{\overline{\mc{K}}_k}{\chi}$ can then be estimated to additive error $O(\epsilon)$ via amplitude estimation~\cite{brassard2000quantum}, and uses $O(1/\epsilon)$ calls to the unitaries $\m{U}_{\overline{\mc{K}}}$ and $\m{U}_{\chi}$ and their adjoints. The normalization $\mathcal{N}_\chi$ can be estimated to lie inside an interval $[(1+\epsilon)^{-1}\mathcal{N}_\chi, (1+\epsilon)\mathcal{N}_\chi] $ using the method given in~\cite{dalzell2024shortcut}, and requires $O(1/\epsilon)$ calls to $\m{U}_{\m{L}}$. This, together with the quadrature error bound and the fact that $\gamma = \Omega(1)$, implies an $O(\epsilon)$ estimate of $\langle \mc{K}\rangle$ with complexity $\tilde{O}(\NF \NH^{1/2}/\epsilon)$ as claimed.

 \section{Conclusions and outlook}

We have presented an end-to-end quantum algorithm for a weakly nonlinear kinetic plasma model, including a priori convergence guarantees, block-encodings of the relevant operators, and a concrete observable-extraction procedure. The model is deliberately simplified, but it contains two obstacles that are central to quantum algorithms for nonlinear plasma physics: the need to embed nonlinear dynamics into a linear quantum computation with controlled error, and the need to encode dense nonlocal field interactions without losing the putative speedup. The first obstacle is addressed through a plasma free energy Lyapunov transform and a corresponding Carleman convergence criterion, building on the broader theory of Carleman embeddings and quantum algorithms for nonlinear differential equations~\cite{forets2017explicit,liu2021efficient,krovi2022improved, jennings2025quantum}. The second is addressed through a hierarchical block-encoding~\cite{nguyen2022block} that converts spatial decay of the screened interaction into an improved normalization. Together these ingredients yield an exponential memory saving and a superquadratic improvement in time relative to a Fourier--Hermite spectral solver for the same truncated problem~\cite{canuto2006spectral}.

This work should be viewed as an initial step rather than a complete algorithm for practical plasma simulation. The physical model considered here is electrostatic, unmagnetized, weakly nonlinear, near Maxwellian, and uses a simple relaxation operator. Many important plasma regimes require additional structure: magnetic geometry, electromagnetic fluctuations, multiple kinetic species, more realistic collision operators, sources and sinks, boundaries, sheaths, and stronger turbulent nonlinearities~\cite{brizard2007foundations,krommes2012gyrokinetic,dimits2000comparisons,jenko2000electron,stangeby2000plasma,helander2005collisional}. Incorporating these ingredients is essential for determining whether the mechanisms identified here can survive in models closer to those used in fusion, space, and laboratory plasma physics.

A central open problem is to enlarge the certified nonlinear regime that can be resolved by our algorithm. The Lyapunov $R$-number condition $R_{\m{P}}<1$ used here is a sufficient condition; it should not be interpreted as the true physical or algorithmic boundary of applicability. It certifies convergence for a rigorously controlled weakly nonlinear window, but it may be conservative for at least three reasons: it uses worst-case operator-norm estimates, it does not exploit detailed structure of particular initial data, and it treats transient growth through global stability bounds rather than trajectory-dependent estimates. Future work could extend the range of amplitudes that can be simulated by constructing sharper Lyapunov functionals, problem-specific pseudospectral bounds, or numerical methods. Another route might be via nonlinear interaction picture methods~\cite{jennings2026quantum}. Determining this practical nonlinear threshold is essential for assessing whether the approach can move from a rigorous perturbative benchmark toward plasma regimes of direct physical interest.

There is also substantial algorithmic work left to do. The present algorithm estimates a scalar spacetime-averaged observable; more general diagnostic observables may require new state-preparation and data-extraction methods. It will also be important to sharpen the block-encoding normalizations, reduce precision overheads, understand the practical constants in the Carleman truncation, and identify broader classes of plasma nonlinearities for which Lyapunov-controlled embeddings are efficient. Another important direction is to compare against more refined classical algorithms, including adaptive spectral methods, particle methods, semi-Lagrangian schemes, and fast algorithms for nonlocal field solves~\cite{canuto1988spectral,cheng1976integration,filbet2001conservative,greengard1987fast}.

The broader lesson is that quantum advantage for plasma simulation cannot be established by dimension counting alone. It depends on the choice of representation, the stability structure of the equations, the cost of loading physical operators, and the cost of extracting useful classical information. The Fourier--Hermite representation sets a strong classical baseline, and the results here show that superquadratic quantum improvements can nevertheless persist for a controlled nonlinear kinetic problem. Extending this conclusion to richer and more realistic plasma physics problems remains an exciting open challenge.

\bigskip

{\textbf{Author contributions:} Authors are listed alphabetically. BB performed the stability analysis and contributed to the development of the Vlasov-Poisson model. DJ developed and implemented the hierarchical block-encoding theory, and assisted with the stability analysis, model development and information extraction. ML developed a preliminary approach to the Vlasov-Poisson equations via Carleman linearization and assisted with the stability analysis, hierarchical block-encoding and information extraction. SP provided theoretical support in choosing the plasma model. BB, DJ, ML, and SP wrote the paper.
\bigskip

\textbf{Acknowledgements:} We thank Carys Harvey for helpful observations on the hierarchical block-encoding theory. We thank Angus Kan, Cristian Cortes and Sukin Sim for helpful discussions on the compilation of quantum algorithms for plasma. We also thank Ryoji Anzaki for discussions on plasma dynamics and on the physical consistency and approximations of the Carleman embedding method. We also thank Ilon Joseph, Frank Graziani, Marta Mauri, Tamás Vaszary and Tomohisa Nagata, for many helpful discussions on plasma physics.

\bibliographystyle{apsrev4-2}
\bibliography{ Bibliography}

\clearpage

\appendix

\onecolumngrid

\phantomsection
\pdfbookmark[1]{Appendices}{appendices}
\begin{center}
{\Large\bfseries Appendices\par}
\end{center}

\begingroup
\appendixonlytocfalse
\tableofcontents
\endgroup

\appendixonlytoctrue

\clearpage

\makeatletter
\renewcommand*{\theHequation}{appendix.\Alph{section}.\arabic{equation}}
\renewcommand*{\theHparentequation}{appendix.\Alph{section}.\arabic{parentequation}}
\makeatother

\section{Notation}
\label{app:notation}

\begin{table}[h!]
\centering
\setlength{\tabcolsep}{8pt}
\renewcommand{\arraystretch}{1.0}
\begin{tabular}{cll }
\toprule
& \textbf{Symbol} & \textbf{Description} \\
\midrule
\multirow{11}{*}{\rotatebox{90}{\textbf{Plasma system}}}
& \v{x} & Position in 3D space \\
& \v{v} & Velocity in 3D space \\
&   $t$             & Time (generic) \\
& $f(\v{x},\v{v},t)$ & Ion phase space distribution function \\
& $f_0(\v{x},\v{v},t)$ & Ion Maxwellian phase space distribution \\
& $h(\v{x},\v{v},t)$ & Ion phase space  perturbation about Maxwellian background, $f-f_0$ \\
& $g(\v{x},\v{v},t)$ & Ion phase space relative perturbation about Maxwellian background, $(f-f_0)/f_0$ \\
&   $q_e$             & Electron charge  \\
&   $\phi(\v{x},t)$         &Electrostatic potential \\ 
&   $\kb$   & Boltzmann's constant \\
&   $\epsilon_0$   & Permittivity of free space \\
&   $T_e$             & Electron temperature \\
& $\nu$ & Collisional frequency \\
&   $\varphi(\v{x},t)$             & Ratio of electrostatic potential energy to electron thermal energy $|q_e|\phi(\v{x},t)/(k_B T_e)$ \\
&   $T_{\mathrm{ion}}$         & Ion temperature \\
& $\tau$ & Ion to electron temperature, $T_{\mathrm{ion}}/T_e$  \\
&   $m$           & Ion mass \\
&   $n_0$             & Ion equilibrium density \\
&   $n_{e,0}$             & Electron equilibrium density \\
& $
\lambda_{\mathrm{D},e}$ & Electron Debye length, $\sqrt{\frac{\epsilon_0 k_{\mathrm{B}}T_e}{n_{e,0}q_e^2}}$ \\
& $\ldion$ & Ion Debye length, $\sqrt{\frac{\epsilon_0\kb T_{\mathrm{ion}}}{n_{0}q_e^2}}$ \\
& $\vth$ & Ion thermal velocity, $\sqrt{\frac{\kb T_{\mathrm{ion}}}{m}}$ \\
\midrule
\multirow{16}{*}{\rotatebox{90}{\textbf{Simulated plasma system}}}
& $T$  & Total simulation time \\
& $L$ & Size of the cubic simulation box \\
& $\mathbb{T}^3_L$
&  Cubic simulation box with periodic boundary conditions (3-torus) \\ 
&$\lvert \v{x}\rvert$  &  Euclidean norm of a three-vector $\vx$ \\
&$\NF$   &           The Galerkin approximation projects onto $(2\NF+1)^3$ Fourier modes \\
&   $\NH$ & The Galerkin approximation projects onto $(\NH+1)^3$ Hermite modes \\
&   $N$ & Total number of modes, $N= (2\NF+1)^3(\NH +1)^3$ \\
&   $\langle \mathcal{K} \rangle $       & Time and space-averaged kinetic energy moment of $g$ \\
&  $\epsilon$ & Additive error tolerance for the estimation of $\langle \mathcal{K} \rangle $ \\
&$\bar{\vx}$ & Position in natural dimensionless units, $\frac{1}{\ldion}\vx$ (bar omitted) \\
& $
\bar{\vv}$ & Velocity in natural dimensionless units, $\frac{1}{\vth}\vv$ (bar omitted) \\
& $\bar{t}$ & Time in natural dimensionless units, $\frac{\vth}{\ldion}t$ (bar omitted) \\
&$\bar{f}$ & Ion distribution function in natural dimensionless units, $\frac{\vth^3}{n_0}f$ (bar omitted) \\
& $\bar{\phi}$ & Electrostatic potential in natural dimensionless units, $\frac{\lvert q_e\rvert}{\kb T_{\mathrm{ion}}}\phi=\frac{1}{\tau}\varphi$ (bar omitted) \\
& $\bar{L}$ & Size of the cubic simulation box in natural dimensionless units, $\frac{1}{\ldion}L$ (bar omitted)\\
& $\bar{\nu}$ & Collisional rate in natural dimensionless units, $ \frac{\nu\ldion}{\vth}$ (bar omitted) \\
& $w(\v{v})$ & Maxwell distribution in natural dimensionless units, $(2\pi)^{-3/2} \exp(-|\v{v}|^2/2)$ \\
& $\v{\xi}_{\vn}$ & Fourier mode wavevector, $ \frac{2\pi}{L}\vn$ for Fourier mode label $\vn \in \mathbb Z^3$  \\
& $\kappa_{\v{n}} = |\v{\xi}_{\v{n}}|^2 + \tau$ & Electrostatic spectral amplitude in Eq.~\eqref{eq:g_mode_hermite} \\
& $g_{\v{n},\v{\alpha}}$ ($\tilde{g}_{\v{n},\v{\alpha}}$) & (Lyapunov-transformed) Ion phase space relative perturbation, Fourier-Hermite component \\
\midrule
\multirow{13}{*}{\rotatebox{90}{\textbf{Quantum algorithm}}}
&   $\NC$               & Carleman truncation order \\
&   $\epsilon_{\mathrm{C}}$        & Carleman truncation error\\
&   $\bar{\m{u}}(t)$, $\bar{\m{u}}^0$ & Vector $\tilde{g}_{\v{n},\alpha}(t)$, suitably rescaled by a factor $\gamma$. Same vector at time $t=0$ \\
&   $\m{F}_1,\m{F}_2$           & Linear and quadratic plasma terms \\
&   $\bar{\m{A}}$, $\bar{\mathcal{A}}$             & Carleman matrix and its square embedding \\
& $\bar{\m{z}}(t)$, $\bar{\m{y}}$ & Carleman vector and Carleman history state \\
& $\Delta t, k$ & Time discretization time step and order of Taylor truncation \\
& $\m{L}$ & Linear system problem constructed from time discretization of the Carleman ODE system \\
&  $\|\m{X}\|_2$, $\| \m{X}\|_{\mathrm{max}}$, $\mu(\m{X})$          & Operator norm, max element and logarithmic norm of matrix $\m{X}$ \\
&   $\alpha_{\m{X}}$          & Block-encoding prefactor of a matrix $\m{X}$ \\
&   $R_{\mathrm{P}}$      & Generalized Lyapunov $R$-number \\
\bottomrule
\end{tabular}
\end{table}

\section{Conventions and definitions}
Here we set out some conventions and theory for our analysis. For spatial degrees of freedom, we work on the $3$-torus $\mathbb{T}_L^3\coloneqq \R^3/(L\mathbb{Z})^3$. On $\mathbb{T}_L^3$, we consider the class of functions 
$$
L^2(\TL^3)\coloneqq \Bigg\{f:\TL^3\to\mathbb{C}: \int_{\TL^3} \lvert f(\vx)\rvert^2\,\mathrm{d}^3 x<\infty  \Bigg\},
$$
namely the space of square-integrable functions on the torus. The Fourier expansion of any element $f$ of this set is given by
\begin{equation}\label{eq:ftofn}
    f(\v{x}) = \frac{1}{L^{3/2}} \sum_{\vn \in \mathbb{Z}^3} \hat{f}_{\vn} \ee^{i \xi_{\vn} \cdot \v{x}},
\end{equation}
where $\v{\xi}_{\vn} := (2 \pi/L ) \vn$, and $\vn \coloneqq (n_1,n_2,n_3)$ with components taking values in the integers. The Fourier coefficients in \eqref{eq:ftofn} can be obtained from
\begin{equation}
\label{eq:Fourier_coeffs}
    \hat{f}_{\vn} \coloneqq \frac{1}{L^{3/2}} \int_{\mathbb{T}_L^3}  f(\v{x}) \ee^{-i\v{\xi}_{\vn} \cdot \v{x}} \,\mathrm{d}^3 x.
\end{equation}
The periodic delta function is given explicitly by
\begin{equation}
    \delta_L^{(3)} (\v{x} - \v{x}') \coloneqq \frac{1}{L^3} \sum_{\vn \in \mathbb{Z}^3} \ee^{i \xi_{\vn} \cdot (\v{x}-\v{x}')}.
\end{equation}
Parseval's identity tells us that
\begin{equation}
    \int_{\mathbb{T}_L^3}  \lvert f(\v{x})\rvert^2\,\mathrm{d}^3 x = \sum_{\vn \in \mathbb{Z}^3} |\hat{f}_{\vn}|^2 .
\end{equation}
The limit $L\rightarrow \infty$ corresponds to
\begin{equation}
    \frac{1}{L^3} \sum_{\vn \in \mathbb{Z}^3} \rightarrow \int_{\R^3} \frac{\mathrm{d}^3\xi}{(2\pi)^3},
\end{equation}
with $L^{-3/2}\hat{f}_{\vn} \rightarrow \hat{f}(\v{k})$. This gives rise to the infinite space Fourier analysis
\begin{align}
    f(\v{x}) =&\; \int \frac{\mathrm{d}^3 \xi}{(2\pi)^3} \hat{f}(\v{\xi}) \ee^{i \v{\xi}\cdot \v{x}} ,\nonumber \\
        \hat{f}(\v{\xi}) =&\; \int \frac{\mathrm{d}^3x}{(2\pi)^3} f(\v{x}) \ee^{-i \v{\xi}\cdot \v{x}}.
\end{align}
We also consider truncating the infinite Fourier series to the set $\mc{Z}_{\NF} \coloneqq \{-\NF, \dots, \NF\}$ to obtain band-limited functions. We define $N_x := 2\NF+1$ as the total number of modes per dimension. This restricts us to a finite-dimensional subspace $L^2_{N_x}(\mathbb{T}_L^3)$ of $L^2(\mathbb{T}_L^3)$ functions on the torus, with finite number of Fourier modes, defined as
\begin{equation}
    L^2_{N_x}(\mathbb{T}_L^3) := \Bigg\{ f \in L^2 (\mathbb{T}_L^3) : f(\v{x}) = \frac{1}{L^{3/2}}\sum_{\vn \in \mc{Z}_{\NF}^3} \hat{f}_{\vn} \ee^{\frac{2\pi}{L} i \vn \cdot \v{x} } \Bigg\}.
\end{equation}
We also define $\Pi_{\NF}$ to be the projector onto this subspace, so that
\begin{equation}
    f(\v{x}) =  \frac{1}{L^{3/2}}\sum_{\vn \in \mathbb{Z}^3} \hat{f}_{\vn} \ee^{\frac{2\pi}{L} i \vn \cdot \v{x} } \mapsto (\Pi_{\NF} f) =  \frac{1}{L^{3/2}}\sum_{\vn \in \mc{Z}_{\NF}^3} \hat{f}_{\vn} \ee^{\frac{2\pi}{L} i \vn \cdot \v{x} }.
\end{equation}

The dimension of the truncated subspace is $\mathrm{dim}(L^2_{N_x}) = N_x^3$. The elements of this set are functions on $\mathbb{T}_L^3$, however we can restrict the domain to the uniform grid $\{\v{x} \in \mathbb{T}_L^3 : \v{x} = \frac{L}{N_x} \v{j}, \quad \v{j} \in \mathbb{Z}_{N_x}^3 \}$. With this restricted domain we now consider
\begin{equation}
    L^2_{N_x}(\mathbb{Z}_{N_x}^3) := \left \{ (f_{\v{j}})= \left (f\left (\frac{L}{N_x}\v{j} \right ) \right) : f_{\v{j}} = \left (\frac{N_x}{L} \right )^{3/2} \frac{1}{N_x^{3/2}}\sum_{\vn \in \mc{Z}_{\NF}^3} \hat{f}_{\vn} \ee^{\frac{2\pi}{N_x} i \vn \cdot \v{j} } \right \},
\end{equation}
where it is understood in the notation on the left that the finite lattice on the torus is isomorphic to $\mathbb{Z}_{N_x}^3$.
This set is therefore the set of finite data vectors obtained by \emph{sampling} the band-limited $L^2$ functions on the torus to the above uniform grid. The sum inside the expression for $f_{\v{j}}$ is given in terms of the characters of the finite group $\mathbb{Z}_{N_x}^3$, and therefore the data $\hat{f}_{\vn}$ can be obtained from $f_{\v{j}}$ via the \emph{discrete} Fourier transform
\begin{equation}
   \hat{f}_{\vn} = \left (\frac{L}{N_x} \right )^{3/2}  \frac{1}{N_x^{3/2}} \sum_{\v{j} \in \mathbb{Z}_{N_x}^3} \ee^{-\frac{2\pi}{N_x} i \vn \cdot \v{j} } f_{\v{j}},
\end{equation}
where we note the additional prefactor required to match data. However, note that $\hat{f}_{\vn}$ are the Fourier data defined from the Fourier series on the continuum torus, and so we can equally define $\tilde{f}_{\vn} := (N_x/L)^{3/2} \hat{f}_{\vn}$, for the Fourier data in the discrete setting, to avoid such prefactors. Performing this sampling on a general $L^2$ function can give rise to aliasing effects, which we address later in Appendix~\ref{app:BE-of-nonlinear-term}.

\newpage

\section{Galerkin approximation for the plasma system}

This appendix derives the Fourier--Hermite Galerkin ODE and the Lyapunov-transformed matrices used in the main text.

\subsection{Fourier--Hermite basis transformation}
\label{app:FHbasis}

In this section we present the matrices that define the plasma dynamics in the Fourier--Hermite basis. This is central to our work, and the system is reduced to a finite-dimensional problem by truncating the Fourier modes to some maximum $\NF$ for each direction and the Hermite modes to some maximum order $\NH$ in each direction~\cite{canuto2006spectral}.

Start from Eqs.~\eqref{eq:gnonlinear}--\eqref{eq:phi_g_single_species}. By Fourier transforming in the spatial coordinates, Eq.~\eqref{eq:phi_g_single_species} becomes
\begin{equation}
\label{eq:helmoltzgfourier}
\kappa_{\vn}\hat\phi_{\vn}
=
\int_{\R^3}\hat g_{\vn}(\vv,t)w(\vv)\,\mathrm d^3v.
\end{equation}
Using the Hermite expansion \eqref{eq:Hermite_g} and the orthonormality of the basis defined in Eq.~\eqref{eq:hermite_basis}, we have
\begin{equation}
g_{\vn,\v 0}(t)
=
\int_{\R^3}\hat g_{\vn}(\vv,t)w(\vv)\,\mathrm d^3v.
\end{equation}
Hence, Eq.~\eqref{eq:helmoltzgfourier} becomes
\begin{equation}
\kappa_{\vn}\hat\phi_{\vn}=g_{\vn,\v 0}.
\end{equation}
Thus, the field is determined solely by the $\v{\alpha}=\v 0$ Hermite coefficient.

Let $Q(g)$ denote the sum of the nonlinear terms in \eqref{eq:gnonlinear},
\begin{equation}
Q(g)\coloneqq -(\vv\cdot\nabla_x\phi)g+\nabla_x\phi\cdot \nabla_v g.
\end{equation}
Then,
\begin{align}
\label{eq:Qhat}
\widehat{Q(g)}_{\vn}
=&\; -\frac{1}{L^{3/2}}
\sum_{\substack{\vn',\vn''\in \mathbb{Z}^3 \\ \vn'+\vn''=\vn}}
\vv\cdot (i\v{\xi}_{\vn'})\hat{\phi}_{\vn'}\hat{g}_{\vn''}
+\frac{1}{L^{3/2}}
\sum_{\substack{\vn',\vn''\in \mathbb{Z}^3 \\ \vn'+\vn''=\vn}}
(i\v{\xi}_{\vn'}\hat{\phi}_{\vn'})\cdot\nabla_v\hat{g}_{\vn''}
\nonumber\\
=&\; \frac{1}{L^{3/2}}
\sum_{\substack{\vn',\vn''\in \mathbb{Z}^3 \\ \vn'+\vn''=\vn}}
\sum_{a=1}^3 i\xi_{\vn',a}\hat{\phi}_{\vn'}
\big(-v_a+\partial_{v_a}\big)\hat{g}_{\vn''}.
\end{align}

Hence, Eq.~\eqref{eq:gnonlinear} in Fourier space reads
\begin{align}
\label{eq:ghat_mode_eq_eliminated}
\partial_t \hat g_{\vn}
& =
-i(\v{\xi}_{\vn}\cdot \vv)\hat g_{\vn}
-i(\v{\xi}_{\vn}\cdot \vv)\hat\phi_{\vn}
-\bar{\nu} \hat g_{\vn}
+\widehat{Q(g)}_{\vn}
\nonumber\\
& =
-i(\v{\xi}_{\vn}\cdot \vv)\hat g_{\vn}
-i\frac{\v{\xi}_{\vn}\cdot \vv}{\kappa_{\vn}}
g_{\vn,\v 0}
-\bar{\nu} \hat g_{\vn}
+\widehat{Q(g)}_{\vn}.
\end{align}

We now project onto the Hermite basis of Eq.~\eqref{eq:hermite_basis}. The standard Hermite identities give
\begin{equation}
v_a H_{\v{\alpha}}(\vv)
=
\sqrt{\alpha_a+1}\,
H_{\v{\alpha}+\v e_a}(\vv)
+
\sqrt{\alpha_a}\,
H_{\v{\alpha}-\v e_a}(\vv),
\end{equation}
and
\begin{equation}
\partial_{v_a}H_{\v{\alpha}}(\vv)
=
\sqrt{\alpha_a}\,
H_{\v{\alpha}-\v e_a}(\vv),
\end{equation}
with the convention $H_{\v{\alpha}-\v e_a}=0$ if $\alpha_a=0$. Therefore,
\begin{equation}
\big(-v_a+\partial_{v_a}\big)H_{\v{\alpha}}
=
-\sqrt{\alpha_a+1}\,
H_{\v{\alpha}+\v e_a},
\end{equation}
and hence
\begin{align}
\label{eq:Hermiteladdercombined}
\big(-v_a+\partial_{v_a}\big)\hat{g}_{\vn}
=&\; -\sum_{\v{\alpha}\in \mathbb{N}_0^3}
\sqrt{\alpha_a+1}\,
g_{\vn,\v{\alpha}}
H_{\v{\alpha}+\v e_a}
\nonumber\\
=&\; -\sum_{\v{\alpha}\in \mathbb{N}_0^3}
\sqrt{\alpha_a}\,
g_{\vn,\v{\alpha}-\v e_a}
H_{\v{\alpha}}.
\end{align}

Putting \eqref{eq:Hermiteladdercombined} into \eqref{eq:Qhat} gives
\small
\begin{align}
\label{eq:Qhat2}
\widehat{Q(g)}_{\vn,\v{\alpha}}
=&\; -\frac{i}{L^{3/2}}
\sum_{\substack{\vn',\vn''\in\mathbb{Z}^3\\ \vn'+\vn''=\vn}}
\hat{\phi}_{\vn'}
\sum_{a=1}^3 \xi_{\vn',a}\sqrt{\alpha_a}\,
g_{\vn'',\v{\alpha}-\v e_a}
\nonumber\\
=&\; -\frac{i}{L^{3/2}}
\sum_{\substack{\vn',\vn''\in\mathbb{Z}^3\\ \vn'+\vn''=\vn}}
\frac{1}{\kappa_{\vn'}}g_{\vn',\v 0}
\sum_{a=1}^3 \xi_{\vn',a}\sqrt{\alpha_a}\,
g_{\vn'',\v{\alpha}-\v e_a}.
\end{align}
\normalsize

Projecting \eqref{eq:ghat_mode_eq_eliminated} onto the Hermite basis gives precisely Eq.~\eqref{eq:g_mode_hermite}. Setting $\m{u}_{\vn,\v{\alpha}}\coloneqq g_{\vn,\v{\alpha}}$, we may write the equations in vector form as
\begin{equation}
\dot{\m{u}}=\m{F}_1\m{u}+\m{F}_2(\m{u}\otimes \m{u}),
\end{equation}
where
\begin{align}
(\m{F}_1)_{\vn,\v\alpha|\vn',\v\alpha'}
\coloneqq &\;
-i\delta_{\vn,\vn'}
\sum_{a=1}^3 \xi_{\vn,a}
\big(
\sqrt{\alpha_a'+1}\,\delta_{\v\alpha,\v\alpha'+\v e_a}
+
\sqrt{\alpha_a'}\,\delta_{\v\alpha,\v\alpha'-\v e_a}
\big)
\nonumber\\
&\;
-i\frac{1}{\kappa_{\vn}}\delta_{\vn,\vn'}
\Bigg(\sum_{a=1}^3 \xi_{\vn,a}\delta_{\v\alpha,\v e_a}\Bigg)
\delta_{\v\alpha',\v 0}
-\bar{\nu}\delta_{\vn,\vn'}\delta_{\v\alpha,\v\alpha'},
\label{eq:A_NComponents}
\end{align}
and
\begin{align}
(\m{F}_2)_{\vn,\v{\alpha}|\vn',\v{\alpha}';\vn'',\v{\alpha}''}
=&\;
-\frac{i}{L^{3/2}}\frac{1}{\kappa_{\vn'}}
\delta_{\vn,\vn'+\vn''}\delta_{\v{\alpha}',\v{0}}
\sum_{a=1}^3\xi_{\vn',a}\sqrt{\alpha_a}\,
\delta_{\v{\alpha}'',\v{\alpha}-\v e_a}
\nonumber\\
=&\;
-\frac{i}{L^{3/2}}\frac{1}{\kappa_{\vn'}}
\delta_{\vn,\vn'+\vn''}\delta_{\v{\alpha}',\v{0}}
\sum_{a=1}^3\xi_{\vn',a}\sqrt{\alpha''_a+1}\,
\delta_{\v{\alpha},\v{\alpha}''+\v e_a},
\end{align}
for all $\v{n} ,\v{n}',\v{n}''\in \mathbb{Z}^3$ and $\v{\alpha},\v{\alpha}',\v{\alpha}'' \in \mathbb{N}_0$.

\subsection{Explicit form of the Lyapunov-transformed matrices}
\label{app:quadraticODErescaled}

Here we compute details related to the key matrices $\tilde{\m{F}}_1$ and $\tilde{\m{F}}_2$, which define the truncated plasma problem. In particular, we also compute sparsity and norm properties relevant for their block-encodings.

For the linear part,
\begin{align}
(\tilde{\m F}_1)_{\vn,\v\alpha|\vn',\v\alpha'}
=&\;
-i\delta_{\vn,\vn'}
\sum_{a=1}^3 \xi_{\vn,a}
\big(
\sqrt{\alpha_a'+1}\delta_{\v\alpha,\v\alpha'+\v e_a}
+
\sqrt{\alpha_a'}\delta_{\v\alpha,\v\alpha'-\v e_a}
\big)
\nonumber\\
&\;
-i\delta_{\vn,\vn'}
\bigg(\sqrt{1+\frac{1}{\kappa_{\vn}}}-1\bigg)
\sum_{a=1}^3 \xi_{\vn,a}
\big(
\delta_{\v\alpha,\v e_a}\delta_{\v\alpha',\v0}
+
\delta_{\v\alpha,\v0}\delta_{\v\alpha',\v e_a}
\big)
-\bar{\nu}\delta_{\vn,\vn'}\delta_{\v\alpha,\v\alpha'}.
\label{eq:F1tilde_components}
\end{align}

The sparsity parameters are obtained as follows. To compute the row-sparsity,
fix $\vn,\v\alpha$. Nonzero elements can occur only when $\vn'=\vn$ and
$\v\alpha'\in \{\v\alpha \pm \v e_a,\v\alpha\}$, subject to the Hermite cutoff.
This implies a row-sparsity of at most $7$. The column-sparsity is obtained by
fixing $\vn', \v\alpha'$, and in the same way is at most $7$.

The maximal normed element $\lVert {\tilde{\m{F}}}_1\rVert_{\max}$ is given by
\small
\begin{align}
 \! \!   \lVert{\tilde{\m{F}}}_1\rVert_{\max} &:= \max_{\vn,\v\alpha, \vn',\v\alpha'} \left | (\tilde{\m F}_1)_{\vn,\v\alpha|\vn',\v\alpha'} \right | \nonumber \\
   &= \max_{\vn,\v\alpha,\v\alpha'} \left | \sum_{a=1}^3 \xi_{\vn,a}
\big(
\sqrt{\alpha_a'+1}\delta_{\v\alpha,\v\alpha'+\v e_a}
+
\sqrt{\alpha_a'}\delta_{\v\alpha,\v\alpha'-\v e_a}
\big)+ \bigg(\sqrt{1+\frac{1}{\kappa_{\vn}}}-1\bigg)
\sum_{a=1}^3 \xi_{\vn,a}
\big(
\delta_{\v\alpha,\v e_a}\delta_{\v\alpha',\v0}
+
\delta_{\v\alpha,\v0}\delta_{\v\alpha',\v e_a}
\big)
+i\bar{\nu}\delta_{\v\alpha,\v\alpha'}\right | \nonumber \\
 &\le \max_{\vn,\v\alpha,\v\alpha'} \left [ 3 \xi_{\vn,1}
\big(
\sqrt{\alpha_1'+1}\delta_{\v\alpha,\v\alpha'+\v e_1}
+
\sqrt{\alpha_1'}\delta_{\v\alpha,\v\alpha'-\v e_1}
\big)+ 3\bigg(\sqrt{1+\frac{1}{\kappa_{\vn}}}-1\bigg)
 \xi_{\vn,1}
\big(
\delta_{\v\alpha,\v e_1}\delta_{\v\alpha',\v0}
+
\delta_{\v\alpha,\v0}\delta_{\v\alpha',\v e_1}
\big)
+\bar{\nu}\right ] \nonumber \\
&\le 6\frac{2\pi}{L}\NF\sqrt{\NH+1} + 6 \frac{2\pi}{L} \NF\max_{\vn}\bigg(\sqrt{1+\frac{1}{\kappa_{\vn}}}-1\bigg) +\bar{\nu}\nonumber \\
&\le \frac{12\pi}{L}\NF\sqrt{\NH+1} +  \frac{12\pi}{L} \NF\bigg(\sqrt{1+\frac{1}{\tau}}-1\bigg) +\bar{\nu}.
\end{align}
\normalsize
Thus, $\lVert {\tilde{\m{F}}}_1\rVert_{\max} = O(\frac{\NF}{L} \sqrt{\NH} + \bar{\nu})$.

We next turn to the nonlinear part, which is given by 
\begin{equation}
\label{eq:F2tildesum}
\tilde{\m{F}}_2=\tilde{\m{F}}_2^{(1)}+\tilde{\m{F}}_2^{(2)},
\end{equation}
where $\tilde{\m{F}}_2^{(1)}$ is the bulk term and $\tilde{\m{F}}_2^{(2)}$ is the first-shell correction. Their components are
\begin{align}
(\tilde{\m{F}}_2^{(1)})_{\vn,\v{\alpha}|\vn',\v{\alpha}';\vn'',\v{\alpha}''}
=&\;
-\frac{i}{L^{3/2}}
\frac{1}{\sqrt{\kappa_{\vn'}(1+\kappa_{\vn'})}}
\delta_{\vn,\vn'+\vn''}
\delta_{\v{\alpha}',\v{0}}
(1-\delta_{\v{\alpha}'',\v{0}})
\sum_{a=1}^3
\xi_{\vn',a}\sqrt{\alpha''_a+1}\,
\delta_{\v{\alpha},\v{\alpha}''+\ve_a}
\label{eq:F2tilde_bulk_components}
\end{align}
and
\begin{align}
(\tilde{\m{F}}_2^{(2)})_{\vn,\v{\alpha}|\vn',\v{\alpha}';\vn'',\v{\alpha}''}
=&\;
-\frac{i}{L^{3/2}}
\sqrt{\frac{\kappa_{\vn''}}{\kappa_{\vn'}(1+\kappa_{\vn'})(1+\kappa_{\vn''})}}
\delta_{\vn,\vn'+\vn''}
\delta_{\v{\alpha}',\v{0}}
\delta_{\v{\alpha}'',\v{0}}
\sum_{a=1}^3
\xi_{\vn',a}\,
\delta_{\v{\alpha},\ve_a}.
\label{eq:F2tilde_firstshell_components}
\end{align}
We first note how the magnitude of the components of $\tilde{\m{F}}_2$ behaves in $\vn'$. Specifically we have a factor $\xi_{\vn',a}/ \sqrt{\kappa_{\vn'}( 1+ \kappa_{\vn'})}$. This implies that the components scale like $1/|\vn'|$, and hence have a slow decay as we increase the size of the 3-d lattice of Fourier modes.

Let us next compute the column-sparsity of $\tilde{\m{F}}_2$. We see that for fixed $\vn',\v{\alpha}';\vn'',\v{\alpha}''$ nonzero elements occur for $\vn = \vn'+\vn''$, $\v\alpha \in \{\v\alpha'' + \v{e_a}, \v{e_a}\}$. A careful examination of cases implies that it is $s_{\mathrm{c}}=3$ column-sparse.

The row-sparsity requires more care. We now fix $\vn, \v\alpha$ and identify what values of $\vn',\v{\alpha}';\vn'',\v{\alpha}''$ give non-zero elements. We first see that $\vn'+\vn'' = \vn$ gives non-zero components and so there are at most $(2\NF+1)^3$ possible pairs $(\vn',\vn'')$ for any $\vn$. For the Hermite labels we have $\v\alpha' =\v{0}$ and $\v\alpha'' \in \{\v\alpha - \v{e_a} ,\v{0}\}$ that are involved in non-zero components. Therefore, the row-sparsity is less than or equal to $s_{\mathrm{r}}=3 (2\NF+1)^3$. This implies that the sparse-encoding method carries a large overhead due to the dense columns.

We next compute the maximal normed element, and use the notation $\vn^\downarrow = (n_1,n_2,n_3)$ with $n_1 \ge n_2 \ge n_3 \ge0$. We have that
\small
\begin{align}
    \|\tilde{\m{F}}_2\|_{\max} =&\; \max_{\vn,\v{\alpha};\vn',\v{\alpha}';\vn'',\v{\alpha}''} |(\tilde{\m{F}}_2)_{\vn,\v{\alpha}|\vn',\v{\alpha}';\vn'',\v{\alpha}''}| \nonumber \\
    \le&\;  \frac{1}{L^{3/2}}\max_{\vn^\downarrow, \v\alpha} \left (\frac{\xi_{\vn, 1}\sqrt{\v\alpha_1+1}}{\sqrt{\kappa_{\vn} (1+\kappa_{\vn})}} \right ) + \frac{1}{L^{3/2}}\max_{\vn^\downarrow, \vn'} \left (\frac{\xi_{\vn, 1}\sqrt{\kappa_{\vn'}}}{\sqrt{\kappa_{\vn} (1+\kappa_{\vn})(1+\kappa_{\vn'})}} \right ) \nonumber \\
    \le&\;  \frac{1}{L^{3/2}}\max_{\vn^\downarrow} \left (\frac{(2\pi n_1/L )\sqrt{\NH+1}}{\sqrt{( (2\pi/L)^2|\vn|^2+\tau) (1+( (2\pi/L)^2|\vn|^2+\tau))}} \right )  \nonumber \\
    &\;+\frac{1}{L^{3/2}}\max_{\vn^\downarrow, \vn'} \left (\frac{(2\pi n_1/L)\sqrt{( (2\pi/L)^2|\vn'|^2+\tau)}}{\sqrt{( (2\pi/L)^2|\vn|^2+\tau) (1+( (2\pi/L)^2|\vn|^2+\tau))(1+( (2\pi/L)^2|\vn'|^2+\tau))}} \right ) \nonumber \\
    =&\;  \frac{1}{L^{3/2}}\max_{n_1} \left (\frac{(2\pi n_1/L )\sqrt{\NH+1}}{\sqrt{( (2\pi/L)^2n_1^2+\tau) (1+( (2\pi/L)^2n_1^2+\tau))}} \right ) + \nonumber \\
    &\;+\frac{1}{L^{3/2}}\max_{n_1, \vn'} \left (\frac{(2\pi n_1/L)\sqrt{( (2\pi/L)^2|\vn'|^2+\tau)}}{\sqrt{( (2\pi/L)^2n_1^2+\tau) (1+( (2\pi/L)^2n_1^2+\tau))(1+( (2\pi/L)^2|\vn'|^2+\tau))}} \right ) .
\end{align}
\normalsize
The maximization over $n_1$ in both terms reduces to computing
\begin{equation}
   M_L := \max_{0\le m \le \NF} \frac{km}{\sqrt{ (k^2m^2+\tau)(1+ (k^2m^2 +\tau))}},
\end{equation}
with $k = 2\pi/L$. The maximization over $\vn'$ is of the form
\begin{equation}
    R_L:=\max_{\vn' \in \mathbb{Z}_{\NF}^3} \sqrt{\frac{u (\vn')}{1+u(\vn')}}, \quad u :=  k^2 |\vn'|^2+\tau.
\end{equation}
This function is monotone increasing in $u$ and is therefore maximized at $u_{\max} :=  3\NF^2k^2+\tau$, with value
\begin{equation}
    R_L = \sqrt{\frac{ 3\NF^2k^2+\tau}{1+3\NF^2k^2+\tau}}.
\end{equation}
We therefore, just need to estimate $M_L$, and to do so we define the smooth function
\begin{equation}
    g(x)\coloneqq \frac{x}{\sqrt{ (x^2+\tau)(1+ (x^2+\tau))}}
\end{equation}
This is found to have a maximum at
\begin{equation}
    x_* = (\tau(1+\tau))^{1/4}.
\end{equation}
Thus,
\begin{equation}
    M_L \le g(\min \{2\pi \NF/L, x_*\}) \le g(x_*) = \frac{1}{\sqrt{\tau}+\sqrt{1+\tau}}.
\end{equation}
For sufficiently large enough $\NF/L$ and fixed $\tau$, the value of $x_*$ is inside the interval associated to the lattice points. Moreover, in the limit of $\NF/L\rightarrow 0$ the point $x_*$ can be approached arbitrarily closely.

Putting all of this together, we obtain
\begin{align}
    \|\tilde{\m{F}}_2\|_{\max} &\le \frac{M_L}{L^{3/2}} (  \sqrt{\NH+1} +  R_L) \nonumber \\
    & \le \frac{1}{L^{3/2}(\sqrt{\tau}+\sqrt{1+\tau})} \left ( \sqrt{\NH+1} +  \sqrt{\frac{ 3\NF^2k^2+\tau}{1+ 3\NF^2k^2+\tau)}} \right ),
\end{align}
with $k = 2\pi/L$. In the above $\NF k$ is the Nyquist cut-off scale for the spectral method associated to a real-space resolution scale $\Delta x \sim L/(2\NF)$ in any direction.

\newpage
\section{Norm estimates}
\label{app:normestimates}

Here, we prove various norm estimates used in the main text, and elsewhere.

\subsection*{Proof of Proposition~\ref{prop:P12norm}}

Observe that the matrices $\m{P}^{\pm 1/2}$ are diagonal and hence normal. It follows that 
\begin{equation}
\lVert \m{P}^{1/2} \rVert_2= \max_{\substack{\vn\in \mc{Z}_{\NF}^3 \\ \v{\alpha}\in \mc{N}_{\NH}^3}} \lvert (\m{P}^{1/2})_{\vn,\v{\alpha}|\vn,\v{\alpha}}\rvert= \max_{\vn\in \mc{Z}_{\NF}^3} \sqrt{1+\frac{1}{{\kappa}_{\vn}}}=\sqrt{1+\frac{1}{\tau}}
\end{equation}
and
\begin{equation}
\lVert \m{P}^{-1/2} \rVert_2= \max_{\substack{\vn\in \mc{Z}_{\NF}^3 \\ \v{\alpha}\in \mc{N}_{\NH}^3}} \big\lvert (\m{P}^{-1/2})_{\vn,\v{\alpha}|\vn,\v{\alpha}}\big\rvert
= \max_{\substack{\vn\in \mc{Z}_{\NF}^3 \\ \v{\alpha}\in \mc{N}_{\NH}^3}} \begin{dcases}
\sqrt{\frac{{\kappa}_{\vn}}{1+{\kappa}_{\vn}}} & \val=\v 0 \\
1 & \val\neq \v 0 
\end{dcases} 
= 1. 
\end{equation}
The condition number follows from $\kappa_2(\m{P}^{1/2})=\lVert \m{P}^{1/2}\rVert_2\lVert \m{P}^{-1/2}\rVert_2$.

\subsection*{Proof of Proposition~\ref{prop:u2Pbounds}}

We have that 
\begin{align}
\lVert \m{P}^{1/2}\m{u}\rVert_2^2=\lVert \m{u}\rVert_2^2+\sum_{\vn\in \mc{Z}_{\NF}^3} \frac{1}{\kappa_{\vn}}\lvert \m{u}_{\vn,\v 0} \rvert^2.
\end{align}
The lower bound in \eqref{eq:u2Pbounds} follows immediately. The corresponding upper bound follows from Proposition~\ref{prop:P12norm} and submultiplicativity.

For the second claim, let us first analyze the solution norm of the rescaled and Lyapunov transformed system in Eq.~\eqref{eq:Lyapunovandrescaledsystem}. Since $\| \bar{\m{u}}_0\| <1$ (Eq.~\eqref{eq:rescalednorm}), $\mu(\bar{\m{F}}_1)<0$ (Eq.~\eqref{eq:mu2F1}) and $R_{\m{P}}<1$ (Eq.~\eqref{eq:RP}), we have that $\| \bar{\m{u}}(t)\| \leq \| \bar{\m{u}}(0)\| <1$ (Ref.~\cite{jennings2025quantum}, Lemma 3.3). Furthermore, since $\bar{\m{u}} = \gamma \m{P}^{1/2} \m{u}$ for some $\gamma >0$, the previous inequality implies $ \|\m{P}^{1/2}\m{u}(t)\| \leq  \|\m{P}^{1/2}\m{u}(0)\|$ and so,
\begin{align}
    \|\m{u}(t)\| = \| \m{P}^{-1/2} \m{P}^{1/2} \m{u}(t)\| \leq \| \m{P}^{-1/2}\| \|\m{P}^{1/2} \m{u}(t)\| \leq \| \m{P}^{-1/2}\| \|\m{P}^{1/2}\m{u}(0)\| \leq \kappa(\m{P}^{1/2}) \| \m{u}(0)\|.
\end{align}

\subsection*{Proof of Proposition~\ref{prop:F1tildenorm}}

We write 
\begin{equation}
\tilde{\m{F}}_{1}=\bigoplus_{\vn\in \mc{Z}_{\NF}^3} \tilde{\m{F}}_{1,\vn}
\end{equation}
and consider the decomposition
\begin{equation}
\tilde{\m{F}}_{1,\vn}
=
-i\m{K}_{\vn}-\bar{\nu}\m{I},
\end{equation}
where
\begin{align}
(\m{K}_{\vn})_{\val|\val'}
= 
\sum_{a=1}^3 \xi_{\vn,a}
\big(
\sqrt{\alpha_a+1}\delta_{\v\alpha',\,\v\alpha+\ve_a}
+
\sqrt{\alpha_a}\delta_{\v\alpha',\v\alpha-\ve_a}
\big)
+\bigg(
\sqrt{\frac{1+\kappa_{\vn}}{\kappa_{\vn}}}-1
\bigg)
\sum_{a=1}^3 \xi_{\vn,a}
\big(
\delta_{\v\alpha,\v 0}\delta_{\v\alpha',\ve_a}
+
\delta_{\v\alpha,\ve_a}\delta_{\v\alpha',\v 0}
\big).
\end{align}
Observe that $\m{K}_{\vn}$ is real-symmetric. It follows that
\begin{equation}
\tilde{\m{F}}_{1,\vn}^{\dag}\tilde{\m{F}}_{1,\vn}=\m{K}_{\vn}^2+\bar{\nu}^2\m{I}
\end{equation}
and hence
\begin{equation}
\lVert\tilde{\m{F}}_{1,\vn}\rVert_2=\sqrt{\lVert \m{K}_{\vn}\rVert_2^2+\bar{\nu}^2}.
\end{equation}

We decompose $\m{K}_{\vn}$ as
\begin{equation}
\m{K}_{\vn}=\m{K}_{\vn}^{(1)}+\m{K}_{\vn}^{(2)},
\end{equation}
where
\begin{align}
\m{K}_{\vn}^{(1)}=&\;
\sum_{a=1}^3 \xi_{\vn,a}
\big(
\sqrt{\alpha_a+1}\delta_{\v\alpha',\,\v\alpha+\ve_a}
+
\sqrt{\alpha_a}\delta_{\v\alpha',\v\alpha-\ve_a}
\big), \\
\m{K}_{\vn}^{(2)}=&\;
\bigg(
\sqrt{\frac{1+\kappa_{\vn}}{\kappa_{\vn}}}-1
\bigg)
\sum_{a=1}^3 \xi_{\vn,a}
\big(
\delta_{\v\alpha,\v 0}\delta_{\v\alpha',\ve_a}
+
\delta_{\v\alpha,\ve_a}\delta_{\v\alpha',\v 0}
\big).
\end{align}

Let us write
\begin{equation}\label{eq:K_decomposition}
\m{K}_{\vn}^{(1)}=\sum_{a=1}^3 \xi_{\vn,a}\m{J}_a,
\end{equation}
where
\begin{equation}
\m{J}_a\coloneqq \m{j}_a+\m{j}_a^{\dag}, \quad (\m{j}_a)_{\val|\val'}\coloneqq \sqrt{\alpha_a}\delta_{\val,\val'+\v{e}_a}.
\end{equation}
We compute
\begin{equation}
\lVert \m{j}_a\rVert_2=\lVert \m{j}_a^{\dag}\rVert_2=\sqrt{\NH},
\end{equation}
and it follows that 
\begin{equation}
\lVert \m{J}_a\rVert_2\leq 2\sqrt{\NH}
\end{equation}
and, using \eqref{eq:K_decomposition}, 
\begin{equation}
\lVert \m{K}_{\vn}^{(1)}\rVert_2\leq  \lVert \xi_{\vn}\rVert_1\,2\sqrt{\NH} \leq \frac{12\pi}{L}\NF\sqrt{\NH}. 
\end{equation}

We next consider $\m{K}_{\vn}^{(2)}$. Define vectors $\m{r}_{\vn}$ and $\m{e}_{\v{0}}$ by 
\begin{equation}
(\m{e}_{\v 0})_{\val}\coloneqq \delta_{\val,\v 0},\quad (\m{r}_{\vn})_{\val}\coloneqq \sum_{a=1}^3 (\xi_{\vn})_a\delta_{\val,\ve_a}
\end{equation}
and write
\begin{equation}
\m{K}_{\vn}^{(2)}=
\bigg(
\sqrt{\frac{1+\kappa_{\vn}}{\kappa_{\vn}}}-1
\bigg)\big( \m{e}_{\v{0}}\m{r}_{\vn}^{\top}+\m{r}_{\vn}\m{e}_{\v{0}}^{\top} \big).
\end{equation}
Moreover, using 
\begin{equation}
\m{e}_{\v 0}^{\top}\m{r}_{\vn}=0 ,\quad  \lVert \m{r}_{\vn}\rVert_2^2=\lvert \v{\xi}_{\vn}\rvert^2, 
\end{equation}
we compute
\begin{align}
(\m{K}_{\vn}^{(2)})^2=&\;
\bigg(
\sqrt{\frac{1+\kappa_{\vn}}{\kappa_{\vn}}}-1
\bigg)^2\big(\lVert \m{r}_{\vn}\rVert_2^2\m{e}_{\v 0}\m{e}_{\v 0}^{\top}+\m{r}_{\vn}\m{r}_{\vn}^{\top} \big),
\end{align}
which manifestly has maximal rank two. To obtain the spectral radius of $(\m{K}_{\vn}^{(2)})^2$, we compute the nonzero eigenvalues by
\begin{equation}
(\m{K}_{\vn}^{(2)})^2\m{e}_{\v 0}=
\bigg(
\sqrt{\frac{1+\kappa_{\vn}}{\kappa_{\vn}}}-1
\bigg)^2\lvert\v{\xi}_{\vn}\rvert^2 \m{e}_{\v 0}, \quad 
(\m{K}_{\vn}^{(2)})^2\m{r}_{\vn}=
\bigg(
\sqrt{\frac{1+\kappa_{\vn}}{\kappa_{\vn}}}-1
\bigg)^2\lvert\v{\xi}_{\vn}\rvert^2 \m{r}_{\vn}.
\end{equation}
It follows that 
\begin{equation}
\lVert \m{K}_{\vn}^{(2)}\rVert_2^2 =\lVert(\m{K}_{\vn}^{(2)})^2\rVert_2=\bigg(
\sqrt{\frac{1+\kappa_{\vn}}{\kappa_{\vn}}}-1
\bigg)^2\lvert\v{\xi}_{\vn}\rvert^2.
\end{equation}

Using
\begin{equation}
\sqrt{1+x}-1\leq \frac{x}{2}
\qquad (x\geq 0),
\end{equation}
with $x=1/\kappa_{\vn}$, we obtain
\begin{align}
\lVert \m{K}_{\vn}^{(2)}\rVert_2^2
\leq
\bigg(\frac{1}{2\kappa_{\vn}}\bigg)^2\lvert\v{\xi}_{\vn}\rvert^2=
\frac{\lvert\v{\xi}_{\vn}\rvert^2}{4\kappa_{\vn}^2}
=
\frac{\lvert\v{\xi}_{\vn}\rvert^2}{4(\lvert\v{\xi}_{\vn}\rvert^2+\tau)^2}
\leq
\frac{1}{16\tau},
\end{align}
where in the last inequality we used
\begin{equation}
\frac{y}{(y+\tau)^2}\leq \frac1{4\tau}
\quad (y\geq 0)
\end{equation}
with $y=\lvert \v{\xi}_{\vn}\rvert^2$. Therefore,
\begin{equation}
\max_{\vn\in \mc{Z}_{\NF}^3}\lVert \m{K}_{\vn}^{(2)}\rVert_2
\leq
\frac{1}{4\sqrt{\tau}}.
\end{equation}

By the triangle inequality,
\begin{equation}
\lVert \m{K}_{\vn}\rVert_2
\leq
\lVert \m{K}_{\vn}^{(1)}\rVert_2+\lVert \m{K}_{\vn}^{(2)}\rVert_2
\leq
\frac{12\pi}{L}\NF\sqrt{\NH}+\frac{1}{4\sqrt{\tau}}.
\end{equation}
Hence,
\begin{equation}
\lVert\tilde{\m{F}}_{1,\vn}\rVert_2
\leq
\sqrt{
\bigg(
\frac{12\pi}{L}\NF\sqrt{\NH}+\frac{1}{4\sqrt{\tau}}
\bigg)^2+\bar{\nu}^2 }.
\end{equation}
Finally, since
\begin{equation}
\tilde{\m{F}}_{1}=\bigoplus_{\vn\in \mc{Z}_{\NF}^3} \tilde{\m{F}}_{1,\vn},
\end{equation}
we conclude that
\begin{align}
\lVert\tilde{\m{F}}_{1}\rVert_2
=&\; 
\max_{\vn\in \mc{Z}_{\NF}^3}\lVert\tilde{\m{F}}_{1,\vn}\rVert_2 \nonumber \\
\leq &\; 
\sqrt{
\bigg(
\frac{12\pi}{L}\NF\sqrt{\NH}+\frac{1}{4\sqrt{\tau}}
\bigg)^2+\bar{\nu}^2 },
\end{align}
which is \eqref{eq:F1tildebound}. The asymptotic scaling follows immediately. \\

\paragraph{Lower bound.}

Set $\vn_*=\NF(1,1,1)$ and $\val_*=\NH(1,1,1)$. Define a vector $\m{e}_{\val_*}\in \mathbb{C}^{\mc{N}_{\NH}^3}$ by $(\m{e}_{\val_*})_{\val}\coloneqq (\delta_{\val,\val_*})_{\val\in \mc{N}_{\NH}^3}$. Then,
\begin{align}
\tilde{\m F}_{1,\vn_*}\m e_{\val_*}=
-i\frac{2\pi}{L}\NF\sqrt{\NH}
\sum_{a=1}^3
\m e_{\val_*-\ve_a}
-\bar{\nu}\m e_{\val_*}.
\end{align}
Because the vectors $\m{e}_{\val_*}$ and $\m{e}_{\val_*-\v{e}_a}$ for $a=1,2,3$ are mutually orthonormal, we have
\begin{equation}
\label{eq:F1enorm}
\lVert \tilde{\m F}_{1,\vn_*}\m e_{\val_*} \rVert_2^2= 
3\left(\frac{2\pi}{L}\NF\right)^2\NH+\bar{\nu}^2.
\end{equation}

Combining \eqref{eq:F1enorm} with
\begin{equation}
\|\tilde{\m F}_{1}\|_2 \geq
\|\tilde{\m F}_{1,\vn_*}\|_2 \geq \lVert \tilde{\m F}_{1,\vn_*}\m e_{\val_*} \rVert_2
\end{equation}
gives the result \eqref{eq:F1tildelowerbound}.

\subsection*{Proof of Proposition~\ref{prop:F2tildenorm}}

The following lemma is adapted from \cite{byrne2009bounds}. 
\begin{lemma}
The following inequality holds,
\begin{equation}
\label{eq:F22estimate}
\lVert \tilde{\m{F}}_2\rVert_2 \leq \bigg(\max_{\substack{\vn',\vn''\in \mc{Z}_{\NF}^3 \\ \v{\alpha}',\v{\alpha}''\in \mc{N}_{\NH}^3}} \sum_{\substack{\vn\in \mc{Z}_{\NF}^3 \\ \v{\alpha}\in \mc{N}_{\NH}^3}}\lvert(\tilde{\m{F}}_2)_{\vn,\v{\alpha}|\vn',\v{\alpha}';\vn'',\v{\alpha}''}\rvert^{0^+}\bigg)^{1/2} \Bigg(\max_{\substack{\vn\in \mc{Z}_{\NF}^3 \\ \v{\alpha}\in \mc{N}_{\NH}^3}} \sum_{\substack{\vn',\vn''\in \mc{Z}_{\NF}^3 \\ \v{\alpha}',\v{\alpha}''\in \mc{N}_{\NH}^3}}\lvert(\tilde{\m{F}}_2)_{\vn,\v{\alpha}|\vn',\v{\alpha}';\vn'',\v{\alpha}''}\rvert^2\Bigg)^{1/2},
\end{equation}
where
\begin{equation}
x^{0^+}\coloneqq \lim_{\varepsilon\searrow 0} x^{\varepsilon}=\begin{dcases}
0 & x=0\\
1 & x>0.
\end{dcases}
\end{equation}
\end{lemma}

We use the lemma to prove the theorem. Recall that (Eq.~\eqref{eq:F2tildesum}) $\tilde{\m{F}}_2=\tilde{\m{F}}_2^{(1)}+\tilde{\m{F}}_2^{(2)}$, with the expressions for the two components given in Eq.~\eqref{eq:F2tilde_bulk_components}
and Eq.~\eqref{eq:F2tilde_firstshell_components}. Then, by the triangle inequality,
\begin{equation}
\label{eq:F2_triangle_inequality}
\lVert \tilde{\m{F}}_2 \rVert_2=\lVert \tilde{\m{F}}_2^{(1)}+\tilde{\m{F}}_2^{(2)} \rVert_2\leq \lVert \tilde{\m{F}}_2^{(1)} \rVert_2+\lVert \tilde{\m{F}}_2^{(2)} \rVert_2.
\end{equation}
We estimate each of these quantities.

For a fixed column $\vn',\vn''\in \mc{Z}_{\NF}^3$ and $\v{\alpha}',\v{\alpha}''\in \mc{N}_{\NH}^3$, we have that for $\lvert(\tilde{\m{F}}_2^{(1)})_{\vn,\v{\alpha}|\vn',\v{\alpha}';\vn'',\v{\alpha}''}\rvert^{0^+}$ to be nonzero, it is necessary that $\vn=\vn'+\vn''$ and $\v{\alpha}=\v{\alpha}''+\ve_a$ for some $a=1,2,3$. Hence, 
\begin{align}
\max_{\substack{\vn',\vn''\in \mc{Z}_{\NF}^3 \\ \v{\alpha}',\v{\alpha}''\in\mc{N}_{\NH}^3}}\sum_{\substack{\vn\in\mc{Z}_{\NF}^3\\\v{\alpha}\in \mc{N}_{\NH}^3 }}\lvert(\tilde{\m{F}}_2^{(1)})_{\vn,\v{\alpha}|\vn',\v{\alpha}';\vn'',\v{\alpha}''}\rvert^{0^+}\leq 3
\end{align}
i.e., $\tilde{\m{F}}_2^{(1)}$ is $3$-column-sparse and the first factor in the estimate \eqref{eq:F22estimate} is bounded by $\sqrt{3}$.

To estimate the second factor, we write
\begin{align}
\lvert(\tilde{\m{F}}_2^{(1)})_{\vn,\v{\alpha}|\vn',\v{\alpha}';\vn'',\v{\alpha}''}\rvert^2 = &\;\frac{1}{L^3}\frac{1}{\kappa_{\vn'}(1+\kappa_{\vn'})}\delta_{\vn,\vn'+\vn''}\delta_{\v{\alpha}',\v{0}}(1-\delta_{\v{\alpha}'',\v{0}}) \Bigg\lvert \sum_{a=1}^3
\xi_{\vn',a}\sqrt{\alpha''_a+1}
\delta_{\v{\alpha},\v{\alpha}''+\ve_a} \Bigg\rvert^2 \nonumber \\
=&\; \frac{1}{L^3}\frac{1}{\kappa_{\vn'}(1+\kappa_{\vn'})}\delta_{\vn,\vn'+\vn''}\delta_{\v{\alpha}',\v{0}}(1-\delta_{\v{\alpha}'',\v{0}})  \sum_{a=1}^3
\xi_{\vn',a}^2(\alpha''_a+1)
\delta_{\v{\alpha},\v{\alpha}''+\ve_a},
\end{align}
where we have used
\begin{equation}
 \Bigg\lvert \sum_{a=1}^3
\xi_{\vn',a}\sqrt{\alpha''_a+1}
\delta_{\v{\alpha},\v{\alpha}''+\ve_a} \Bigg\rvert^2 = \sum_{a=1}^3
\xi_{\vn',a}^2(\alpha''_a+1)
\delta_{\v{\alpha},\v{\alpha}''+\ve_a};
\end{equation}
for distinct $a\in \{1,2,3\}$, the constraints $\v{\alpha}=\v{\alpha}''+\ve_a$ are mutually exclusive, so no cross terms survive. Hence,
\begin{align}
\sum_{\substack{\vn',\vn''\in \mc{Z}_{\NF}^3 \\ \v{\alpha}',\v{\alpha}''\in \mc{N}_{\NH}^3}}\lvert(\tilde{\m{F}}_2^{(1)})_{\vn,\v{\alpha}|\vn',\v{\alpha}';\vn'',\v{\alpha}''}\rvert^2=&\; 
\frac{1}{L^3}\sum_{a=1}^3 \alpha_a(1-\delta_{\v{\alpha},\ve_a})\sum_{\substack{\vn'\in\mc{Z}_{\NF}^3 \\ \vn-\vn'\in\mc{Z}_{\NF}^3 }} \frac{\xi_{\vn',a}^2}{\kappa_{\vn'}(1+\kappa_{\vn'})} \nonumber \\\leq &\; 
\frac{1}{L^3}\sum_{a=1}^3 \alpha_a(1-\delta_{\v{\alpha},\ve_a})\sum_{\substack{\vn'\in\mc{Z}_{\NF}^3}} \frac{\xi_{\vn',a}^2}{\kappa_{\vn'}(1+\kappa_{\vn'})}.
\end{align}
and
\begin{align}
\max_{\substack{\vn\in\mc{Z}_{\NF}^3 \\ \v{\alpha}\in \mc{N}_{\NH}^3}}\sum_{\substack{\vn',\vn''\in \mc{Z}_{\NF}^3 \\ \v{\alpha}',\v{\alpha}''\in \mc{N}_{\NH}^3}}\lvert(\tilde{\m{F}}_2^{(1)})_{\vn,\v{\alpha}|\vn',\v{\alpha}';\vn'',\v{\alpha}''}\rvert^2\leq     
\frac{3\NH}{L^3}\sum_{\vn'\in\mc{Z}_{\NF}^3} \frac{\lvert\v{\xi}_{\vn'}\rvert^2}{\kappa_{\vn'}(1+\kappa_{\vn'})}.
\end{align}

We need the following lemma to proceed. 
\begin{lemma}\label{lem:latticesums}
The lattice sums
\begin{equation}
\mc{S}_1(\NF)\coloneqq \sum_{\vn\in\mc{Z}_{\NF}^3} \frac{\lvert\v{\xi}_{\vn}\rvert^2}{\kappa_{\vn}(1+\kappa_{\vn})}
\end{equation}
and 
\begin{equation}
\mc{S}_2(\NF)\coloneqq \sum_{\vn\in\mc{Z}_{\NF}^3} \frac{1}{1+\kappa_{\vn}}
\end{equation}
satisfy
\begin{equation}\label{eq:S1bound}
\mc{S}_1(\NF)\leq \mc{S}_2(\NF)\leq \frac{1}{1+\tau}+\frac{13L^2}{2\pi^2}\NF 
\end{equation}
respectively.
\end{lemma}

\begin{proof}
We first estimate $\mc{S}_2(\NF)$. We have
\begin{align}\label{eq:S2sum1}
\mc{S}_2(\NF)=  \frac{1}{1+\tau}+\sum_{\vn\in \mc{Z}_{\NF}^3\setminus\{\v 0\}}\frac{1}{1+\frac{4\pi^2}{L^2}\lvert\vn\rvert^2+\tau} \leq \frac{1}{1+\tau}+ \frac{L^2}{4\pi^2}\sum_{\vn\in \mc{Z}_{\NF}^3\setminus\{\v 0\}}\frac{1}{\lvert\vn\rvert^2}. 
\end{align}

Let
\begin{equation}
A_m\coloneqq \{\vn\in \mc{Z}_{\NF}^3: \lVert \vn\rVert_{\infty}=m\} \quad (m=1,\ldots,\NF)
\end{equation}
so that
\begin{equation}
\mc{Z}_{\NF}^3\setminus\{\v 0\}=\bigsqcup_{m=1}^{\NF} A_m
\end{equation}
and 
\begin{equation}
\# A_m=(2m+1)^3-(2m-1)^3=24m^2+2\leq 26m^2.
\end{equation}
It then follows from $\lvert \vn\rvert\geq \lVert\vn\rVert_{\infty}$ that 
\begin{equation}\label{eq:S2sum2} 
\sum_{\vn\in\mc{Z}_{\NF}^3\setminus\{\v 0\}} \frac{1}{\lvert\vn\rvert^2}=\sum_{m=1}^{\NF}\sum_{\vn\in A_m} \frac{1}{\lvert\vn\rvert^2} \leq  \sum_{m=1}^{\NF} 26 =26\NF. 
\end{equation}
Inserting \eqref{eq:S2sum2} into \eqref{eq:S2sum1} gives the second inequality in \eqref{eq:S1bound}. 

To estimate $\mc{S}_1(\NF)$, we use
\begin{equation}
\frac{\lvert\v{\xi}_{\vn}\rvert^2}{\kappa_{\vn}}=\frac{\lvert\v{\xi}_{\vn}\rvert^2}{\lvert\v{\xi}_{\vn}\rvert^2+\tau}\leq 1.
\end{equation}
Hence, 
\begin{equation}
\frac{\lvert\v{\xi}_{\vn}\rvert^2}{\kappa_{\vn}(1+\kappa_{\vn})}\leq \frac{1}{1+\kappa_{\vn}}
\end{equation}
and
\begin{equation}
\mc{S}_1(\NF)\leq \mc{S}_2(\NF). 
\end{equation}
We have established the first inequality in \eqref{eq:S1bound}. This completes the proof.  
\end{proof}

It follows that 
\begin{equation}\label{eq:F21bound}
\lVert \tilde{\m{F}}_2^{(1)}\rVert_2 \leq \frac{3\sqrt{\NH}}{L^{3/2}} \bigg( \frac{1}{1+\tau}+\frac{13L^2}{2\pi^2}\NF \bigg)^{1/2}.
\end{equation}

Next we estimate $\lVert \tilde{\m{F}}_2^{(2)}\rVert_2$. For a fixed column $\vn',\vn''\in \mc{Z}_{\NF}^3$ and $\v{\alpha}',\v{\alpha}''\in \mc{N}_{\NH}^3$, we have that for $\lvert(\tilde{\m{F}}_2^{(2)})_{\vn,\v{\alpha}|\vn',\v{\alpha}';\vn'',\v{\alpha}''}\rvert^{0^+}$ to be nonzero, it is necessary that $\vn=\vn'+\vn''$ and $\v{\alpha}=\ve_a$ for some $a=1,2,3$. Hence, 
\begin{align}
\max_{\substack{\vn',\vn''\in \mc{Z}_{\NF}^3 \\ \v{\alpha}',\v{\alpha}''\in\mc{N}_{\NH}^3}}\sum_{\substack{\vn\in\mc{Z}_{\NF}^3\\\v{\alpha}\in \mc{N}_{\NH}^3 }}\lvert(\tilde{\m{F}}_2^{(2)})_{\vn,\v{\alpha}|\vn',\v{\alpha}';\vn'',\v{\alpha}''}\rvert^{0^+}\leq 3
\end{align}
i.e., $\tilde{\m{F}}_2^{(2)}$ is $3$-column-sparse and the first factor in the estimate \eqref{eq:F22estimate} is bounded by $\sqrt{3}$. 

\begin{align}
\lvert(\tilde{\m{F}}_2^{(2)})_{\vn,\v{\alpha}|\vn',\v{\alpha}';\vn'',\v{\alpha}''}\rvert^2 =&\; \frac{1}{L^{3}}
\frac{\kappa_{\vn''}}{\kappa_{\vn'}(1+\kappa_{\vn'})(1+\kappa_{\vn''})}\delta_{\vn,\vn'+\vn''}
\delta_{\v{\alpha}',\v{0}}\delta_{\v{\alpha}'',\v{0}}\Bigg\lvert\sum_{a=1}^3
\xi_{\vn',a}
\delta_{\v{\alpha},\ve_a}\Bigg\rvert^2 \nonumber \\
=&\; 
\frac{1}{L^{3}}
\frac{\kappa_{\vn''}}{\kappa_{\vn'}(1+\kappa_{\vn'})(1+\kappa_{\vn''})}\delta_{\vn,\vn'+\vn''}
\delta_{\v{\alpha}',\v{0}}\delta_{\v{\alpha}'',\v{0}}\sum_{a=1}^3
\xi_{\vn',a}^2
\delta_{\v{\alpha},\ve_a}.
\end{align}
Using
\begin{equation}
 \Bigg\lvert\sum_{a=1}^3
\xi_{\vn',a}
\delta_{\v{\alpha},\ve_a}\Bigg\rvert^2=  \sum_{a=1}^3
\xi_{\vn',a}^2
\delta_{\v{\alpha},\ve_a}, 
\end{equation}
It follows that 
\begin{equation}
\sum_{\substack{\vn',\vn''\in \mc{Z}_{\NF}^3 \\ \v{\alpha}',\v{\alpha}''\in \mc{N}_{\NH}^3}}\lvert(\tilde{\m{F}}_2^{(2)})_{\vn,\v{\alpha}|\vn',\v{\alpha}';\vn'',\v{\alpha}''}\rvert^2=   \frac{1}{L^3}\sum_{a=1}^3 \delta_{\v{\alpha},\ve_a}\sum_{\substack{\vn'\in \mc{Z}_{\NF}^3 \\ \vn-\vn'\in \mc{Z}_{\NF}^3}} \frac{\kappa_{\vn-\vn'}\xi_{\vn',a}^2}{\kappa_{\vn'}(1+\kappa_{\vn'})(1+\kappa_{\vn-\vn'})}.
\end{equation}

Next, 
\begin{equation}
\max_{\substack{\vn\in\mc{Z}_{\NF}^3 \\ \v{\alpha}\in \mc{N}_{\NH}^3 }}\sum_{\substack{\vn',\vn''\in \mc{Z}_{\NF}^3 \\ \v{\alpha}',\v{\alpha}''\in \mc{N}_{\NH}^3}}\lvert(\tilde{\m{F}}_2^{(2)})_{\vn,\v{\alpha}|\vn',\v{\alpha}';\vn'',\v{\alpha}''}\rvert^2 = \frac{1}{L^3}\max_{\substack{\vn\in\mc{Z}_{\NF}^3\\ a\in \{1,2,3\} }} \sum_{\substack{\vn'\in \mc{Z}_{\NF}^3 \\ \vn-\vn'\in \mc{Z}_{\NF}^3}} \frac{\kappa_{\vn-\vn'}\xi_{\vn',a}^2}{\kappa_{\vn'}(1+\kappa_{\vn'})(1+\kappa_{\vn-\vn'})}.
\end{equation}

Using the uniform bounds
\begin{equation}
\frac{\kappa_{\vn-\vn'}}{1+\kappa_{\vn-\vn'}}\leq 1,\quad \frac{\xi_{\vn',a}^2}{\kappa_{\vn'}}\leq 1,
\end{equation}
we see that 
\begin{equation}
\max_{\substack{\vn\in\mc{Z}_{\NF}^3 \\ \v{\alpha}\in \mc{N}_{\NH}^3 }}\sum_{\substack{\vn',\vn''\in \mc{Z}_{\NF}^3 \\ \v{\alpha}',\v{\alpha}''\in \mc{N}_{\NH}^3}}\lvert(\tilde{\m{F}}_2^{(2)})_{\vn,\v{\alpha}|\vn',\v{\alpha}';\vn'',\v{\alpha}''}\rvert^2 \leq 
\frac{1}{L^3}\sum_{\v{n'}\in\mc{Z}_{\NF}^3}\frac{1}{1+\kappa_{\v{n}'}}.
\end{equation}

Hence, using Lemma~\ref{lem:latticesums},
\begin{equation}\label{eq:F22bound}
\lVert \tilde{\m{F}}_2^{(2)}\rVert_2\leq \frac{\sqrt{3}}{L^{3/2}}\bigg(\frac{1}{1+\tau}+\frac{13L^2}{2\pi^2}\NF \bigg)^{1/2}.
\end{equation}

Putting \eqref{eq:F21bound} and \eqref{eq:F22bound} into \eqref{eq:F2_triangle_inequality} gives the result \eqref{eq:F2tildespectralnorm}.

\newpage

\section{Block-encoding the nonlinear term}
\label{app:BE-of-nonlinear-term}
A direct sparse-encoding of the nonlinear term is inefficient. Even leveraging hierarchical methods, block-encoding in the Fourier-Hermite basis and Lyapunov-transformed frame of choice is not straightforward, due to the lack of appropriate decay properties. 

To circumvent this issue, we will analyze the nonlinear interaction term in real space and Hermite basis. We will construct the form of the Lyapunov-transformed interaction in this basis, for finite $L$ and in the $L\rightarrow \infty$ analytic limit. We shall see that, under appropriate regularity conditions, sampling the interaction kernel from a lattice only introduces small aliasing errors. Furthermore, we will see that, under appropriate conditions, we can exchange the finite $L$ for the $L\rightarrow \infty$ limit. The resulting kernel can be block-encoded efficiently via hierarchical methods, and transformed back into Fourier-Hermite basis via Fourier transform. This indirect procedure introduces controlled errors, and allows us to block-encode the target $\tilde{\m{F}}_2$ matrix up to decreasing aliasing errors in the Fourier-Hermite basis, even if this does not have the required properties to directly apply hierarchical block-encoding methods.

\subsection{Analysis of Lyapunov-transformed electrostatic field}

We start our analysis from the $\m{F}_2$ after Hermite truncation but before Fourier truncation 
\begin{align}
(\tilde{\m{F}}_2)_{\vn,\val | \vn', \val', \vn'', \val''}
= 
\delta_{\vn,\vn'+\vn''}(\tilde{\m{F}}'_2)_{\val | \vn', \val', \vn'', \val''}, \qquad \val, \val', \val'' \in \mathcal{N}_{\NF}, \vn, \vn' \vn'' \in \mathbb{Z}^3. 
\end{align}

Define the kernel as
\begin{align}
\label{eq:Kuntruncated}
\m{K}^L_{\val | \val', \val''}(\vx | \vx',\vx'') & = \frac{1}{L^{9/2}}
\sum_{\vn, \vn',\vn'' \in \mathbb{Z}^3}
\ee^{i (\v{\xi}_{\vn}\cdot \vx -\v{\xi}_{\vn'}\cdot \vx' -\v{\xi}_{\vn''}\cdot \vx'')}
(\tilde{\m{F}}'_2)_{\val | \vn', \val', \vn'', \val''} \delta_{\vn,\vn'+\vn''}.
\end{align}

In order to efficiently block-encode the nonlinear term we leverage the decay profile that the electric force has in real-space. We need the following technical lemma:

\begin{lemma}
\label{lem:transformed_field}
Define the function
\begin{equation}
\phi_L(\v{r})\coloneqq \frac{1}{L^3}\sum_{\vn\in\mathbb{Z}^3}\frac{1}{\sqrt{\kappa_{\vn}(1+\kappa_{\vn})}}\ee^{i\v{\xi}_{\vn}\cdot \v{r}}
\end{equation}
and the function
\begin{equation}
\sigma_L(\v{r})\coloneqq \frac{1}{L^3}\sum_{\vn\in\mathbb{Z}^3}\sqrt{\frac{\kappa_{\vn}}{1+\kappa_{\vn}}}\ee^{i\v{\xi}_{\vn}\cdot \v{r}}.
\end{equation}
Moreover, let 
\begin{equation}\label{eq:mus}
\mu(s)\coloneqq \sqrt{s+\tau}.
\end{equation}
Then, the following limits hold as $L\to\infty$:
\begin{equation}
\label{eq:phixlimit}
\phi_L(\v{r})\to \phi(\v{r})= \frac{1}{\pi}\int_0^1 \frac{1}{\sqrt{s(1-s)}}  \frac{\ee^{-\mu(s)\lvert\v{r}\rvert }}{4\pi\lvert\v{r}\rvert}\,\mathrm{d}s
\end{equation}
uniformly on compact subsets of $\R^3\setminus\{\v 0\}$ and
\begin{equation}
\label{eq:sigmaxlimit}
\sigma_L(\v{r})\to \sigma(\v{r})=   \delta^{(3)}(\v{r})-\frac{1}{\pi}\int_0^1 \sqrt{\frac{s}{1-s}}\,\frac{\ee^{-\mu(s)\lvert\v{r}\rvert }}{4\pi\lvert\v{r}\rvert}\,\mathrm{d}s,
\end{equation}
in the sense of distributions. 
\end{lemma}
\begin{proof}
    By the Feynman parameter identity
\begin{equation}
\frac{1}{\sqrt{AB}}
=
\frac{1}{\pi}\int_0^1 \frac{1}{\sqrt{s(1-s)}}\frac{1}{sA+(1-s)B}\,\mathrm ds
\end{equation}
with
\begin{equation}
A=\kappa_{\vn}+1,
\qquad
B=\kappa_{\vn},
\end{equation}
we obtain
\begin{equation}\label{eq:Feynman1}
\frac{1}{\sqrt{\kappa_{\vn}(1+\kappa_{\vn})}}
=
\frac{1}{\pi}\int_0^1 \frac{1}{\sqrt{s(1-s)}}\frac{1}{\kappa_{\vn}+s}\,\mathrm ds=
\frac{1}{\pi}\int_0^1 \frac{1}{\sqrt{s(1-s)}}\frac{1}{|\v\xi_{\vn}|^2+\mu(s)^2}\,\mathrm ds,
\end{equation}
where $\mu(s)$ is given by \eqref{eq:mus}. Hence
\begin{align}
\phi_L(\v r)
&=
\frac{1}{\pi}\int_0^1 \frac{1}{\sqrt{s(1-s)}}
\Bigg(
\frac{1}{L^3}\sum_{\vn\in\mathbb{Z}^3}
\frac{\ee^{i\v\xi_{\vn}\cdot \v r}}{|\v\xi_{\vn}|^2+\mu(s)^2}
\Bigg)\,\mathrm ds
\nonumber\\
&\to
\frac{1}{\pi}\int_0^1 \frac{1}{\sqrt{s(1-s)}}
\Bigg(
\frac{1}{(2\pi)^3}\int_{\R^3}
\frac{\ee^{i\v\xi\cdot \v r}}{|\v\xi|^2+\mu(s)^2}\,\mathrm d^3\xi
\Bigg)\,\mathrm ds
\nonumber\\
&=
\frac{1}{\pi}\int_0^1 \frac{1}{\sqrt{s(1-s)}}
\frac{\ee^{-\mu(s)|\v r|}}{4\pi |\v r|}\,\mathrm ds,
\end{align}
uniformly on compact subsets of $\R^3\setminus\{\v 0\}$, where we have used the standard identity
\begin{equation}
\frac{1}{(2\pi)^3}\int_{\R^3}
\frac{\ee^{i\v\xi\cdot \v r}}{|\v\xi|^2+\mu^2}\,\mathrm d^3\xi
=
\frac{\ee^{-\mu |\v r|}}{4\pi |\v r|},
\qquad
\mu>0.
\end{equation}

Similarly, using the previously established identity \eqref{eq:Feynman1}, we compute 
\begin{align}
\sqrt{\frac{\kappa_{\vn}}{1+\kappa_{\vn}}}
&=
\kappa_{\vn}\frac{1}{\sqrt{\kappa_{\vn}(1+\kappa_{\vn})}}
\nonumber\\
&=
\frac{1}{\pi}\int_0^1 \frac{1}{\sqrt{s(1-s)}}\frac{\kappa_{\vn}}{\kappa_{\vn}+s}\,\mathrm ds
\nonumber\\
&=
\frac{1}{\pi}\int_0^1 \frac{1}{\sqrt{s(1-s)}}
\left(1-\frac{s}{\kappa_{\vn}+s}\right)\,\mathrm ds
\nonumber\\
&=
1-\frac{1}{\pi}\int_0^1 \sqrt{\frac{s}{1-s}}\frac{1}{\kappa_{\vn}+s}\,\mathrm ds
\nonumber\\
&=
1-\frac{1}{\pi}\int_0^1 \sqrt{\frac{s}{1-s}}\,
\frac{1}{|\v\xi_{\vn}|^2+\mu(s)^2}\,\mathrm ds,
\end{align}
where we have used the exact integral
\begin{equation}
\frac{1}{\pi}\int_{0}^1 \frac{1}{\sqrt{s(1-s)}}\,\mathrm{d}s=1.
\end{equation}
It follows that
\begin{align}
\sigma_L(\v r)
&=
\delta_L^3(\v r)
-
\frac{1}{\pi}\int_0^1 \sqrt{\frac{s}{1-s}}
\Bigg(
\frac{1}{L^3}\sum_{\vn\in \mathbb{Z}^3}
\frac{\ee^{i\v\xi_{\vn}\cdot \v r}}{|\v\xi_{\vn}|^2+\mu(s)^2}
\Bigg)\,\mathrm ds
\nonumber\\
&\to
\delta^3(\v r)
-
\frac{1}{\pi}\int_0^1 \sqrt{\frac{s}{1-s}}
\Bigg(
\frac{1}{(2\pi)^3}\int_{\R^3}
\frac{\ee^{i\v\xi\cdot \v r}}{|\v\xi|^2+\mu(s)^2}\,\mathrm d^3\xi
\Bigg)\,\mathrm ds
\nonumber\\
&=
\delta^3(\v r)
-
\frac{1}{\pi}\int_0^1 \sqrt{\frac{s}{1-s}}
\frac{\ee^{-\mu(s)|\v r|}}{4\pi |\v r|}\,\mathrm ds,
\end{align}
in the sense of distributions.
This proves the claimed limits.

\end{proof}

We now have the following results that specify the decay profile of the Lyapunov-transformed field. In particular, we find that it behaves as  mixture of Yukawa potentials, and thus rapidly decays.
\begin{theorem}
\label{theorem:kerneluntruncatedcontinuum}
    The untruncated kernel in Eq~\eqref{eq:Kuntruncated} admits the following decomposition
    \begin{equation}
\m{K}^L_{\val | \val', \val''}(\vx|\vx',\vx'')
=
\m{K}^{L,(1)}_{\val | \val', \val''}(\vx|\vx',\vx'')
+
\m{K}^{L,(2)}_{\val|\val',\val''}(\vx|\vx',\vx''),
\end{equation}
where, for $a=1,2,3$
\begin{align}
\m{K}^{L,(1)}_{\val|\val',\val''}(\vx|\vx',\vx'')
& =
-\delta_L^3(\vx-\vx'')\nabla_x \phi_L(\vx-\vx')\cdot \v{W}^{(1)}_{\val|\val',\val''}, \quad 
(W_{\val|\val',\val''}^{(1)})_a
\coloneqq
\delta_{\val',\v{0}}(1-\delta_{\val'',\v{0}})\sqrt{\alpha''_a+1}\,
\delta_{\v{\val},\v{\val}''+\ve_a},
\nonumber
\\ 
\m{K}^{L,(2)}_{\val|\val',\val''}(\vx|\vx',\vx'') & = -\sigma_L(\vx-\vx'')\nabla_x\phi_L(\vx-\vx')\cdot \v{W}^{(2)}_{\val|\val',\val''},\quad
(W_{\val|\val',\val''}^{(2)})_a
\coloneqq
\delta_{\val',\v 0}\delta_{\val'',\v 0}\delta_{\val,\ve_a}.
\label{eq:kerneldecomposition}
\end{align}
and where we introduced the periodic Dirac delta
\begin{equation}
\label{eq:badlimiteddelta}
\delta_L^3(\vx)\coloneqq \frac{1}{L^3}\sum_{\vn\in \mathbb{Z}^3}\ee^{i\v{\xi}_{\vn}\cdot \vx}.
\end{equation}
Taking the limit $L\rightarrow \infty$ we have, 
\begin{align}
\m{K}^{L,(1)} \rightarrow \m{K}^{(1)}_{\val|\val',\val''}(\vx|\vx',\vx'')
&=\delta^3(\vx-\vx'')\mc{E}(\vx-\vx')\cdot \v{W}^{(1)}_{\val|\val',\val''} \nonumber \\
\m{K}^{L,(2)} \rightarrow  \m{K}^{(2)}_{\val | \val', \val''}(\vx | \vx',\vx'')&=
\sigma(\vx-\vx'')\mc{E}(\vx-\vx')\cdot \v{W}^{(2)}_{\val|\val',\val''},
\label{eq:kernelcontinuum}
\end{align}
on $\mathbb{R}^3 \setminus \{ \v{0}\}$ in the sense of distributions,
where
\begin{align}
 \mc{E}(\v{x}) := \frac{1}{\pi}\int_0^1 \left (\frac{\mu(s) \hat{\v{x}}}{\sqrt{s(1-s)}}
\frac{\ee^{-\mu(s)|\vx|}}{4\pi |\vx|} + \frac{\hat{\v{x}}}{\sqrt{s(1-s)}}
\frac{\ee^{-\mu(s)|\vx|}}{4\pi |\vx|^2}\right ) \,\mathrm ds.
\end{align}
\end{theorem}
\begin{proof}
    Replace in Eq.~\eqref{eq:Kuntruncated} the decomposition of $\tilde{\m{F}}_2 = \tilde{\m{F}}^{(1)}_2  + \tilde{\m{F}}^{(2)}_2 $  in Eq.~\eqref{eq:F2tildesum}. This correspondingly gives 
\begin{equation}
\m{K}^{L}_{\val|\val',\val''}(\vx|\vx',\vx'')
=
\m{K}^{L, (1)}_{\val|\val',\val''}(\vx|\vx',\vx'')
+
\m{K}^{L, (2)}_{\val|\val',\val''}(\vx|\vx',\vx'').
\end{equation}
First, we compute $\m{K}^{L, (1)}$ by replacing the expression for $\tilde{\m{F}}_2^{(1)}$
Eq.~\eqref{eq:F2tilde_bulk_components}:
\begin{align}
\m{K}^{L, (1)}_{\val | \val', \val''}(\vx | \vx',\vx'')
=&\; 
-\frac{i}{L^{6}}
\sum_{\vn',\vn''\in \mathbb{Z}^3}
\ee^{i (\v{\xi}_{\vn'}\cdot (\vx-\vx') +\v{\xi}_{\vn''}\cdot (\vx-\vx''))}
\frac{1}{\sqrt{\kappa_{\vn'}(1+\kappa_{\vn'})}}
\delta_{\val',\v 0}(1-\delta_{\val'',\v 0})\sum_{a=1}^3
\xi_{\vn',a}\sqrt{\alpha''_a+1}\,
\delta_{\val,\val''+\ve_a}
\nonumber\\
=&\; 
-i
\delta_{\val',\v 0}(1-\delta_{\val'',\v 0})
\sum_{a=1}^3
\sqrt{\alpha''_a+1}\,
\delta_{\val,\val''+\ve_a}
\nonumber\\
&\;\times
\Bigg(\frac{1}{L^3}
\sum_{\vn'\in\mathbb{Z}^3}
\frac{\xi_{\vn',a}}{\sqrt{\kappa_{\vn'}(1+\kappa_{\vn'})}}
\ee^{i \v{\xi}_{\vn'}\cdot (\vx-\vx')}
\Bigg)
\Bigg(\frac{1}{L^3}
\sum_{\vn''\in\mathbb{Z}^3}
\ee^{i \v{\xi}_{\vn''}\cdot (\vx-\vx'')}
\Bigg).
\end{align}
We obtain
\begin{align}
\m{K}^{L, (1)}_{\val | \val', \val''}(\vx | \vx',\vx'')
=&\;
-\delta_L^3(\vx-\vx'')
\delta_{\val',\v 0}(1-\delta_{\val'',\v 0})
\sum_{a=1}^3
\sqrt{\alpha''_a+1}\,
\delta_{\val,\val''+\ve_a}
\Bigg(
\frac{i}{L^3}\sum_{\vn'\in\mathbb{Z}^3}
\frac{\xi_{\vn',a}}{\sqrt{\kappa_{\vn'}(1+\kappa_{\vn'})}}
\ee^{i \v{\xi}_{\vn'}\cdot (\vx-\vx')}
\Bigg).
\end{align}
Hence,
\begin{equation}
\m{K}^{L, (1)}_{\val|\val',\val''}(\vx|\vx',\vx'')
=
-\delta_L^3(\vx-\vx'')\nabla_x \phi_L(\vx-\vx')\cdot \v{W}^{(1)}_{\val|\val',\val''},
\end{equation}
where
\begin{equation}
\label{eq:W1}
(W_{\val|\val',\val''}^{(1)})_a
\coloneqq
\delta_{\val',\v{0}}(1-\delta_{\val'',\v{0}})\sqrt{\alpha''_a+1}\,
\delta_{\v{\val},\v{\val}''+\ve_a},
\qquad a=1,2,3,
\end{equation}
which is the first of the two equations in \eqref{eq:kerneldecomposition}.

Next, we compute $\m{K}^{L, (2)}$ by replacing the expression for $\tilde{\m{F}}_2^{(2)}$
Eq.~\eqref{eq:F2tilde_firstshell_components}:
\begin{align}
\m{K}^{L, (2)}_{\val | \val', \val''}(\vx | \vx',\vx'')
=&\; -\frac{i}{L^{6}}
\sum_{\vn',\vn''\in\mathbb{Z}^3}
\ee^{i (\v{\xi}_{\vn'}\cdot (\vx-\vx') +\v{\xi}_{\vn''}\cdot (\vx-\vx''))}
\sqrt{\frac{\kappa_{\vn''}}{\kappa_{\vn'}(1+\kappa_{\vn'})(1+\kappa_{\vn''})}}
\delta_{\val',\v 0}\delta_{\val'',\v 0}
\sum_{a=1}^3 \xi_{\vn',a}\delta_{\val,\ve_a}
\nonumber\\
=&\;
-i
\delta_{\val',\v 0}\delta_{\val'',\v 0}
\sum_{a=1}^3 \delta_{\val,\ve_a}
\Bigg(\frac{1}{L^3}
\sum_{\vn'\in\mathbb{Z}^3}
\frac{\xi_{\vn',a}}{\sqrt{\kappa_{\vn'}(1+\kappa_{\vn'})}}
\ee^{i \v{\xi}_{\vn'}\cdot (\vx-\vx')}
\Bigg)
\Bigg(\frac{1}{L^3}
\sum_{\vn''\in\mathbb{Z}^3}
\sqrt{\frac{\kappa_{\vn''}}{1+\kappa_{\vn''}}}
\ee^{i \v{\xi}_{\vn''}\cdot (\vx-\vx'')}
\Bigg)
\nonumber\\
&=
-\sum_{a=1}^3
\delta_{\val',\v 0}\delta_{\val'',\v 0}\delta_{\val,\ve_a}\,
(\partial_{x_a}\phi_L)(\vx-\vx')\,\sigma_L(\vx-\vx'')
\nonumber\\
&=
-\sigma_L(\vx-\vx'')\nabla_x\phi_L(\vx-\vx')\cdot \v{W}^{(2)}_{\val|\val',\val''},
\end{align}
where
\begin{equation}
(W_{\val|\val',\val''}^{(2)})_a
\coloneqq
\delta_{\val',\v 0}\delta_{\val'',\v 0}\delta_{\val,\ve_a},
\qquad a=1,2,3,
\end{equation}
which is the second equation in~\eqref{eq:kerneldecomposition}.

The relation in the $L\rightarrow \infty$ limit is obtained from Lemma~\ref{lem:transformed_field}. 
The contribution to the first term in Eq.~\eqref{eq:kerneldecomposition} in the $L\rightarrow \infty$ is given by
\begin{align}
   \delta^3(\vx-\vx'') [ -\nabla_x \phi(\vx -\vx')] &= \delta^3(\vx-\vx'')(-\nabla_x) \left ( 
\frac{1}{\pi}\int_0^1 \frac{1}{\sqrt{s(1-s)}}
\frac{\ee^{-\mu(s)|\vx -\vx'|}}{4\pi |\vx -\vx' |}\,\mathrm ds \right ) \nonumber \\
&=   
\delta^3(\vx-\vx'')\frac{1}{\pi}\int_0^1 \left (\frac{\mu(s) \hat{\v{x}}}{\sqrt{s(1-s)}}
\frac{\ee^{-\mu(s)|\vx -\vx'|}}{4\pi |\vx -\vx'|} + \frac{\widehat{\v{x}-\v{x}'}}{\sqrt{s(1-s)}}
\frac{\ee^{-\mu(s)|\vx -\vx'|}}{4\pi |\vx -\vx'|^2}\right ) \,\mathrm ds.
\end{align}
The second term in Eq.~\eqref{eq:kerneldecomposition} in the $L\rightarrow \infty$ is similar, and given by
\begin{align}
    \sigma(\v{x}-\v{x}'') [ - \nabla_x \phi (\v{x}-\v{x}')] &= \left ( \delta^{3}(\v{x}-\v{x}'') - \frac{1}{\pi } \int_0^1 \sqrt{\frac{s}{1-s}} \frac{\ee^{-\mu(s)|\v{x}-\v{x}'' |}}{4\pi |\v{x}-\v{x}''|}\, \mathrm{d}s\right) \left [ - \nabla_x \phi (\v{x}-\v{x}') \right ].
\end{align}

\end{proof}
This provides an exact specification of the nonlinear interaction for the plasma system in the Lyapunov-transformed picture, both on the finite torus $\mathbb{T}_L^3$ and also with an explicit, analytic form in the $L\rightarrow \infty$ limit in terms of a mixture of Yukawa/screened fields. In the next section we show that in fact the unbounded analytic form can also be used for the bounded torus setting, with some simple caveats.

\subsection{The kernel $\m{K}$ to be block-encoded}

The form of the distributions $\phi_L (\v{x})$, $\sigma_L(\v{x})$ on the 3-torus can be related to that of the unbounded $\phi(\v{x})$, $\sigma(\v{x})$ via periodization by images. The precise statement is now given. 

\begin{proposition}[Periodization of $\phi(\vx)$ and $\sigma(\vx)$]
\label{prop:periodization}

The distributions $\phi_L$ and $\sigma_L$ on $\TL^3$ are, in the sense of distributions, the periodizations of $\phi$ and $\sigma$, respectively, 
\begin{align}
    \phi_L(\vx)
    =&\;
    \sum_{\v{m}\in \mathbb{Z}^3}
    \phi(\vx+\v{m}L),
    \label{eq:periodic_result_1}
    \\
    \sigma_L(\vx)
    =&\;
    \sum_{\v{m}\in \mathbb{Z}^3}
    \sigma(\vx+\v{m}L).
    \label{eq:periodic_result_2}
\end{align}
Moreover, \eqref{eq:periodic_result_1} converges locally uniformly on compact subsets of $\vx\in\TL^3\setminus\{\v 0\}$.
\end{proposition}

\begin{proof}
We use the standard periodization identity for the Yukawa Green's function. For
$\mu>0$, define
\begin{equation}
\label{eq:Yukawa_kernel_periodization_proof}
G_{\mu}(\vx)
\coloneqq
\frac{\ee^{-\mu|\vx|}}{4\pi|\vx|}.
\end{equation}
Then, in the sense of distributions on $\TL^3$ and pointwise on $\TL^3\setminus\{\v{0}\}$,
\begin{equation}
\label{eq:Yukawa_periodization_identity}
\sum_{\vn\in\mathbb Z^3}
G_\mu(\vx+\vn L)
=
\frac{1}{L^3}
\sum_{\vn\in\mathbb Z^3}
\frac{\ee^{i\v{\xi}_{\vn}\cdot\vx}}
{|\v{\xi}_{\vn}|^2+\mu^2}.
\end{equation}
Indeed, both sides are $L$-periodic distributions and both solve
\begin{equation}
\label{eq:Yukawa_periodic_Green_equation}
(-\nabla_x^2+\mu^2)G_{\mu}^L=\delta_L^{(3)}(\vx),
\end{equation}
where $G_{\mu}^L$ denotes the left-hand side of
\eqref{eq:Yukawa_periodization_identity}. The Fourier representation in
\eqref{eq:Yukawa_periodization_identity} is then fixed uniquely by
\begin{equation}
\label{eq:Yukawa_periodic_coefficients}
\delta_L^{(3)}(\vx)
\coloneqq
\frac{1}{L^3}
\sum_{\vn\in\mathbb Z^3}
\ee^{i\v{\xi}_{\vn}\cdot\vx}.
\end{equation}

We first prove the claim for $\phi_L$. By Lemma~\ref{lem:transformed_field},
\begin{equation}
\label{eq:phi_unbounded_yukawa_mixture}
\phi(\vx)
=
\frac{1}{\pi}
\int_0^1
\frac{G_{\mu(s)}(\vx)}{\sqrt{s(1-s)}}
\,\mathrm ds.
\end{equation}
We have $\mu(s)\geq \sqrt{\tau}$ uniformly in $s$. Hence, the
image sum is exponentially convergent away from $\vx=\v 0$, and the
interchange of the image sum with the $s$-integral is justified there. Thus,
for $\vx\in\TL^3\setminus\{\v 0\}$,
\begin{equation}
\label{eq:periodized_phi_first_step}
\sum_{\vn\in\mathbb Z^3}
\phi(\vx+\vn L)
=
\frac{1}{\pi}
\int_0^1
\frac{1}{\sqrt{s(1-s)}}
\sum_{\vn \in\mathbb Z^3}
G_{\mu(s)}(\vx+\vn L)
\,\mathrm ds.
\end{equation}
Applying \eqref{eq:Yukawa_periodization_identity} gives
\begin{equation}
\label{eq:periodized_phi_fourier_step}
\sum_{\vn \in\mathbb Z^3}
\phi(\vx+\vn L)
=
\frac{1}{L^3}
\sum_{\vn\in\mathbb Z^3}
\ee^{i\v{\xi}_{\vn}\cdot\vx}
\frac{1}{\pi}
\int_0^1
\frac{1}{\sqrt{s(1-s)}}
\frac{1}{|\v{\xi}_{\vn}|^2+\tau+s}
\,\mathrm ds.
\end{equation}
Using $\kappa_{\vn}=|\v{\xi}_{\vn}|^2+\tau$ and the Feynman-parameter identity
\begin{equation}
\label{eq:Feynman_identity_phi_periodization}
\frac{1}{\pi}
\int_0^1
\frac{1}{\sqrt{s(1-s)}}
\frac{1}{\kappa_{\vn}+s}
\,\mathrm ds
=
\frac{1}{\sqrt{\kappa_{\vn}(1+\kappa_{\vn})}},
\end{equation}
we obtain
\begin{equation}
\label{eq:periodized_phi_result}
\sum_{\vn \in\mathbb Z^3}
\phi(\vx+\vn L)
=
\frac{1}{L^3}
\sum_{\vn\in\mathbb Z^3}
\frac{\ee^{i\v{\xi}_{\vn}\cdot\vx}}
{\sqrt{\kappa_{\vn}(1+\kappa_{\vn})}}
=
\phi_L(\vx).
\end{equation}
This proves the first claim \eqref{eq:periodic_result_1}.

We now prove the claim for $\sigma_L$. By Lemma~\ref{lem:transformed_field},
\begin{equation}
\label{eq:sigma_unbounded_yukawa_mixture}
\sigma(\vx)
=
\delta^{(3)}(\vx)
-
\frac{1}{\pi}
\int_0^1
\sqrt{\frac{s}{1-s}}
G_{\mu(s)}(\vx)\,\mathrm ds.
\end{equation}
Periodizing the delta distribution gives
\begin{equation}
\label{eq:delta_periodization}
\sum_{\vn \in\mathbb Z^3}
\delta^{(3)}(\vx+\vn L)
=
\delta_L^{(3)}(\vx)
=
\frac{1}{L^3}
\sum_{\vn\in\mathbb Z^3}
\ee^{i\v{\xi}_{\vn}\cdot\vx}.
\end{equation}
Therefore, using \eqref{eq:Yukawa_periodization_identity} again, we find in
the sense of distributions that
\begin{equation}
\label{eq:periodized_sigma_fourier_step}
\sum_{\vn \in\mathbb Z^3}
\sigma(\vx+\vn L)
=
\frac{1}{L^3}
\sum_{\vn\in\mathbb Z^3}
\ee^{i\v{\xi}_{\vn}\cdot\vx}
\Bigg(
1
-
\frac{1}{\pi}
\int_0^1
\sqrt{\frac{s}{1-s}}
\frac{1}{\kappa_{\vn}+s}
\,\mathrm ds
\Bigg).
\end{equation}
The remaining scalar integral is elementary. For $\kappa>0$,
\begin{equation}
\label{eq:Feynman_identity_sigma_periodization}
\frac{1}{\pi}
\int_0^1
\sqrt{\frac{s}{1-s}}
\frac{1}{\kappa+s}
\,\mathrm ds
=
1-\sqrt{\frac{\kappa}{1+\kappa}}.
\end{equation}
Applying \eqref{eq:Feynman_identity_sigma_periodization} with
$\kappa=\kappa_{\vn}$ gives
\begin{equation}
\label{eq:periodized_sigma_result}
\sum_{\vn \in\mathbb Z^3}
\sigma(\vx+\vn L)
=
\frac{1}{L^3}
\sum_{\vn\in\mathbb Z^3}
\sqrt{\frac{\kappa_{\vn}}{1+\kappa_{\vn}}}
\ee^{i\v{\xi}_{\vn}\cdot\vx}
=
\sigma_L(\vx).
\end{equation}
This proves the second claim \eqref{eq:periodic_result_2}.
\end{proof}

 If we simply assume that $L$ is extremely large compared to the Debye length, and that the plasma is localized within the computational cell far from the boundaries, then we can simply take $$\phi_L(\v{x}) \approx \phi(\v{x}) \qquad \sigma_L(\v{x}) \approx \sigma(\v{x}),$$ 
 and the error can be made arbitrarily small for sufficiently large $L$ (we took $L=1000$ in the examples in Figs.~\ref{fig:Fourier_plot}--\ref{fig:Gaussian_plot}). This is the physically relevant regime, since when these assumptions fail, the periodic boundary conditions introduce spurious interactions of the plasma with its copy across the boundary. Note however we consider a Fourier mode as a benchmark for the convergence of our algorithm, since this is standard in the field. Alternatively, one could use more terms in the above Poisson summation formula, and combine block-encodings via LCU (we expect nearest neighbor corrections should suffice in practice), but we will not pursue this approach here.

 Under these assumptions, Theorem~\ref{theorem:kerneluntruncatedcontinuum} tells us that the untruncated kernel~\eqref{eq:Kuntruncated} can be approximated with its continuum version
\begin{align}
\label{eq:Kapprox}
\m{K}^{L}_{\val|\val',\val''}(\vx|\vx',\vx'') & \approx \m{K}_{\val|\val',\val''}(\vx|\vx',\vx'') = \m{K}^{(1)}_{\val|\val',\val''}(\vx|\vx',\vx'') +\m{K}^{(2)}_{\val|\val',\val''}(\vx|\vx',\vx'') \\ & =\delta^3(\vx-\vx'')\mc{E}(\vx-\vx')\cdot \v{W}^{(1)}_{\val|\val',\val''} +
\sigma(\vx-\vx'')\mc{E}(\vx-\vx')\cdot \v{W}^{(2)}_{\val|\val',\val''},
\end{align}
with $\v{W}^{(1,2)}$ having the form in Eq.~\eqref{eq:kerneldecomposition}.

\subsection{Regulation of the real-space singularity}
The electrostatic field in the Lyapunov picture blows up for $\v{x} \rightarrow \v{0}$, and so we must regulate it. We make the natural choice of a hard-core cut-off for $\|\v{x}\| \le a$, where $a$ is the cut-off scale. This has the consequence of distorting the spectral data for $\tilde{\m{F}}_2$ in Fourier space, and thus the nonlinear interaction in the plasma model on the spatial degrees of freedom given by $L^2(\mathbb{T}_L^3)$. Note that this regulation is separate from the projection onto the finite dimensional subspace under $\Pi_{\NF}$.

Recall that the interaction is of circulant form, and has a Coulomb-like singularity
\begin{equation}
    \int_0^1 \left (\frac{\mu(s) \hat{\v{x}}}{\sqrt{s(1-s)}}
\frac{\ee^{-\mu(s)|\vx|}}{4\pi |\vx|} + \frac{\widehat{\v{x}}}{\sqrt{s(1-s)}}
\frac{\ee^{-\mu(s)|\vx|}}{4\pi |\vx|^2}\right ) \,\mathrm ds
\end{equation}
with similar contributions arising in $\sigma (\v{x})$. For large separations we get exponential decay in real-space, while for very small separations we get divergences. We first consider regulating a single fixed term in the integrand, and set $r = |\vx|$. For convenience locally to this subsection we define
\begin{equation}
    E_\mu(\vx) := F_\mu(r) \hat{\vx},
\end{equation}
with
\begin{equation}
    F_\mu (r) := \ee^{-\mu r} \frac{(1 +\mu r)}{r^2}.
\end{equation}
We now regulate this to the form
\begin{equation}
F_{\mu,a}(r)=\begin{cases}
F_\mu(a), & 0<r\le a,\\
F_\mu(r), & r>a.
\end{cases}
\end{equation}
We can now consider both the regulated and unregulated fields in Fourier space. The unregulated case has
\begin{align}
    \widehat{\v{E}}_\mu (\v{\xi}) &= -i \v{\xi} \widehat{Y}_\mu(\v{\xi}) \nonumber \\
    \widehat{Y}_\mu(\v{\xi}) &= \frac{4 \pi}{|\v{\xi}|^2 + \mu^2}.
\end{align}
This implies that $|\widehat{\v{E}}_\mu| = O(1/|\v{\xi}|)$, which is the slow decay in Fourier space that we highlighted in the main text.

The Fourier space profile for the regulated field can also be computed from
\begin{equation}
\widehat{\v{E}}_{\mu,a}(\v{\xi}) = \int_{\mathbb{R}^3} \v{E}_{\mu,a}(\v{x}) \ee^{-i \v{\xi}\cdot \v{x}}\,\mathrm{d}^3 x.
\end{equation}
Indeed, it suffices to consider the difference between the regulated and unregulated fields, which takes the form
\begin{equation}
    \widehat{\v{E}}_{\mu}(\v{\xi}) - \widehat{\v{E}}_{\mu,a}(\v{\xi})  = -i \v{\xi} (\widehat{Y}_\mu - \widehat{Y}_{\mu,a}) =  \int_{r \le a} \left (\v{E}_{\mu}(\v{x}) - \v{E}_{\mu,a}(\v{x}) \right ) \ee^{-i \v{\xi}\cdot \v{x}}\,\mathrm{d}^3 x.
\end{equation}
A direct calculation gives that
\begin{equation}
    \widehat{Y}_\mu - \widehat{Y}_{\mu,a} = -\int_0^a\left[ r^2F_\mu(a)-\ee^{-\mu r}(1+\mu r)\right ]j_1(qr)\,\mathrm{d}r,
\end{equation}
where $q := |\v{\xi}|$ and $j_1(x)$ is the first spherical Bessel function. We now consider the regimes $q \ll 1/a$ and $q \gg 1/a$.

For $q \ll 1/a$ we have that
\begin{equation}
     \widehat{Y}_\mu - \widehat{Y}_{\mu,a} = O\left ( a^2 ( q^2 + \mu^2)\widehat{Y}_\mu  \right),
\end{equation}
and so the low frequency regime is only weakly affected by the regulation. For the regime $q \gg 1/a$ on the other hand, we find that
\begin{equation}
    |\widehat{\v{E}}_{\mu,a}| = O(1/q^3),
\end{equation}
and so the regulated Yukawa field $\widehat{\v{E}}_{\mu,a}$ has a much faster decay in Fourier space. The above cutoff has a discontinuity in derivatives at $r=a$, which limits the rate at which the Fourier spectrum decays. For smoother cutoffs we can obtain faster decays of the form $O(1/q^m)$ with $m \ge 4$.

The consequence of this is that if we consider a truncation to some scale $\|\v{n}\|_{\infty} \le \NF$ we can define an electric field $\v{E}^L_{\mu, a, \NF}$ by discarding modes with components above $\NF$. This provides a very good approximation to the actual regulated Yukawa field $\v{E}_{\mu,a}$, provided $\NF \ge 1/a$. Moreover it converges rapidly to the correct regulated field as we take $\NF \rightarrow \infty$.

In terms of block-encoding, the truncated interaction is dense in the Fourier indices, but does decay algebraically. However, we can do better by using the fact that the interaction has \emph{exponential} decay in real-space. We next discuss how this block-encoding is done in real-space, and for simplicity do the regularization tailored to the real-space sampling method we use.

\subsection{Real-space sampling}

Since the spatial dependence has been truncated to the Fourier modes $
\mathcal{Z}_{\NF}^3=\{-\NF,\ldots,\NF\}^3$,
each spatial variable is represented by a finite Fourier series. We therefore evaluate the kernel~\eqref{eq:Kapprox} on the uniform lattice
\begin{align}
    x_{\vj}=\frac{L}{N_x}\vj,
\qquad
\vj\in \mathbb Z_{N_x}^3,
\qquad
N_x=2\NF+1, \qquad \Delta_x := \frac{L}{N_x}
\end{align}
On this lattice, the exponentials $
    \{\ee^{2\pi i \vn\cdot \vj/N_x}\}_{\vn\in \mathcal{Z}_{\NF}^3}$,
form a discrete Fourier basis. Hence the grid values of any function with Fourier support in $\mathcal{Z}_{\NF}^3$ uniquely determine its Fourier coefficients, and vice versa. This is consistent with the fact that the largest mode is below the Nyquist sampling rate~\cite{shannon1949communication,nyquist1928certain}. With this, and the appropriate delta-function discretization, 
$$\delta(\v{x} - \v{x}') \rightarrow \frac{1}{\Delta^3_x} \delta_{\vj,\vj'}$$
we have from Eq.~\eqref{eq:Kapprox}
\begin{align}
\m{K}_{\v{j}, \val |\v{j}', \val', \v{j}'',\val''} & = \m{K}^{(1)}_{\v{j}, \val |\v{j}', \val', \v{j}'',\val''} + 
\m{K}^{(2)}_{\v{j}, \val |\v{j}', \val', \v{j}'',\val''} \\
\m{K}^{(1)}_{\v{j}, \val |\v{j}', \val', \v{j}'',\val''}
&=\frac{1}{\Delta^3_x} \delta_{\v{j}, \v{j}''} \mc{E}_{\v{j},\v{j}'}\cdot \v{W}^{(1)}_{\val|\val',\val''}  \\
\m{K}^{(2)}_{\v{j}, \val |\v{j}', \val', \v{j}'',\val''} &=
 \sigma_{\v{j}, \v{j}''} \mc{E}_{\v{j},\v{j}'}\cdot \v{W}^{(2)}_{\val|\val',\val''}.
\end{align}
where the expression $\mc{E}_{\v{j},\v{j}'}$ can be unpacked as follows,
\begin{align}
   \mc{E}_{\v{j},\v{j}'}&=  \frac{1}{\pi}\int_0^1 \mathrm{d}s \left (\frac{\mu(s) \hat{\v{x}}}{\sqrt{s(1-s)}}
\frac{\ee^{-\mu(s)\Delta_x|\v{j} -\v{j}'|}}{4\pi \Delta_x|\v{j} -\v{j}'|} + \frac{\hat{\v{x}}}{\sqrt{s(1-s)}} \frac{\ee^{-\mu(s) \Delta_x | \v{j} - \v{j}'|}}{4\pi \Delta_x^2|\v{j}-\v{j}'|^2 } \right ) .
\end{align}
Similarly,
\begin{align}
\sigma_{\v{j},\v{j}''}&= \frac{1}{\Delta_x^3} \delta_{\v{j}, \v{j}''} - \frac{1}{\pi} \int_0^1 \sqrt{\frac{s}{1-s}} \frac{\ee^{-\mu(s)\Delta_x|\v{j}-\v{j}''|}}{4\pi \Delta_x|\v{j}-\v{j}''|} \,\mathrm{d}s.
\end{align} 
We must regulate all singularities that occur at $\v{j}=\v{j}'$. We do this simply by setting
\begin{equation}
    \frac{1}{|\v{j}-\v{j}'|} \rightarrow 1 \mbox{ when } \v{j} \rightarrow \v{j}',
\end{equation}

With this, we define
\begin{align}
    \sigma'_{\v{j},\v{j}''} := \begin{dcases} \frac{1}{\pi } \int_0^1 \sqrt{\frac{s}{1-s}} \frac{\ee^{-\mu(s)\Delta_x|\v{j}-\v{j}''|}}{4\pi \Delta_x|\v{j}-\v{j}''|}\,\mathrm{d}s, \quad \v{j} \ne \v{j}'' \\
    \frac{1}{8\pi  \Delta_x}, \quad \v{j} = \v{j}''.
    \end{dcases}
\end{align}
and we have
\begin{align}
    \sigma_{\vj,\vj'}= \frac{1}{\Delta_x^3} \delta_{\v{j}, \v{j}''} - \sigma'_{\vj \vj''}.
\end{align}

Furthermore, denoting by $\v{e}_a$ unit vectors along coordinate directions
\begin{align}
    \mc{E}_{\v{j},\v{j}}&:=  \frac{\sum_{a=1}^3 \v{e_a}}{\pi}\int_0^1 \mathrm{d}s \left (\frac{\sqrt{s+\tau}}{\sqrt{s(1-s)}}
\frac{1}{4\pi \Delta_x} + \frac{1}{\sqrt{s(1-s)}} \frac{1}{4\pi \Delta_x^2 } \right )  \nonumber \\
&=  \frac{\sum_{a=1}^3 \v{e_a}}{\pi} \left (\int_0^1 \frac{\sqrt{s+\tau}}{\sqrt{s(1-s)}}
\frac{1}{4\pi \Delta_x} \,\mathrm{d}s +  \frac{1}{4 \Delta_x^2 } \right ) \nonumber \\
&=  \frac{\sum_{a=1}^3 \v{e_a}}{4\pi} \left (\sqrt{\tau}\,{}_2F_1\!\left(-\frac12,\frac12;1;-\frac{1}{\tau}\right)
\frac{1}{ \Delta_x}  +  \frac{1}{ \Delta_x^2 } \right ) \nonumber \\
& \le \frac{\sum_{a=1}^3 \v{e_a}}{4\pi} \left (\frac{\sqrt{1+\tau}}{ \Delta_x}
  +  \frac{1}{ \Delta_x^2 } \right ),
\end{align}
where the last inequality understood to apply component-wise. For simplicity we shall replace 
\begin{equation}
\mathcal{E}_{\v{j}, \v{j}} \rightarrow \mathcal{E}_{\v{j}, \v{j}} = \frac{\sum_{a=1}^3 \v{e_a}}{4\pi} \left (\frac{\sqrt{1+\tau}}{ \Delta_x}
  +  \frac{1}{ \Delta_x^2 } \right )
\end{equation}
This choice preserves monotonicity of the components in the bin-separation $|\v{j}-\v{j}'|$, and so
\begin{equation}
    \|\mc{E}^a\|_{\max} \le \frac{1}{4\pi} \left (\frac{\sqrt{1+\tau}}{ \Delta_x}
  +  \frac{1}{ \Delta_x^2 } \right ),
\end{equation}
for any $a=1,2,3$ with the max occurring on the diagonal entries.
Also we have $\|\sigma'\|_{\max} = 1/(8\pi \Delta_x)$, obtained on the diagonal entries.
We therefore have that $\m{K} = \m{K}^{(c)} + \m{K}^{(n)}$ is composed of the term
\begin{equation}
   \m{K}^{(c)}_{\v{j}, \val |\v{j}', \val', \v{j}'',\val''}:= \frac{1}{\Delta^3_x}\delta_{\v{j}, \v{j}''} \mc{E}_{\v{j},\v{j}'}\cdot (\v{W}^{(1)}_{\val|\val',\val''} +\v{W}^{(2)}_{\val|\val',\val''}),
\end{equation}
plus the term with the weighting $\sigma'_{\v{j},\v{j}''}$,
\begin{equation}
   \m{K}^{(n)}_{\v{j}, \val |\v{j}', \val', \v{j}'',\val''} := \sigma'_{\v{j},\v{j}''} \mc{E}_{\v{j},\v{j}'}\cdot  \v{W}^{(2)}_{\val|\val',\val''}.
\end{equation}
The superscripts $(c,n)$ are chosen since the first corresponds to a convolution, while the second does not.

The matrix $\m{K} = (\m{K}_{\v{j}, \val |\v{j}', \val', \v{j}'',\val''})$ defined above is obtained via a \emph{real-space sampling} of the full interaction that arises from the electrostatic field. This sampled matrix of data corresponds to a finite dimensional matrix $\tilde{\m{G}}_2$ in Fourier space defined as
\begin{align}
   \frac{1}{\Delta^{9/2}_x} (\tilde{\m{G}}_2)_{\vn,\val | \vn', \val', \vn'', \val''} &:= \frac{1}{N_x^{9/2}}
\sum_{\vj,\vj',\vj''\in \mathbb{Z}_{N_x}^3}
\ee^{-\frac{2\pi i}{N_x} (\vn \cdot \v{j} - \vn' \cdot \v{j}' - \vn'' \cdot \v{j}'') }\m{K}_{\v{j}, \val |\v{j}', \val', \v{j}'',\val''} \nonumber \\
&= \left [(\mathbb{F}_{N_x}^\dagger \otimes \m{I}) \m{K} ((\mathbb{F}_{N_x}\otimes \m{I})^{\otimes 2} \right ]_{\vn,\val | \vn', \val', \vn'', \val''},
\end{align}
where $\mathbb{F}_{N_x}$ denotes the discrete Fourier transform in 3-d over all the spatial indices $(j_1,j_2,j_3)$. By this construction we therefore define a truncated nonlinear interaction by sampling the exponentially decaying field. 

The interaction obtained from the real-space sampling in turn approximates the truncation in Fourier space to modes $\v{n} \in \{-\NF, \dots, \NF\}^{3}$, and
\begin{equation}
    \tilde{\m{G}}_2 \approx (\Pi_{\NF} \otimes \m{I})\tilde{\m{F}}_2 (\Pi_{\NF} \otimes \m{I})^{\otimes 2}.
\end{equation}
The error in this approximation is due to aliasing errors where high-frequency modes also contribute to $\tilde{\m{G}}_2$, while they are absent in $\tilde{\m{F}}_2$. However, as explained above, the regulation causes the Fourier modes to decay quickly, and therefore as $\NF\rightarrow \infty$ the gap between $\tilde{\m{G}}_2$ and $\tilde{\m{F}}_2$ closes. Since here we are concerned with the asymptotic scaling of the algorithm, the aliasing errors will not contribute under a smooth regularization. Alternatively we can simply take the data $(\tilde{\m{F}}_1, \tilde{\m{G}}_2, \v{g}(0))$ as defining the finite dimensional benchmark problem that approximates the plasma system of interest. Moreover, as we shall see later, the block-encoding scale factor of $\tilde{\m{G}}_2$ is $O(\sqrt{\NF \NH})$, which coincides with that of $\|\tilde{\m{F}}_2\|$, which in turn is asymptotically smaller than the norm of $\tilde{F}_1$, which is the actual bottleneck cost. We therefore have,

\begin{align}
 \frac{1}{\Delta^{9/2}_x} (\tilde{\m{F}}_2)_{\vn,\val | \vn', \val', \vn'', \val''} \approx \frac{1}{\Delta^{9/2}_x} (\tilde{\m{G}}_2)_{\vn,\val | \vn', \val', \vn'', \val''} &= \left [(\mathbb{F}_{N_x}^\dagger \otimes \m{I}) \m{K} ((\mathbb{F}_{N_x}\otimes \m{I})^{\otimes 2} \right ]_{\vn,\val | \vn', \val', \vn'', \val''},
\end{align}
and we simply take $\tilde{\m{F}}_2 = \tilde{\m{G}}_2$ in what follows.
This implies that a unitary block-encoding $\m{U}_{\m{K}}$ of $\m{K}$ with scale factor $\alpha_{\m{K}}$ gives a unitary block-encoding $\m{U}_{\tilde{\m{F}}_2}$ of $\tilde{\m{F}}_2$ with scale factor 
\begin{equation}
    \alpha_{\tilde{\m{F}}_2} = \Delta_x^{9/2}\alpha_{\m{K}},
\end{equation}
via controlled discrete Fourier transforms conjugating $\m{U}_{\m{K}}$, which can be efficiently implemented on a quantum computer. The problem of constructing a block-encoding of $\tilde{\m{F}}_2$ can thus be reduced to constructing a unitary block-encoding of $\m{K}$. This  is a dense matrix obtained from evaluating a rapidly decreasing profile on a cube of lattice points, however we construct $\m{U}_{\m{K}}$ via a Hierarchical Block-Encoding (HBE) approach. We turn to this theory in the next section.

\newpage
\section{Hierarchical block-encodings}
\label{app:HBEs}

This appendix develops the hierarchical block-encoding framework used to encode dense spatial interactions with scale factors controlled by kernel decay.

\subsection{Block-encoding of rectangular matrices}
\label{app:block-encoding-rectangular}
We begin by generalizing the notation of block-encodings to rectangular matrices. We assume for simplicity that a given matrix transforms definite numbers of qubits, so the dimensions are power of $2$. In what follows, we shall use the notation that \emph{unprimed} indices or variables refer to the \emph{rows} of a matrix $\m{A}$, and \emph{primed} indices or variables refer to the \emph{columns} of $\m{A}$. We begin by introducing a convenient extension of the standard unitary block-encoding.

\textbf{Rectangular block-encodings.} Let $\m{A} \in \mathbb{C}^{2^s \times 2^{s'}}$ be a rectangular matrix, mapping $s'$ qubits into $s$ qubits. We say that $\m{U}_{\m{A}}$ is an $(\alpha, a_{\mathrm{out}}, a'_{\mathrm{in}}, \epsilon)$ block-encoding of $A$~if 
\begin{equation}
\| (\bra{0^{a_{\mathrm{out}}}} \otimes I_s) \m{U}_{\m{A}} (\ket{0^{a'_{\mathrm{in}}}} \otimes I_{s'}) - \m{A}/\alpha\| \le \epsilon,
\end{equation}
and $a_{\mathrm{out}}+s = a'_{\mathrm{in}}+s'$.

In particular, for an exact block-encoding, with $\epsilon=0$, we have that
\begin{equation}
    \m{U}_{\m{A}} \ket{0^{a'_{\mathrm{in}}}}\ket{\psi} = \ket{0^{a_{\mathrm{out}}}}\otimes \frac{\m{A}}{\alpha}\ket{\psi} + \ket{\perp},
\end{equation}
where $\m{A}$ maps a state $\ket{\psi}$ on $s'$ qubits into a vector in an $s$--qubit space, and $\ket{\perp}$ denotes a vector in the nullspace of the projector $\ketbra{0^{a_{\mathrm{out}}}}{0^{a_{\mathrm{out}}}}\otimes I_s$.
\subsection{Weighted sparse encodings}
We now give relevant block-encodings for sparse rectangular matrices.  We say a rectangular matrix $\m{A}\in \mathbb{C}^{2^s\times 2^{s'}}$ is $\dr$ row-sparse if it has at most $\dr$ non-zero elements in each row, and say it is $\dc$ column sparse if it has at most $\dc$ non-zero elements in each column. I.e., if $\m{A}=(a_{mn'})$, then row-sparsity of $\dr$ means that for fixed $m$ the number of values for $n'$ that give a nonzero element is (at most) $\dr$. Similarly, column-sparsity of $\dc$ means that for fixed $n'$ the number of values for $m$ that give a nonzero element is $\dc$.

\begin{lemma}[Sparse encoding]
\label{lemma:sparseBE}
Let $\m{A} \in \mathbb{C}^{2^s \times 2^{s'}}$ be any rectangular matrix sending $s'$-qubits into $s$-qubits,  and assume $\m{A} = (a_{mn'})$ is $(\dr,\dc)$ (row,column)-sparse. Assume we have access to the following oracle
\begin{equation}
\m{O}_\m{A}: \ket{m} \ket{n'}\ket{0^b} \rightarrow  \ket{m} \ket{n'}\ket{\tilde{a}_{m n'}}, 
\end{equation}
where $\tilde{a}_{mn'}$ is a $b$-bit encoding of the component $a_{mn'}$. Also, assume we have access to oracles $\m{O}'_{r}$ and $\m{O}_{c}$ that both act on $s+s'+1$ qubit registers, such that
\begin{align}
\m{O}'_{r} \ket{m}\ket{n}&=\ket{m} \ket{r'_{m n}} \\
\m{O}_{c} \ket{m}\ket{n'} &= \ket{c_{m n'}} \ket{n'} .
\end{align}
where $r'_{mn}$ is the column index of the $n$'th non-zero element in row $m$ if such an entry exists, or is equal to $n+2^{s'}$ otherwise.  Likewise $c_{mn'}$ is the row index of the $m$'th non-zero element in column $n'$,  or if no such entry exists it is equal to $m+2^s$.  Let $\hat{a} := \max_{m,n'} |a_{m,n'}|$.

Then we may construct an $(\hat{a} \sqrt{\dr \dc}, s'+3, s+3, \epsilon)$ block-encoding of the rectangular matrix $\m{A}$ using a single call to each of $\m{O}_{\m{A}}$, $\m{O}^\dagger_{\m{A}}$, $\m{O}'_{r}$, $\m{O}_{c}$ and $O(b, polylog (\hat{a} \sqrt{\dr \dc}/\epsilon))$ auxiliary qubits and primitive logical gates.
\end{lemma}
The proof of this is near-identical as for square matrices (Lemma 48~\cite{gilyen2018quantum}) so it will be omitted.

We now consider the case where we can split $A$ up into a sum of matrices with varying sparsity parameters and max-norm elements.  We get the following key result. 
\begin{lemma}[Weighted sparse encoding] 
\label{lem:weighted-sparse-encoding}
Let $\m{A} \in \mathbb{C}^{2^s \times 2^{s'}}$ be any rectangular matrix sending $s'$-qubits into $s$-qubits. Suppose we write 
\begin{equation}
\m{A} = \sum_{l\in [\ell_{\max}]} \m{A}^{(l)},
\end{equation}
where $[\ell_{\max}] = \{1,2,\dots \ell_{\max}\}$ and each matrix $\m{A}^{(l)} \in  \mathbb{C}^{2^s \times 2^{s'}}$ is $(d^l_{\mathrm{r}},d^l_{\mathrm{c}})$ row/column sparse, for each $l=1,2,\dots, \ell_{\max}$.  Assume the matrices $\{\m{A}^{(l)}\}$ have disjoint support, and denote by $a^l_{mn'}$ the components of $\m{A}^{(l)}$ in the computational basis, and define $\hat{a}_l := \max_{m,n'} | a^l_{mn'}|$. 

Assume we have access to the following oracle
\begin{equation}
\m{O}_{\m{A}}: \ket{m} \ket{n'}\ket{0^b} \rightarrow  \ket{m} \ket{n'}\ket{\tilde{a}_{m n'}}, 
\end{equation}
where $\tilde{a}_{mn'}$ is a $b$-bit encoding of the component $a_{mn'}$.  Define the oracles $\m{O}^l_{r}$ and $\m{O}^l_{c}$ that both act on $s+s'+1$ qubit registers, such that
\begin{align}
\m{O}^l_{r} \ket{m}\ket{n}&=\ket{m} \ket{r^l_{m n}} \\
\m{O}^l_{c} \ket{m}\ket{n'} &= \ket{c^l_{m n'}} \ket{n'} .
\end{align}
where $r^l_{m n}$ is the column index of the $n$'th non-zero element in row $m$ of $\m{A}^{(l)}$ if such an entry exists, or is equal to $n+2^{s'}$ otherwise.  Likewise $c^l_{ m n'}$ is the row index of the $m$'th non-zero element in column $n'$ of $\m{A}^{(l)}$,  or if no such entry exists it is equal to $m+2^s$.  
We assume access to the controlled oracles
\begin{align}
\m{O}_r &:= \sum_{l \in [\ell_{\max}]} \ketbra{l}{l} \otimes \m{O}^l_{r} \\ 
\m{O}_c &:= \sum_{l \in [\ell_{\max}]} \ketbra{l}{l} \otimes \m{O}^l_{c} .
\end{align}
We also assume we have access to an oracle $\textsc{prep}$ acting on $\Lambda = \lceil \log \ell_{\max} \rceil$ qubits, such that
\begin{equation}
\textsc{prep}\ket{0^\Lambda} = \sum_{l\in [\ell_{\max}]} \sqrt{w_l} \ket{l} ,
\end{equation}
where the amplitudes are given by $w_l := \hat{a}_l \sqrt{\dr^l \dc^l} /\big(\sum_{y\in[\ell_{\max}]} \hat{a}_y \sqrt{d^y_{r}d^y_{c}}\big)$. Finally, we assume access to an oracle $\m{O}_{\max}$ such that
\begin{equation}
    \m{O}_{\max} \ket{l} \ket{0} = \ket{l}\ket{\hat{a}_l},
\end{equation}
which encodes a $b$-bit encoding of $\hat{a}_l$ in the second register.

Then we may construct an $(\alpha, s'+\Lambda+3, s+\Lambda+3, \epsilon)$ block-encoding of the rectangular matrix $A$ with
\begin{equation}
\alpha =\sum_{l\in [\ell_{\max}]} \hat{a}_l \sqrt{\dr^l\dc^l}.
\end{equation}
The construction uses a single call to each of $\m{O}_{\m{A}}$, $\m{O}^\dagger_{\m{A}}$, $\textsc{prep},$ $\textsc{prep}^\dagger$, $\m{O}_{\max}$, $\m{O}_{\max}^\dagger$, $\m{O}'_{r}$, $O_{c}$ and also using $O(b, \Lambda, polylog (\hat{a} \max_l(\sqrt{\dr^l\dc^l})/\epsilon))$ auxiliary qubits and primitive logical gates.
\end{lemma}
\begin{proof} 
For any $\m{A}^{(l)}$, Lemma~\ref{lemma:sparseBE} tells us that we can construct a block-encoding $\m{U}_{\m{A}^{(l)}}$ with scale factor $\alpha_l:=\hat{a}_l \sqrt{\dr^l \dc^l}$. Given that
\begin{equation}
    \m{A}= \sum_{l\in [\ell_{\max}]} \alpha_l  \left(\frac{\m{A}^{(l)}}{\alpha_l} \right ),
\end{equation}
it follows that by constructing an LCU of $\sum_l \alpha_l \m{U}_{\m{A}^{(l)}}$ using $\Lambda$ auxiliary qubits, we obtain a block-encoding of $\m{A}$ with scale factor $\alpha$, as claimed.
Since the non-zero elements of $\m{A}^{(l)}$ appear in $\m{A}$ we need only call $\m{O}_{\m{A}}$ and $\m{O}_\m{A}^\dagger$ once in the LCU circuit. Similarly, the other oracles are only required to be called once. The qubit and gate dependencies on $b$ and $\epsilon$ follow from those of block-encoding the terms $\m{A}^{(l)}$.
\end{proof}

Crucially, we shall use this lemma in the case of the nonlinear interaction to obtain an efficient, and in fact asymptotically optimal, block-encoding of the nonlinear plasma term. Recall that an asymptotically optimal unitary block-encoding for a matrix is one whose constant prefactor scales as the norm of the matrix.

\subsection{Hierarchical decomposition of the nonlinear interaction}
We now apply this theory to the nonlinear term $\tilde{\m{F}}_2$ in the plasma problem. Recall that this is in the Fourier-Hermite basis, but with a discrete Fourier transform we obtain a matrix with spatial indices that have a decaying profile in terms of the separation of points. The hierarchical decomposition we now describe is applied to the spatial degrees of freedom only. The central idea is to separate out the contributions to the electric field interaction that arise at different length scales. We then can block-encode the contribution at each level, and then finally obtain a block-encoding of the full interaction via the weighted sparse encoding given in Lemma~\ref{lem:weighted-sparse-encoding}. The weightings used correspond to the largest magnitude of the interaction at a given level, and since the electric interaction drops off rapidly in separation, we can mitigate the fact it is a dense matrix by forming large ``blocks" associated to well-separated points and small ``blocks" associated to points near to each other.

We first illustrate the idea by the following toy example where we replace the matrix for the electric field with a near-pointlike interaction $\m{K}=(\m{K}_{ij})$ that is of unit strength $\m{K}_{ii}=1$ for points that coincide $i=j$, and is negligible for any two separated points: $\m{K}_{ij} = \epsilon$ for $i\ne j$. In particular, we can consider the following $2^n\times 2^n$ matrix on $n$ qubits,
\begin{equation}
    \m{K} = \begin{bmatrix}
        1 & \epsilon & \epsilon &\epsilon &\dots &\epsilon \\
        \epsilon & 1 & \epsilon &\epsilon &\dots &\epsilon \\
        \epsilon & \epsilon & 1 &\epsilon &\dots &\epsilon \\
        \vdots & \vdots & \ddots &\ddots &\dots &\vdots \\
        \epsilon & \epsilon & \epsilon &\epsilon &1 &\epsilon \\
        \epsilon & \epsilon & \epsilon &\epsilon &\dots &1 \\
    \end{bmatrix}
\end{equation}
We can write this as $\m{K}=\m{I} + \epsilon \m{D}$, where $\m{D}$ is the matrix of all ones on the non-diagonal elements. It is clear that $\m{K}$ is a dense matrix, and so a direct sparse-encoding of $\m{K}$ would lead to an exponentially growing scale factor 
\begin{equation}
    \alpha_{\m{K},\textrm{sparse}} = 2^n.
\end{equation}
However, if we separately handle the ``well-separated" points (here $i\ne j$) from the ``near" points (here $i=j$) then we can separately block encode $\m{K}^{(1)}:= \epsilon \m{D}$ and $\m{K}_{\mathrm{ad}} := I$ we obtain scale factors
\begin{align}
    \alpha_{\m{K}^{(1)}} = (2^n-1)\epsilon \nonumber \\
    \alpha_{\m{K}_{\mathrm{ad}}} = 1.
\end{align}
Now, a weighted sparse encoding gives
\begin{equation}
    \alpha_{\m{K},\textrm{weighted}} = 1+\epsilon (2^n-1).
\end{equation}
Therefore, if $\epsilon = O(2^{-n})$ say, we obtain an $O(1)$ scale factor for the interaction matrix. This is essentially the idea behind the hierarchical block-encoding we perform on a cubic grid of $O(N_x^3)$ points, with the difference being that instead of decomposing $\m{K}$ into two pieces we decompose it into a sum of $O(\log N_x)$ pieces, each corresponding to a scale or `level' labeled by an index $l$, and construct corresponding matrices $\m{K}^{(l)}$. The remaining contributions from very nearby points are put into an `adjacent' matrix $\m{K}_{\mathrm{ad}}$, which is a sparse, near-diagonal matrix.

The high-level details are as follows. We will consider a kernel matrix $\m{K}_{\v{j},\v{j}'}$
where $\v{j}, \v{j}'$ are  in $\mathbb{Z}_{N_x}^3$ and which has a decay profile
\begin{equation}
    |\m{K}_{\v{j},\v{j}'}| \le \mathcal{E}(|\v{j}-\v{j}'|), \quad \mbox{ for all } \v{j},\v{j}',
\end{equation}
for some monotonically decreasing function $\mc{E}$. We shall assume $N_x= 2^{\ell_{\max}}$ for simplicity, where $\ell_{\max}$ is some fixed integer. We now define hierarchy levels labeled by $l$. At level $l=0,1, 2,\dots,{\ell_{\max}}$, the set of all lattice points is partitioned into cubic clusters of size $2^{3(\ell_{\max}-l)}$. A level $l$ can be viewed as a coarse-graining, or resolution scale. $l=0$ corresponds to coarse-graining the whole set to a single element, while $l=\ell_{\max}$ is the finest scale, with clusters composed of the individual lattice points. At any level $l$, two $l$-clusters are said to be \emph{adjacent} if they share a face, edge or a corner, otherwise there are said to form an \emph{admissible block} in the matrix. 

The following construction gives the hierarchical decomposition.  
\begin{enumerate}
    \item Start from $l=2$.
    \item Split the cubic lattice in $l$-clusters.
    \item Select the admissible blocks in $\m{K}$.
    \item Remove from matrix $\m{K}$ the elements $\m{K}_{\v{j},\v{j}'}$ with indices ($\v{j},\v{j}')$ selected in step 3 and put the elements into a matrix $\m{K}^{(l)}$.
    \item Increase $l\mapsto l+1$ and go back to step 2, repeating the process on the resulting sublattice.
    \item Repeat up to $l={\ell_{\max}}$  when the filtering process terminates.
    \item Collect the remaining points not filtered into a single matrix $\m{K}_{\mathrm{ad}}$.
\end{enumerate}

For example, Fig.~\ref{fig:blockstructure2D} shows the result of applying this algorithm 
in a $16 \times 16$ lattice. 
\begin{figure}
    \centering
\includegraphics[width=0.4\linewidth]{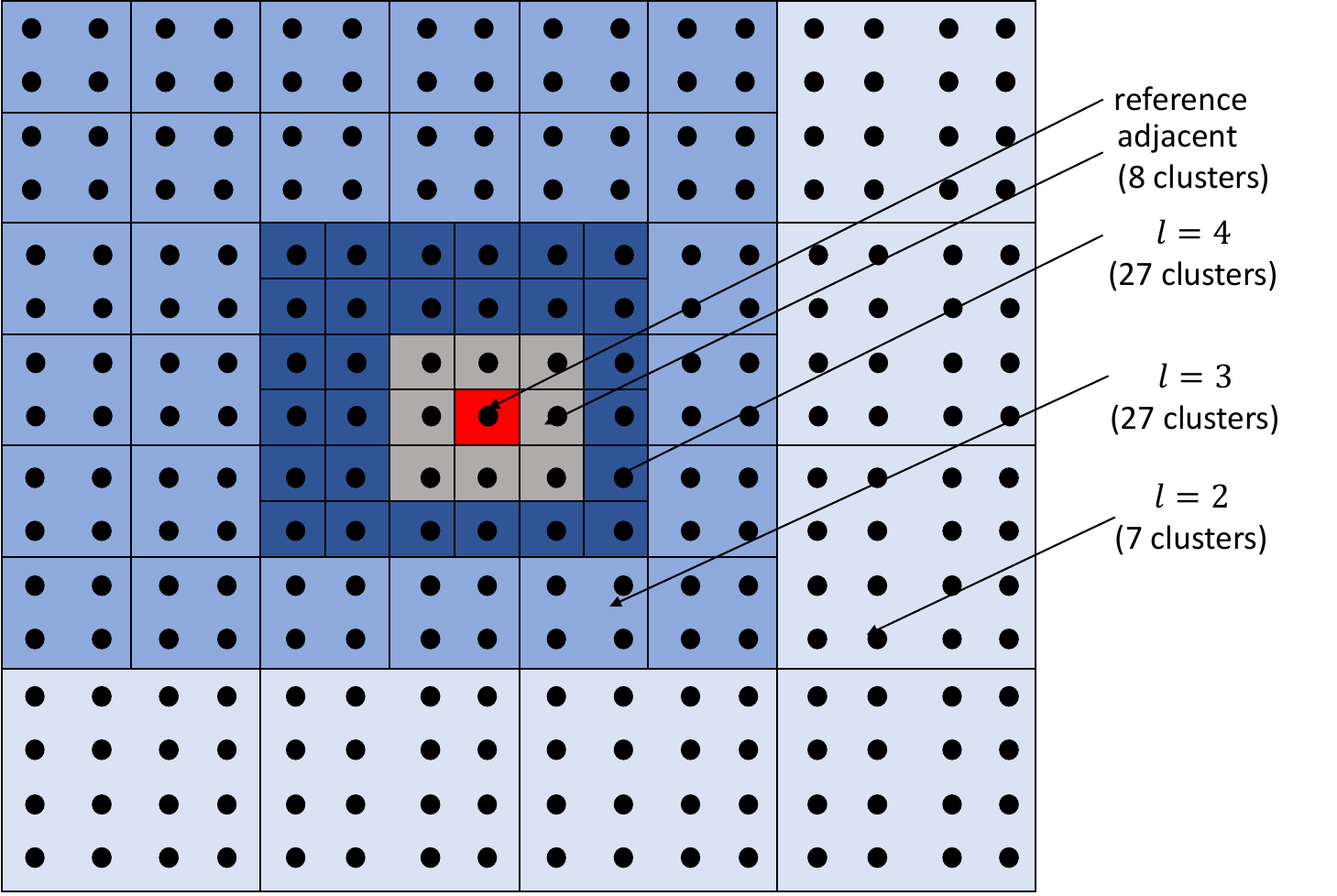}
\caption{\textbf{Hierarchical decomposition of a 2-d kernel}: The combinatorics of the hierarchical decomposition for a 2-d kernel, with $N=16$, ${\ell_{\max}}= \log_2 16=4$. The red square is the reference, and the selected $2$-clusters, $3$-clusters and $4$-clusters have been highlighted.  Can be extended to $3$ dimensions in a straightforward manner.}
    \label{fig:blockstructure2D}
\end{figure}
A decomposition of the kernel results by repeating the above construction for every reference particle, and defining matrices $\m{K}^{(l)}$ that only include the interactions from clusters selected at level $l$. We then obtain
\begin{equation}
    \m{K}= \sum_{l=2}^{\ell_{\max}} \m{K}^{(l)}+ \m{K}_{\mathrm{ad}},
\end{equation}
 one element per hierarchical level plus an extra for the residual interactions not included at any of the levels. We then block-encode $\m{K}$ using the weighted sparse encoding of Lemma~\ref{lem:weighted-sparse-encoding} over this hierarchical decomposition.

 In the coming sections we unpack this theory more precisely.

\subsubsection{Clusters, admissible blocks and fresh blocks}
We now give some key terminology and theory that is needed to describe the hierarchical block-encoding. Recall that for each spatial direction we index the points with an indexing set $I=\{0,1,2, \dots, 2^{\ell_{\max}}-1\}$ for a fixed ${\ell_{\max}}$. An arbitrary point $\v{x}$ in space is indexed by a triple $\v{j} = (j_1,j_2,j_3) \rightarrow \v{x}_{\v{j}}$, and the full indexing set is therefore $I^{3}$ for a cubic lattice. We shall also define 
\begin{align}
    \v{j}^l= (j^l_1,j^l_2,j^l_3),
\end{align} 
where each $j^l_k$ is an $l\leq {\ell_{\max}}$-bit number, i.e., $j^l_k \in \{ (00\dots 0), (00\dots 01),\dots, (11 \dots 1)\}$, however for clarity we write $0,1, \dots 2^{\ell}-1$ for the corresponding strings $j^l_k$.

 We coarse-grain the set $I^{3}$ and group indices into ``clusters" at different scales.  More precisely, we split the set $I^{3}$ up into \emph{clusters} of size $2^{3({\ell_{\max}}-l)}$, where $l=0,1,\dots {\ell_{\max}}$ is the \emph{level} of the hierarchical decomposition. We denote a single cluster at level $l$ as $\tau^l_{\v{j}^l}$. The triple $\v{j}^l$ is the cluster label at level $l$, and the elements in this cluster are given by 
\begin{equation}
\label{eq:levellcluster}
    \tau^l_{\v{j}^l} := \{ 2^{{\ell_{\max}}-l}\v{j}^l + \v{j}^{\bar{l}} : 0 \le \v{j}^{\bar{l}} \le 2^{{\ell_{\max}}-l}-1\}.
\end{equation}
The inequalities are understood to apply to each of the components of the vector $\v{j}^{\bar{l}}$.
In simpler terms, the cluster is obtained by taking the union of all points $\v{j}$ for which the $l$ most significant bits are fixed to be $\v{j}^l$ and the remaining $\ell_{\max}-l$ bits are varied over.
 The following properties can be verified for any cluster. 
\begin{lemma}[Cluster properties]\label{lem:cluster-properties}
    Consider a cubic lattice of points that are indexed by $I^{3} = \{0,1,2,\dots, 2^{\ell_{\max}}-1\}^{3}$. Let $\tau^l_{\v{j}^l}$ be a cluster at level $l$, defined in Eq.~\eqref{eq:levellcluster}. Then the following hold:
    \begin{enumerate}
        \item The number of elements in a level--$l$ cluster is $|\tau^l_{\v{j}^l}| = 2^{3({\ell_{\max}}-l)}$ for any index $\v{j}^l$.
        \item The largest separation in the $L^\infty(\mathbb{N})$ norm between two points inside $\tau^l_{\v{j}^l}$ is $2^{\ell_{\max}-l}-1$, namely the diameter of any cluster at level $l$.
        \item Any point $\v{j}\in I^{3}$ belongs to a unique cluster $\tau^l_{\v{j}^l}$ at level $l$, where
        \begin{equation}
            \v{j} = 2^{{\ell_{\max}}-l} \v{j}^l + \v{j}^{\bar{l}},
        \end{equation}
        Therefore, for any $k=1,2,3$ writing components as binary strings, we have that $j_k = j_k^l j_k^{\bar{l}}$, namely $j_k^l$ in binary is given by the leading $l$ bits of $j_k$, and $j_k^{\bar{l}}$ is the remaining string of $ {\ell_{\max}}-l$ bits of $j_k$ in binary notation.
        \item The number of distinct clusters at level-$l$ is $2^{3l}$.
        \item The clusters at each level $l$ define a partition of $I^{3}$: namely, $\cup_{\v{j}^l} \tau^l_{\v{j}^l} = I^{3}$,
         where the union is over all $\v{j}^l = (j_1^l,j_2^l,j^l_3)$ with $j_k^l = 0, \dots, 2^l-1$, and all clusters are disjoint.
        \item Subsequent levels of the hierarchy have the following nesting relation for clusters: 
        \begin{equation}
            \bigcup_{\v{b}:\, b_k \in \{0,1\}} \tau^l_{2\v{j}^{l-1}+\v{b}}  = \tau^{l-1}_{\v{j}^{l-1}},
        \end{equation}
        and hence the set of clusters over all levels forms an octree in $3$--dimensions, with an edge denoting set-inclusion.
    \end{enumerate}
\end{lemma}
\begin{proof}
To see that (1) is true, note that the cluster is obtained by ranging $\v{j}^{\bar{l}}$ over all possible values. Since each component can take on $2^{\ell_{\textrm{max}}-l}$ values the result follows.\\
To see that (2) is true, note that for any two points in a given cluster we have $\|\v{j}-\v{j}'\|_{\infty} = \|\v{j}^{\bar{l}} - \v{j}^{'\bar{l}}\|_\infty$. Here $\|\v{v}\|_\infty = \max_i |v_i|$, is the largest normed element of the vector $\v{v}$.  This is maximized by, e.g. $\v{j}^{'\bar{l}} = \v{0}$ and $\v{j}^{\bar{l}} = (2^{\ell_{\textrm{max}}-l}-1,0,0)$, from which the result follows.\\
To see that (3) is true, note that by Euclid's division lemma the decomposition
\begin{equation}
            j_k = 2^{{\ell_{\max}}-l} j_k^l + j_k^{\bar{l}},
        \end{equation}
exists and is unique for any component $k$ and any choice of $2^{\ell_{\textrm{max}}-l}$. Therefore, the point $\v{j}$ belongs to a unique cluster at any level $l$. In binary notation $2^{{\ell_{\max}}-l} j_k^l$ consists of $l$ leading digits of $j_k$ followed by $\bar{l}$ zeroes. The claimed binary decomposition follows.\\
To see that (4) and (5) are correct, note that the total number of points is $2^{3 \ell_{\textrm{max}}}$, however each point is an element of a unique cluster. Moreover, it is clear from the definition that $\tau^l_{\v{j}^l}$ and $\tau^l_{\v{j}^{'l}}$ are disjoint whenever $\v{j}^l\ne \v{j}^{'l}$. Therefore the set of clusters forms a partition of the full set of points and since each cluster at level $l$ has exactly $2^{3(\ell_{\textrm{max}}-l)}$ points, it follows that there must be $2^{3l}$ clusters at level $l$.\\
To see that (6) is true, we first note that a point $\v{j}$ is in $\tau^{l-1}_{\v{j}^{l-1}}$ if and only if it takes the form
\begin{align}
    \v{j} = 2^{{\ell_{\max}}-l+1} \v{j}^{l-1} + \v{j}^{\overline{l-1}} =2(2^{{\ell_{\max}}-l}) \v{j}^{l-1} + \v{j}^{\overline{l-1}} .
\end{align}
However, now let $b_k$ be the leading binary digit in $j_k^{\overline{l-1}}$ so that $j_k^{\overline{l-1}} = 2^{\ell_{\textrm{max}}-l} b_k + j_k^{\bar{l}}$, with $j_k^{\bar{l}}$ being the remaining $\ell_{\textrm{max}}-l$ bits of the string. Therefore $\v{j} = 2^{{\ell_{\max}}-l}(2\v{j}^{l-1} +\v{b})+\v{j}^{\bar{l}} $. Ranging over the full cluster at level $l-1$ involves ranging over all values for $\v{b}$ and $j_k^{\bar{l}}$. Separating these two out, and noting the definition of a cluster at level $l$, the result follows.
\end{proof}

The coarsest-grained cluster is $\tau^0_{\v{0}} = I$ at level $l=0$, which is the ``root" of the hierarchical tree of clusters. We also note that the size of clusters in $3$ dimensions grows cubically in $2^{({\ell_{\max}}-l)}$. This is natural since the vectors $\v{j}^{\bar{l}}$ comprise ``internal" indices that uniquely label a point within a fixed cluster, and the components of $\v{j}^{\bar{l}}$ are arbitrary $(\ell_{\max}-l)$--bit numbers.  This scaling becomes important in the weighted sparse encoding subroutine.

The indices $\v{j}\in I^{3}$ are associated to vectors $\v{x}_{\v{j}}$, but they also are matrix indices $\m{K}=(\m{K}_{\v{j},\v{j}'})$ where $\v{j}$ and $\v{j}'$ both lie in $I^{3}$. For matrices, the coarse-graining implies we form \emph{blocks} of the matrix $\m{K}$. We do this at a fixed level $l$ so that such a block $b$ is entirely described by
\begin{equation}
    b =(\tau^l_{\v{j}^l} , \tau^l_{\v{j}^{'l}}),
\end{equation}
for indexing vectors $\v{j}^l,\v{j}^{'l}$ with components in $\{0,1,\dots, 2^l-1\}$. This $b$ labels the block of any matrix $\m{K}$ with indices $(\v{j},\v{j}')$ and $\v{j}\in \tau^l_{\v{j}^l}$ and $\v{j}' \in \tau^l_{\v{j}^{'l}}$. Recall, we use primes to label column indices, for clarity with rectangular matrices.

In order to handle kernels that are dense, but decaying off the diagonal (such as the electric field kernel) we need the clusters $\tau^l_{\v{j}^l}$ and $\tau^l_{\v{j}^{'l}}$ that form a block to be sufficiently far apart. This \emph{admissibility condition} is that
\begin{equation}
\label{eq:clusterdistance}
    \|\v{j}^l - \v{j}^{'l} \|_\infty \ge 2,
\end{equation}
and if this holds we say that block $b$ is an \emph{admissible block}. Such blocks only occur for $l \ge 2$.

We now use the hierarchical method, and the concept of an admissible block to iteratively break $\m{K}$ up into a sum of matrices as
\begin{equation}
    \m{K} = \m{K}_{\mathrm{ad}}+\sum_{l=2}^{{\ell_{\max}}} \m{K}^{(l)},
\end{equation}
where at level $l$ the matrix $\m{K}^{(l)}$ is composed of non-zero admissible blocks at level $l$ and zero elsewhere. The matrix $\m{K}_{\mathrm{ad}}$ is a constant sparsity matrix containing ``adjacent" clusters, which we explain later. The algorithm to construct this splitting of $\m{K}$ is to start at level $l=2$, and separate out from it the admissible blocks at level $l=2$. This leaves us with the matrix $\m{K}-\m{K}^{(2)}$, and we move to level $l=3$ and then subtract off the admissible blocks at this level. This is continued until we reach $l={\ell_{\max}}$. We now define a \emph{fresh block} at level $l$ as one that is admissible and has not been subtracted off at some level $y <l$ of the hierarchy (this is not standard terminology, but it is useful here for clarity). Therefore the non-zero components of $\m{K}^{(l)}$ are composed of precisely those fresh blocks at level~$l$.
\begin{equation}
    \mbox{ Matrix $\m{K}^{(l)}$} : \mbox{ all fresh blocks at level }l.
\end{equation}
Therefore, we just need to compute the indexing for fresh blocks at a given level. It turns out that the boundary conditions affect the fresh blocks at large scales. No fresh blocks occur at $l=0$ and $l=1$, but they do appear at $l=2$. In the next section we describe these $l=2$ blocks, and then proceed to give the prescription for $l\ge 3$.

\subsubsection{Boundary conditions and fresh blocks at $l=2$}\label{sec:Kl=2}

For the spatial directions we assume periodic boundary conditions:
\begin{equation}
    \m{K}_{\v{j} + N_x \v{e}_k, \v{j}'} = \m{K}_{\v{j} , \v{j}'+ N_x \v{e}_k} = \m{K}_{\v{j}, \v{j}'},
\end{equation}
for any direction $\v{e}_k$. However, since we work with registers with $N_x$ points in each direction, this is attained by computing modulo $N_x$ in the indexing registers. Therefore, we identify
\begin{equation}
    j_k \cong j_k + 2^{\ell_{\max}}. 
\end{equation}
Note that we can decompose any component as
\begin{equation}
    j_k = 2^{{\ell_{\max}}-l} j_k^l + j_k^{\bar{l}},
\end{equation}
and so $j_k^l$ is the quotient of $j_k$ under division by $2^{{\ell_{\max}}-l}$. If we take the quotient of $j_k + 2^{\ell_{\max}}$ by $2^{{\ell_{\max}}-l}$ then we get $j_k^l + 2^l$. Therefore, identifying the two points implies periodicity of $2^l$ for the \emph{clusters} at level $l$:
\begin{equation}
    j_k^l \cong j_k^l + 2^l.
\end{equation}

This periodicity affects the algorithm for identifying the fresh blocks at level $l=2$, since clusters that are distance $3$ away according to Eq.~\eqref{eq:clusterdistance} become a distance $1$ away once we impose periodicity.

We can without loss of generality consider $\v{j}^l = \v{0}$, and so the admissibility condition becomes $\|\v{j}^{'l}\| \ge 2$. However, at $l=2$ it is readily seen that only $\|\v{j}^{'l}\| = 2$ occurs. The full set is partitioned into $2^{ 3 \times 2}=64$ clusters, and the clusters that are not admissible with $\v{j}$ are those in a unit ball in the $L^\infty(\mathbb{N})$ distance. This contains $3^3=27$ clusters which are excluded -- corresponding to the number of points at distance $0$ or $1$ away from a point in the bulk of a regular cubic lattice of points with spacing equal to $1$. These excluded clusters are resolved at higher levels $l\ge 3$. The remaining $64-27=37$ clusters at a distance $2$ provide the admissible, and hence fresh blocks at level $2$. Therefore, at $l=2$ the matrix $\m{K}^{(2)}$ is $37$-block sparse with $37 \times 2^{3({\ell_{\max}}-2)}$ non-zero elements in each row/column. Note that while these blocks contain a large number of points, the key thing is that the magnitude of the electric term is ``small" since the separation is large on the scale of the full system.

The impact of the periodic boundary conditions for $l\ge 3$ is straightforward: we work modulo~$2^l$. This includes that all clusters are contained in the same number of fresh blocks, since there are no boundary conditions to truncate the indexing, and the resolution is sufficiently fine such that the periodicity does not eliminate potential fresh blocks.

\subsubsection{The adjacent matrix $\m{K}_{\mathrm{ad}}$}\label{sec:Kad}
We next turn to the adjacent matrix. This arises from those blocks that have not been included at any level $l$. When we get to level $l={\ell_{\max}}$, we are down to clusters made of individual points. The clusters that are not admissible are the points adjacent to the reference point. The corresponding blocks are then simply the matrix elements of $\m{K}$ on the diagonal and those adjacent to it. Therefore, for any row-index $\v{j}$ the column indices containing the non-zero elements of the adjacent matrix are simply the set
\begin{equation}
   \m{K}_{\mathrm{ad}} = (\m{K}_{\v{j},\v{j}'}) \mbox{ such that }  \|\v{j}-\v{j}'\|_\infty \le 1 .
\end{equation}
Also note that for any fixed row-index $\v{j}$ the cardinality of the associated column indices $\v{j}'$ obeying this condition is $3^3=27$. Therefore the sparsity parameter for the adjacent matrix $\m{K}_{\mathrm{ad}}$ is $27$. The largest element is $1$, by the choice of cutoff.

\subsubsection{The hierarchical partition of the kernel}
\label{sec:Coulombkernelhierarchical}
The preceding sections have described how the kernel $\m{K}$, or indeed a similar kernel, can be decomposed using the metric of the space involved (here a 3--dimensional torus). This method involves labeling blocks as `fresh' at every level $l$, and then identifying the remaining terms as those comprising an `adjacent' matrix. This procedure results in a decomposition 
\begin{equation}
    \m{K} = \m{K}_{\mathrm{ad}} + \sum_{l=2}^{\ell_{\max}} \m{K}^{(l)}.
\end{equation}
We obtain a \emph{valid (hierarchical) partition} of $\m{K}$ and every nonzero element of the matrix $\m{K}$ appears in one and only one of the matrices on the right-hand side and the non-zero entries of matrices $\m{K}_{\mathrm{ad}}, \m{K}^{(l)}$ are precisely these terms. This in turn allows us to construct the weighted sparse encoding, and as we shall see shortly it gives an asymptotically optimal block-encoding scale factor.

\newpage
\subsection{General HBE for two-kernel, dense rectangular matrices}
Note that the $\m{K}$ that generates our nonlinear term is a dense rectangular matrix, and the terms appearing in $\m{K}$ can be slightly awkward. To handle this neatly we give a theorem that is sufficiently flexible to accommodate all cases of interest here.

\begin{theorem}[Hierarchical block-encoding for a two-kernel rectangular interaction]
\label{thm:two-kernel-hbe}
Let $N_x = 2^{\ell_{\max}}$ with $\ell_{\max} \ge 2$, and let \(\mathcal A,\mathcal A',\mathcal A''\) be finite index sets. Consider the rectangular matrix $\m K \in\mathbb C^{N_x^3|\mathcal A| \times N_x^6|\mathcal A'||\mathcal A''|} $, with entries
\begin{equation}
\label{eq:two-kernel-rectangular-K}
\m K_{\v j,\v\alpha \mid \v j',\v\alpha';\v j'',\v\alpha''}
=
\sigma_{\v j,\v j''}
\sum_{a=1}^3
K^a_{\v j,\v j'}\,
W^a_{\v\alpha\mid \v\alpha',\v\alpha''}.
\end{equation}
Here $\v j,\v j',\v j''\in\mathbb{Z}_{N_x}^3$, \(\v\alpha\in\mathcal A\), and
\((\v\alpha',\v\alpha'')\in\mathcal A'\times\mathcal A''\). Assume that \(\sigma\) and each \(K^a\) are spatial kernels on the periodic cubic lattice, and assume that \(\sigma\) and \(K^a\) admit monotone decaying envelopes
\begin{equation}
    |\sigma_{\v j,\v k}|
    \le
    \mathcal E_\sigma\!\left(\rho(\v{j},\v{k})\right), \quad 
    |K^a_{\v j,\v k}|
    \le
    \mathcal E_{K,a}\!\left(\rho(\v{j},\v{k})\right),
\end{equation}
where $\mathcal E_\sigma$ and $\mathcal E_{K,a}$ are real-valued functions and $\rho(\v{j},\v{k})$ is any metric measuring the separation of indices $\v{j}$ and $\v{k}$. Consider hierarchical decompositions to both spatial kernels:
\begin{equation}
    \sigma =\sum_{\lambda=1}^{\ell_{\max}}\sigma^{(\lambda)},\qquad K^a=\sum_{\ell=1}^{\ell_{\max}} K^{a,(\ell)}.
\end{equation}
Here the label $1$ is re-purposed to denote the adjacent matrix, and the labels
$\ell,\lambda\ge 2$ denote the non-trivial levels of the hierarchy. Define $q_\ell := 2^{\ell_{\max}-\ell},$ for $ \ell\ge 2$, and $q_1=1$, and for each $\ell$ define the spatial sparsity number $D_\ell$ as
\begin{equation}
\label{eq:Dell-def}
D_\ell :=\begin{cases}
27,
& \ell=1,\\[2mm]
37\times 2^{3(\ell_{\max}-2)},
& \ell=2,\\[2mm]
189\times 2^{3(\ell_{\max}-\ell)},
& 3\le \ell\le \ell_{\max}.
\end{cases}
\end{equation}
Moreover, define max-entry bounds $w_a:= \|W^a\|_{\max}$, $\bar\sigma_\lambda:=\|\sigma^{(\lambda)}\|_{\max}$ and $\bar K_{a,\ell}:=\|K^{a,(\ell)}\|_{\max}$. 

Let $\dr(W^a)$ and $\dc(W^a)$ denote the row and column sparsities of $W^a$, and for each pair of spatial levels \((\lambda,\ell)\), define the kernel intersection number
\begin{equation}
\label{eq:Gamma-def}
\Gamma^a_{\lambda,\ell}
:=
\max_{\v j',\v j''\in\mathbb{Z}_{N_x}^3}
\left|
\left\{
\v j\in\mathbb{Z}_{N_x}^3
:
\sigma^{(\lambda)}_{\v j,\v j''}
K^{a,(\ell)}_{\v j,\v j'}
\neq 0
\right\}
\right| \le \min\{D_\lambda,D_\ell\}.
\end{equation}
Assume coherent oracles for the hierarchical support of
\(\sigma^{(\lambda)}\) intersected with that of \(K^{a,(\ell)}\), oracles for \(W^a\), and value oracles that return \(b\)-bit approximations to
$\sigma_{\v j,\v j''},K^a_{\v j,\v j'},W^a_{\v\alpha\mid\v\alpha',\v\alpha''}$.
Assume also access to a state-preparation oracle over the labels
\((a,\lambda,\ell)\) with amplitudes proportional to the square roots of the weights defined below.

Then we can construct a rectangular block-encoding of \(\m K\) with scale factor
\begin{equation}
\label{eq:alpha-two-kernel-HBE}
\alpha_{\m K}\le \sum_{a=1}^3\sum_{\lambda=1}^{\ell_{\max}}\sum_{\ell=1}^{\ell_{\max}}\bar\sigma_\lambda\,\bar K_{a,\ell}\,w_a\,\sqrt{D_\lambda D_\ell\,\Gamma^a_{\lambda,\ell}\,\dr(W^a)\dc(W^a)}.
\end{equation}
For \(\lambda,\ell\ge 2\), and taking $\rho$ to be the $L^\infty(\mathbb{N})$--norm on the periodic lattice $\mathbb{Z}_{N_x}^3$, the decay assumptions imply the bounds
\begin{equation}
\label{eq:alpha-two-kernel-HBE-L-infty}
\alpha_{\m K}\le \sum_{a=1}^3\sum_{\lambda=1}^{\ell_{\max}}\sum_{\ell=1}^{\ell_{\max}}\mathcal E_\sigma(q_\lambda-1)\,\mathcal E_{K,a}(q_\ell-1)\,w_a\,\sqrt{D_\lambda D_\ell\,\Gamma^a_{\lambda,\ell}\,\dr(W^a)\dc(W^a)}.
\end{equation}
The resulting block-encoding uses
\begin{equation}
O\!\left(b +\log N_x+\log|\mathcal A|+\log|\mathcal A'|+\log|\mathcal A''|+\log\ell_{\max}+ \operatorname{polylog}\frac{\alpha_{\m K}}{\epsilon}
\right)
\end{equation}
auxiliary qubits and primitive gates, in addition to one controlled use of the relevant sparse-access and value oracles per LCU term. Here $\epsilon = O(2^{-b})$, and the rectangular block-encoding is understood with sufficient padding ancillae so that the total input and output Hilbert-space dimensions agree.
\end{theorem}
The above theorem is the most general form we need and we recover important applications as special cases. In the use of the theorem, it is assumed that the bounding decay profiles $\mc{E}_\sigma$ and $\mc{E}_{K,a}$ are explicitly known, and these can be used to replace $\bar{\sigma}_\lambda$ and $\bar{K}_{a,\ell}$ in both the scale factor formula, and any oracle construction for the matrix entries. The distance function $\rho(\v{j},\v{k})$ is kept general. For the torus setting, the correct definition is
\begin{equation}
    \rho(\v{j},\v{k}) = \min_{ \v{n} \in \mathbb{Z}^3} \| \v{j}-\v{k} - \v{n}N_x\|,
\end{equation}
for some vector norm $\|\cdot \|$, and which handles periodicity correctly. In the hierarchical decomposition, this was previously incorporated into the $\ell=2$ analysis, since this is the first level for which periodicity becomes relevant.

Note also that the intersection number $\Gamma^a_{\lambda,\ell}$ can have significant impact in the context of block-encoding nonlinear interactions, as we highlight below. Using the universal bound \(\Gamma^a_{\lambda,\ell}\le \min\{D_\lambda,D_\ell\}\), one obtains the simpler expression
\begin{equation}
\label{eq:alpha-two-kernel-HBE-min}
\alpha_{\m K}\le \sum_{a=1}^3 \sum_{\lambda=1}^{\ell_{\max}} \sum_{\ell=1}^{\ell_{\max}} \bar\sigma_\lambda\,\bar K_{a,\ell}\, w_a\,\sqrt{D_\lambda D_\ell\,\min\{D_\lambda,D_\ell\}\,\dr(W^a)\dc(W^a)}.
\end{equation}

\begin{proof}
First decompose both spatial kernels into their hierarchical decompositions:
\begin{equation}
    \sigma
    =
    \sum_{\lambda=1}^{\ell_{\max}}
    \sigma^{(\lambda)},
    \qquad
    K^a
    =
    \sum_{\ell=1}^{\ell_{\max}}
    K^{a,(\ell)}.
\end{equation}
Substituting these decompositions into \eqref{eq:two-kernel-rectangular-K} gives
\begin{equation}
    \m K
    =
    \sum_{a=1}^3
    \sum_{\lambda=1}^{\ell_{\max}}
    \sum_{\ell=1}^{\ell_{\max}}
    \m K^{a,\lambda,\ell},
\end{equation}
where
\begin{equation}
\label{eq:K-a-lambda-ell}
\m K^{a,\lambda,\ell}_{\v j,\v\alpha
\mid
\v j',\v\alpha';\v j'',\v\alpha''}
=
\sigma^{(\lambda)}_{\v j,\v j''}
K^{a,(\ell)}_{\v j,\v j'}
W^a_{\v\alpha\mid\v\alpha',\v\alpha''}.
\end{equation}

We now bound the sparsity and max-entry norm of each summand. By construction of the hierarchical decomposition, for fixed \(\v j\), the number of possible \(\v j''\) such that
\(\sigma^{(\lambda)}_{\v j,\v j''}\neq 0\) is at most \(D_\lambda\), which was previously shown to have the level-dependent values stated. Similarly, the number of possible \(\v j'\) such that
\(K^{a,(\ell)}_{\v j,\v j'}\neq 0\) is at most \(D_\ell\). Finally, for fixed \(\v\alpha\), the number of possible pairs
\((\v\alpha',\v\alpha'')\) such that
\(W^a_{\v\alpha\mid\v\alpha',\v\alpha''}\neq 0\) is at most \(\dr(W^a)\). Hence
\begin{equation}
\label{eq:row-sparsity-K-a-lambda-ell}
\dr(\m K^{a,\lambda,\ell})
\le
D_\lambda D_\ell \dr(W^a).
\end{equation}

For the column sparsity, fix a column index $(\v j',\v\alpha';\v j'',\v\alpha'')$. A row index \((\v j,\v\alpha)\) contributes if and only if $\sigma^{(\lambda)}_{\v j,\v j''} K^{a,(\ell)}_{\v j,\v j'} \neq 0$
and $W^a_{\v\alpha\mid\v\alpha',\v\alpha''}\neq 0$.
The number of possible spatial row labels \(\v j\) is therefore at most
\(\Gamma^a_{\lambda,\ell}\), by definition \eqref{eq:Gamma-def}. For each such \(\v j\), the number of possible output velocity labels \(\v\alpha\) is at most \(\dc(W^a)\). Thus
\begin{equation}
\label{eq:column-sparsity-K-a-lambda-ell}
\dc(\m K^{a,\lambda,\ell})
\le
\Gamma^a_{\lambda,\ell}\dc(W^a).
\end{equation}

Assume that the distance measure is the $L^{\infty}$--norm on the periodic lattice of indices. For any level $\lambda$ the indices $\v{j}$ and $\v{k}$ must lie in an admissible block, and therefore at least separated by one cluster of side $q_\lambda = q^{\ell_{\max}-\lambda}-1$ in this norm. Therefore,
\begin{equation}
    \rho(\v{j},\v{k}) \ge q_\lambda -1\quad \mbox{ at level $\lambda$}.
\end{equation}
The assumption of monotonicity implies that the max-entry norm of the summand is bounded by
\begin{equation}
\label{eq:max-entry-K-a-lambda-ell}
\|\m K^{a,\lambda,\ell}\|_{\max}
\le\bar\sigma_\lambda\,\bar K_{a,\ell}\,w_a \le \mc{E}_\sigma(q_\lambda-1)\mc{E}_{K,a}(q_\ell-1)\,w_a.
\end{equation}
Therefore, by the rectangular sparse block-encoding lemma, each
\(\m K^{a,\lambda,\ell}\) has a block-encoding with normalization
\begin{align}
\alpha_{a,\lambda,\ell}
\coloneqq &\;
\|\m K^{a,\lambda,\ell}\|_{\max}
\sqrt{
\dr(\m K^{a,\lambda,\ell})
\dc(\m K^{a,\lambda,\ell})
}
\nonumber\\
\le &\; 
\bar\sigma_\lambda\,
\bar K_{a,\ell}\,
w_a\,
\sqrt{
D_\lambda D_\ell\,
\Gamma^a_{\lambda,\ell}\,
\dr(W^a)\dc(W^a)
} \nonumber \\
\le &\; 
\mc{E}_\sigma(q_\lambda-1)\mc{E}_{K,a}(q_\ell-1)
w_a\,
\sqrt{
D_\lambda D_\ell\,
\Gamma^a_{\lambda,\ell}\,
\dr(W^a)\dc(W^a)
}.
\end{align}

Finally, we apply the weighted sparse-encoding construction / LCU construction, to the decomposition
\begin{equation}
    \m K
    =
    \sum_{a,\lambda,\ell}
    \m K^{a,\lambda,\ell}.
\end{equation}
Prepare the state
\begin{equation}
    \textsc{prep}\ket{0}
    =\sum_{a=1}^3 \sum_{\lambda=1}^{\ell_{\max}} \sum_{\ell=1}^{\ell_{\max}}
    \sqrt{\frac{\alpha_{a,\lambda,\ell}}{\alpha_{\m K}} }\ket{a,\lambda,\ell},
\end{equation}
where
\begin{equation}
    \alpha_{\m K}
    =
    \sum_{a,\lambda,\ell}
    \alpha_{a,\lambda,\ell}.
\end{equation}
The LCU block-encoding then block-encodes
\begin{equation}
    \sum_{a,\lambda,\ell}
    \m K^{a,\lambda,\ell}
    =
    \m K
\end{equation}
with normalization \(\alpha_{\m K}\), giving
\eqref{eq:alpha-two-kernel-HBE}. The error \(\epsilon\) is obtained by choosing the arithmetic precision and individual sparse block-encoding errors so that the total LCU error is at most \(\epsilon\). This proves the claim.
\end{proof}

\subsection{Applications of the hierarchical block-encoding theorem}
We now discuss some applications of the above theorem. We first unpack the scaling behaviour as a function of decay profiles, then specialize to a natural convolutional case that arises for quadratic nonlinear interactions, and then recover previous square kernel results.

\paragraph{Scaling as a function of power-law decays}
The scaling of $\alpha_{\m{K}}$ in terms of the parameter $N_x$ is of importance. Consider Theorem~\ref{thm:two-kernel-hbe}, and suppose that the velocity-space factors are independent of \(N_x\). 

Assume that kernels have the upper bound decay profiles
\begin{equation}   |\sigma_{\v{j},\v{k}}| \le \frac{C_\sigma}{(1+\rho(\v{j},\v{k}))^{p_\sigma}},\quad 
    |K^a_{\v{j},\v{k}}| \le \frac{C_{K,a}}{(1+\rho(\v{j},\v{k}))^{p_{K,a}}}
\end{equation}
where we have the periodic $L^\infty(\mathbb{N})$ norm
\begin{equation}
    \rho(\v{j},\v{k})= \min_{\v{n}\in \mathbb{Z}^3}\|\v{j}-\v{k} - \v{n}N_x\|_\infty.
\end{equation}
Note that Theorem~\ref{thm:informal-two-kernel-hbe} uses the Euclidean norm, however since $\|\v{j}\|_\infty \le \|\v{j}\| \le \sqrt{3} \|\v{j}\|_\infty$, in 3-d, the complexity analysis is unchanged if we use the $L^\infty(\mathbb{N})$ norm.

Putting this into the above theory gives
\begin{align}
\alpha_{\m K} &= \tilde{O} \left( \sum_{a=1}^3\sum_{\lambda=1}^{\ell_{\max}}\sum_{\ell=1}^{\ell_{\max}} (q_\lambda^{-p_{\sigma}})\,(q_{\ell}^{-p_{K,a}})\,\sqrt{D_\lambda D_\ell\,\Gamma^a_{\lambda,\ell}} \right) \nonumber \\
&= \tilde{O} \left( \sum_{a=1}^3\sum_{\lambda=1}^{\ell_{\max}}\sum_{\ell=1}^{\ell_{\max}} (q_\lambda^{3/2-p_{\sigma}})\,(q_{\ell}^{3/2-p_{K,a}})\,\sqrt{\Gamma^a_{\lambda,\ell}} \right).
\end{align}
In the absence of structure we can use the bound $\Gamma^a_{\lambda,\ell} \le \min \{D_\lambda,D_\ell \}$. Therefore,
\begin{align}
\alpha_{\m K} &= \tilde{O} \left( \sum_{a=1}^3\sum_{\lambda=1}^{\ell_{\max}}\sum_{\ell=1}^{\ell_{\max}} (q_\lambda^{3/2-p_{\sigma}})\,(q_{\ell}^{3/2-p_{K,a}})\,\min \{q_\lambda^{3/2}, q_\ell^{3/2} \}\right) \nonumber \\
 &= \tilde{O} \left( \sum_{a=1}^3\sum_{\lambda=1}^{\ell_{\max}}\sum_{\ell=1}^{\ell_{\max}} (q_\lambda^{3/2-p_{\sigma}})\,(q_{\ell}^{3/2-p_{K,a}})\,(\min \{q_\lambda, q_\ell \})^{3/2}\right) \nonumber \\
 &= \tilde{O} \left( \sum_{a=1}^3\sum_{x,y}  x^A y^{B_a}\,(\min \{x, y \})^{3/2}\right) ,
\end{align}
where $x = q_\lambda$, $y=q_\ell$, $A= 3/2-p_\sigma$, $B_a=3/2-p_{K,a}$, and the sum is over the values attained from the discrete level indices. We consider $x\le y$ and $x > y$ separately,
\begin{equation}
    \sum_{x,y}  x^A y^{B_a}\,(\min \{x, y \})^{3/2} = S_a= S_a^\le + S_a^>,
\end{equation}
where
\begin{align}
    S_a^\le &= \sum_y y^{B_a} \sum_{x\le y} x^{A+3/2} \nonumber \\
    &= O\left (\sum_y y^{B_a} y^{\max\{0,A+3/2\}} \right) = O\left (   (N_x)^{\max \{ 0, B_a+\max\{0,A+3/2\}\}} \right) \nonumber \\
        &= O\left (   (N_x)^{\max \{ 0, B_a, B_a+A+3/2\}} \right) \nonumber \\
        &= O\left (   (N_x)^{\max \{ 0, 3/2-p_{K,a}, 9/2-p_\sigma-p_{K,a}\}} \right) .
\end{align}
For the other sum, we have
\begin{align}
    S_a^> &= \sum_x x^{A} \sum_{y<x} y^{B_a+3/2} \nonumber \\
    &= O \left ( \sum_x   x^{A + \max \{0,B_a+3/2\}} \right ) \nonumber \\
     &= O \left ((N_x)^{\max \{0,A + \max \{0,B_a+3/2\} \}} \right ) \nonumber \\
     &= O \left ((N_x)^{\max \{0,A, A + B_a+3/2 \}} \right ) \nonumber \\
     &= O \left ((N_x)^{\max \{0,3/2-p_\sigma, 9/2 -p_\sigma - p_{K,a}  \}} \right ) .
\end{align}
Combining the two sum contributions we have that
\begin{equation}
    \alpha_{\m K}
    =
    \widetilde O\!\left(
    \sum_{a=1}^3
    N_x^{r_a}
    \right),
\end{equation}
where
\begin{equation}
    r_a
    :=
    \max\left\{
    0,\,
    \frac32-p_\sigma,\,
    \frac32-p_{K,a},\,
    \frac92-p_\sigma-p_{K,a}
    \right\}.
\end{equation}
Equivalently, for each component \(a\), the contribution is
polylogarithmic in \(N_x\) whenever
\begin{equation}
    p_\sigma\ge \frac32,
    \qquad
    p_{K,a}\ge \frac32,
    \qquad
    p_\sigma+p_{K,a}\ge \frac92.
\end{equation}
If all three inequalities are strict, then the scale factor is
\(O(1)\). If one or more inequalities is saturated, the scale factor is
polylogarithmic, with at most powers of \(\log N_x\). In particular, if both spatial kernels have the same decay exponent, $p_\sigma=p_{K,a}=p$, then 
\begin{equation}
    \alpha_{\m K}
    =
    \widetilde O\!\left(
    N_x^{\max\{0,\,\frac92-2p\}}
    \right).
\end{equation}
Thus
\begin{equation}
    \alpha_{\m K} = \begin{cases}
        O(1), \quad \mbox{ for } p >9/4 \\
        O(\operatorname{polylog}(N_x)), \quad \mbox{ for } p =9/4 \\
        \widetilde O\!\left(N_x^{\frac92-2p}\right), \quad \mbox{ for } p <9/4 .
    \end{cases}
\end{equation}
This implies that for Coulomb-like decay, $p=2$, that the scale factor has only a mild $\widetilde O\!\left(\sqrt{N_x} \right )$ growth.
\paragraph{HBE for convolutional interactions.}

We next consider nonlinear contributions that arise due to an electric field interaction that decays in space. Recall, for simplicity we work on the torus $\mathbb{T}_L^3$. These are described by a kernel of the form
\begin{equation}
\m{K}_{\val|\val',\val''}(\vx|\vx',\vx'')
=
\delta_L^3(\vx-\vx'')\v{E}(\vx-\vx')\cdot \v{W}_{\val|\val',\val''},
\end{equation}
where the term $\v{W}$ encodes velocity terms that couple to an electric field $\v{E} (\v{x})$, and the first delta function arises due to the quadratic interaction occurring at the same point in real-space. If we now discretize on a cubic lattice in real-space with points $(L/N_x)\v{j} $, then the discretized kernel matrix elements are given by 
\begin{equation}
\m{K}_{\v{j}, \v\alpha | \v{j}', \v\alpha' ; \v{j}'', \v\alpha''} = \delta_{\v{j}, \v{j}''} \sum_{a=1}^3 K^a_{\v{j},\v{j}'} W^a_{\val|\val',\val''}
\end{equation}
Within the context of Theorem~\ref{thm:two-kernel-hbe}, this corresponds to the case $\sigma_{\v{j},\v{j}''} = \delta_{\v{j}, \v{j}''}$ we have $\Gamma^a_{\lambda,\ell} = 1$, and $\bar{\sigma}_{\lambda} = \delta_{\lambda,1}$, and in fact we can replace $D_1$ with $1$ since the matrix is now trivial. A direct application of Theorem~\ref{thm:two-kernel-hbe} gives
\begin{equation}
\label{eq:alpha-K-delta}
\alpha_{\m K}
=
\sum_{a=1}^3
\sum_{\ell=1}^{\ell_{\max}}
\bar K_{a,\ell}\,
w_a\,
\sqrt{
D_\ell\,\dr(W^a)\dc(W^a)
}.
\end{equation}
Equivalently, using the decay envelope for the far-field levels,
\begin{align}
\alpha_{\m K}\le \sum_{a=1}^3 w_a\sqrt{\dr(W^a)\dc(W^a)}
\bigg[
&
 K_{a,1}\sqrt{27}+ K_{0,a}\sqrt{37\,2^{3(\ell_{\max}-2)}}
\,\mathcal E_{K,a}(2^{\ell_{\max}-2}-1)
\nonumber\\
&
+
K_{0,a}
\sum_{\ell=3}^{\ell_{\max}}
\sqrt{189\,2^{3(\ell_{\max}-\ell)}}
\,\mathcal E_{K,a}(2^{\ell_{\max}-\ell}-1)
\bigg].
\end{align}
If $\mc{E}_{K,a}(x) = O(x^{-p})$ then,
\begin{equation}
    \alpha_K = \begin{cases}
        O(1), \quad \mbox{ for }p >3/2 \\
        O(\log N_x), \quad \mbox{ for } p =3/2 \\
        O(N_x^{3/2-p}), \quad \mbox{ for } p <3/2.
    \end{cases}
\end{equation}
Note that, in contrast to the square kernel in the next section, the Coulomb decay rate $p=2$ gives rise to $\alpha_K = O(1)$ in the case of the nonlinear interaction. The origin of this is that the 2-point interaction gives rise to a very small intersection number, which in turn lowers the scaling.
\paragraph{HBE for a square kernel}
Another specialization is to the square kernel case where one has a square matrix $K_{\v{j},\v{j}'}$ that has a decay profile in $|\v{j}-\v{j}'|$. To obtain this we consider
\begin{align}
    \sigma_{\v{j},\v{j}''} &= \delta_{\v{j}'',\v{e_1}} \nonumber\\
    K^a &= K \delta_{a,1} \nonumber \\
    W^a & = \delta_{a,1}, \quad \mbox{ so that }|\mathcal{A}| =|\mathcal{A}'| =|\mathcal{A}''| =1.
\end{align}
This choice of $\sigma$ corresponds to effectively eliminating the $\v{j}''$ index, and we recover the pure, square kernel case in the $\v{j}, \v{j}'$ indices from the matrix $K$. Again, using Theorem~\ref{thm:two-kernel-hbe}, we obtain a block-encoding scale factor
\begin{equation}
    \alpha_{K} \le 27\,\|K^{(1)}\|_{\max} + 37\,2^{3(\ell_{\max}-2)} K_0\,\mathcal E\!\left(2^{\ell_{\max}-2}-1\right) + 189 \sum_{\ell=3}^{\ell_{\max}}
2^{3(\ell_{\max}-\ell)} K_0\,\mathcal E\!\left(2^{\ell_{\max} -\ell}-1\right).
\end{equation}
This gives the scale factor for an arbitrary dense square matrix, and if $\mc{E}(x) = O(x^{-p})$ then 
\begin{equation}
    \alpha_K = \begin{cases}
        O(1), \quad \mbox{ for }p >3 \\
        O(\log N_x), \quad \mbox{ for } p =3 \\
        O(N_x^{3-p}), \quad \mbox{ for } p <3.
    \end{cases}
\end{equation}
This reproduces the results of \cite{nguyen2022block}. In particular, for Coulomb $p=2$ we have $\alpha_{\m K} = O(N_x)$.

\subsection{Block-encoding the Lyapunov-transformed nonlinear term}
Recall that we wish to construct a block-encoding $\m{U}_{\tilde{\m{F}}_2}$ of the Lyapunov-transformed nonlinear term. We explained how this can be obtained from a block-encoding $\m{U}_{\m{K}}$ of the real-space data $\m{K}$ defined earlier, through controlled discrete Fourier transforms. Moreover, if $\m{U}_{\m{K}}$ has scale factor $\alpha_{\m{K}}$ then $\m{U}_{\tilde{\m{F}}_2}$ has scale factor
\begin{equation}
    \alpha_{\tilde{\m{F}}_2} = \Delta_x^{9/2} \alpha_{\m{K}}.
\end{equation}
It remains to show that an efficient block-encoding $\m{U}_{\m{K}}$ can be constructed, and to compute its scale factor $\alpha_{\m{K}}$.

Recall  that $\m{K} = \m{K}^{(c)} + \m{K}^{(n)}$ is composed of the term
\begin{equation}
   \m{K}^{(c)}_{\v{j}, \val |\v{j}', \val', \v{j}'',\val''}:= \frac{1}{\Delta^3_x}\delta_{\v{j}, \v{j}''} \mc{E}_{\v{j},\v{j}'}\cdot (\v{W}^{(1)}_{\val|\val',\val''} +\v{W}^{(2)}_{\val|\val',\val''}),
\end{equation}
plus the term with the weighting $\sigma'_{\v{j},\v{j}''}$,
\begin{equation}
   \m{K}^{(n)}_{\v{j}, \val |\v{j}', \val', \v{j}'',\val''} := \sigma'_{\v{j},\v{j}''} \mc{E}_{\v{j},\v{j}'}\cdot  \v{W}^{(2)}_{\val|\val',\val''}.
\end{equation}
Therefore, we shall use Theorem~\ref{thm:two-kernel-hbe} to obtain separate block-encodings $\m{U}_{\m{K}^{(c)}}$ for $\m{K}^{(c)}$ and $\m{U}_{\m{K}^{(n)}}$ for $\m{K}^{(n)}$, and then combine via LCU to obtain $\m{U}_{\m{K}}$. 
Recall from above that the non-diagonal elements of these matrices are
\begin{align}
    \mc{E}_{\v{j},\v{j}'}&=  \frac{1}{\pi}\int_0^1 \mathrm{d}s \left (\frac{\mu(s) \hat{\v{x}}}{\sqrt{s(1-s)}}
\frac{\ee^{-\mu(s)\Delta_x|\v{j} -\v{j}'|}}{4\pi \Delta_x|\v{j} -\v{j}'|} + \frac{\hat{\v{x}}}{\sqrt{s(1-s)}} \frac{\ee^{-\mu(s) \Delta_x | \v{j} - \v{j}'|}}{4\pi \Delta_x^2|\v{j}-\v{j}'|^2 } \right ) \nonumber \\
    \sigma_{\v{j},\v{j}''}&= \frac{1}{\pi} \int_0^1 \sqrt{\frac{s}{1-s}} \frac{\ee^{-\mu(s)\Delta_x|\v{j}-\v{j}''|}}{4\pi \Delta_x|\v{j}-\v{j}''|}\, \mathrm{d}s.
\end{align}

We first note some relevant points to explicitly block-encoding the matrix. Firstly, the constructions of the oracles for the entries are slightly more involved than a pure Coulomb construction. The form of the elements are explicitly known, and the integrals involved can be readily approximated via quadratures. For example, a quadrature for the Feynman identity can be obtained from the discrete sum in the following limit
\begin{equation}
    \frac{1}{\sqrt{AB}} =\lim_{N\rightarrow \infty} \frac{1}{N} \sum_{k=1}^N \frac{1}{A \sin^2 \theta_k + B \cos^2 \theta_k},
\end{equation}
with $\theta_k = (2k-1)\pi/(4N)$. And so the integral can be easily replaced with discrete sum approximation. Besides being very time-consuming, there is no obstacle to this being realized via coherent arithmetic, although a detailed analysis would need to be done to determine the exact resources and errors. Secondly, since for the Fourier-truncated system the total number of points is $(2\NF+1)^3$, we must embed these in a sufficiently large power of 2, and pad with zeros, to directly apply the previous theorem.

The case of $\m{K}^{(c)}$ corresponds to the convolutional, quadratic interaction case. Recall that
\begin{equation}
    \|\mc{E}^a\|_{\max} \le \frac{1}{4\pi} \left (\frac{\sqrt{1+\tau}}{ \Delta_x}
  +  \frac{1}{ \Delta_x^2 } \right ),
\end{equation}
for any $a=1,2,3$ with the max occurring on the diagonal entries. In terms of Theorem~\ref{thm:two-kernel-hbe}, factoring in the delta-function contributions gives that the largest value attained for the kernel is
\begin{equation}
  \frac{1}{4\pi}\frac{1}{\Delta^3_x} \left (\frac{\sqrt{1+\tau}}{ \Delta_x}
  +  \frac{1}{ \Delta_x^2 } \right ).
\end{equation}
The decay profile is exponential and so we obtain $O(1)$ contribution from the sum. Thus,
\begin{equation}
    \alpha_{\m{K}^{(c)}} =  \frac{1}{4\pi}\frac{1}{\Delta^3_x} \left (\frac{\sqrt{1+\tau}}{ \Delta_x}
  +  \frac{1}{ \Delta_x^2 } \right ) \times O(1) \times \sum_{a=1}^3\sqrt{s_{\mathrm{c}}(W) s_{\mathrm{r}}(W)}\|W^a\| = O(\Delta_x^{-5}).
\end{equation}
The case $\m{K}^{(n)}$ is handled similarly with Theorem~\ref{thm:two-kernel-hbe}, and since both kernels are rapidly decaying and $\|\sigma'\|_{\max} = 1/(8\pi \Delta_x)$, we see that 
\begin{equation}
    \alpha_{\m{K}^{(n)}} = O(\Delta_x^{-3}).
\end{equation}
Combining via LCU, and obtaining the block-encoding of $\tilde{\m{F}}_2$ via controlled Fourier transforms, we find that
\begin{equation}
    \alpha_{\tilde{\m{F}}_2} = \Delta_x^{9/2} O(\Delta_x^{-5} \sqrt{\NH}) = O(\sqrt{\NF \NH}),
\end{equation}
for the resulting unitary block-encoding of the nonlinear term $\tilde{\m{F}}_2$.

\newpage

\section{Additional proofs}

This appendix collects auxiliary estimates used in the main-text stability and observable-extraction arguments.

\subsection{Conservation of mass}
\label{subsec:mass_conservation}

\begin{proposition}[Preservation of zero perturbation mass]
\label{prop:zero_mass_preserved}
Let $h$ be a sufficiently regular solution of \eqref{eq:kinetich} on
$\TL^3\times\mathbb R^3\times [0,T]$, with periodic boundary conditions in $x$ and
sufficient decay in $v$. Then,
\begin{equation}
\label{eq:Mh_ODE}
\frac{\mathrm d}{\mathrm dt}M(t)
=
-\nu M(t).
\end{equation}
Consequently,
\begin{equation}
\label{eq:Mh_solution}
M(t)=\ee^{-\nu t}M(0).
\end{equation}
In particular, if $M(0)=0$, then $M(t)=0$ for all $t\in [0,T]$.
\end{proposition}

\begin{proof}
From \eqref{eq:Mh} and \eqref{eq:kinetich}, we compute
\begin{align}
\label{eq:ddtMh}
\frac{\mathrm d}{\mathrm dt}M(t)
=&\; 
\int_{\TL^3\times\mathbb R^3}
\partial_t h(\vx,\vv,t)\,\mathrm d^3x\,\mathrm d^3v
\nonumber\\
=&\; 
-\int_{\TL^3\times\mathbb R^3}
\vv\cdot\nabla_x h\,\mathrm d^3x\,\mathrm d^3v
-\frac{\lvert q_e\rvert}{m}
\int_{\TL^3\times\mathbb R^3}
\v E\cdot\nabla_v f_{0}\,\mathrm d^3x\,\mathrm d^3v
\nonumber\\
&\;
-\frac{\lvert q_e\rvert}{m}
\int_{\TL^3\times\mathbb R^3}
\v{E} \cdot\nabla_v h\,\mathrm d^3x\,\mathrm d^3v
-\nu
\int_{\TL^3\times\mathbb R^3}
h\,\mathrm d^3x\,\mathrm d^3v .
\end{align}
The first term in \eqref{eq:ddtMh} vanishes by periodicity:
\begin{equation}
\int_{\TL^3}\vv\cdot\nabla_x h\,\mathrm d^3x
=
\vv\cdot 
\int_{\TL^3}\nabla_x h\,\mathrm d^3x
=0.
\end{equation}
The second term in \eqref{eq:ddtMh} vanishes because $\v E(\vx,t)$ is independent of
$\vv$, while $f_{0}$ has sufficient decay in $\vv$:
\begin{equation}
\int_{\mathbb R^3}
\v E(\vx,t)\cdot\nabla_v f_{0}(\vv)\,\mathrm d^3v
=
\v E(\vx,t)\cdot
\int_{\mathbb R^3}\nabla_v f_0(\vv)\,\mathrm d^3v
=0.
\end{equation}
Similarly, the third term in \eqref{eq:ddtMh} vanishes,
\begin{equation}
\int_{\mathbb R^3}
\v E(\vx,t)\cdot\nabla_v h(\vx,\vv,t)\,\mathrm d^3v
=
\v E(\vx,t)\cdot
\int_{\mathbb R^3}\nabla_v h(\vx,\vv,t)\,\mathrm d^3v
=0,
\end{equation}
again by decay in $\vv$. Hence, \eqref{eq:ddtMh} becomes \eqref{eq:Mh_ODE}. \\

By existence and uniqueness for the scalar linear ODE \eqref{eq:Mh_ODE}, we obtain \eqref{eq:Mh_solution}. In particular, if $M(0)=0$, then the zero
function solves the same scalar initial-value problem, and uniqueness gives
$M(t)=0$ for all $t\geq 0$.
\end{proof}

\subsection{Quadratic Lyapunov functional}
\label{subsec:quadratic_Lyapunov}

We next identify the Lyapunov matrix \(\m P\) from the quadratic part of the truncated plasma free energy.

\begin{proposition}[Lyapunov matrix as the truncated quadratic free energy]
\label{prop:P_quadratic_free_energy}
Let $\mc F^{(2)}$ denote the quadratic part of the continuum free energy
in Eq.~\eqref{eq:freenergycontinuum}, written in the dimensionless variables
of Eq.~\eqref{eq:natural_units} and with $\phi$ constrained by
Eq.~\eqref{eq:phi_g_single_species}. Let $\Pi_{\NF,\NH}$ denote the
orthogonal projection onto the finite Fourier--Hermite subspace defined by
Eqs.~\eqref{eq:Hermite_g} and \eqref{eq:Truncation}. For a coefficient vector
$\m u$, let $g_{\NF,\NH}$ be the corresponding finite Fourier--Hermite function
and let $\phi_{\NF,\NH}$ be determined from $g_{\NF,\NH}$ by
Eq.~\eqref{eq:phi_g_single_species}. Define
\begin{equation}
\label{eq:F2_truncated_def}
\mc F^{(2)}_{\NF,\NH}(\m u)
\coloneqq
\mc F^{(2)}[g_{\NF,\NH},\phi_{\NF,\NH}].
\end{equation}
Then,
\begin{equation}
\label{eq:F2_truncated_equals_P_norm}
\mc F^{(2)}_{\NF,\NH}(\m u)
=
\frac12\lVert \m u\rVert_{2,\m P}^2.
\end{equation}
\end{proposition}

\begin{proof}
We evaluate the quadratic free energy on a finite Fourier--Hermite
function $g_{\NF,\NH}$ and on the corresponding constrained field
$\phi_{\NF,\NH}$.

By Eqs.~\eqref{eq:kinetic_ansatz}, \eqref{eq:honf}, and \eqref{eq:w_def}, in
dimensionless variables we have $f=w(1+g)$ and $f_0=w$. Expanding the entropy
term gives
\begin{align}
S(f\|f_0)
&=
\int_{\TL^3\times\R^3}
w(\vv)(1+g(\vx,\vv))
\log(1+g(\vx,\vv))\,
\mathrm d^3x\,\mathrm d^3v
\nonumber\\
&=
\int_{\TL^3\times\R^3}
wg\,\mathrm d^3x\,\mathrm d^3v
+
\frac12
\int_{\TL^3\times\R^3}
w|g|^2\,\mathrm d^3x\,\mathrm d^3v
+
O(g^3).
\end{align}
The first term is linear in the perturbation and is therefore not part of
$\mc F^{(2)}$. Since the field term in
Eq.~\eqref{eq:freenergycontinuum} is already quadratic, we have
\begin{align}
\mc F^{(2)}[g,\phi]
=
\frac12
\int_{\TL^3\times\R^3}
|g(\vx,\vv)|^2w(\vv)\,
\mathrm d^3v\,\mathrm d^3x
+\frac12
\int_{\TL^3}
\left(
|\nabla_x\phi(\vx)|^2+\tau|\phi(\vx)|^2
\right)
\,\mathrm d^3x .
\label{eq:F2_continuum_explicit}
\end{align}

We now evaluate Eq.~\eqref{eq:F2_continuum_explicit} on
$g_{\NF,\NH}$. By Parseval's theorem in $x$ and orthonormality of the Hermite
basis in $L^2(\R^3,w(\vv)\mathrm d^3v)$,
\begin{equation}
\frac12
\int_{\TL^3\times\R^3}
|g_{\NF,\NH}(\vx,\vv)|^2w(\vv)\,
\mathrm d^3v\,\mathrm d^3x
=
\frac12
\sum_{\substack{\vn\in\mc{Z}_{\NF}^3\\ \val\in\mc{N}_{\NH}^3}}
|g_{\vn,\val}|^2 .
\end{equation}

It remains to rewrite the field contribution. Taking the Fourier transform of
Eq.~\eqref{eq:phi_g_single_species} gives
\begin{equation}
\kappa_{\vn}\hat\phi_{\vn}=g_{\vn,\v 0},
\end{equation}
with $\kappa_{\vn}$ defined in Eq.~\eqref{eq:kappa_n_def}. Therefore
$\hat\phi_{\vn}=0$ outside the Fourier truncation whenever
$g=g_{\NF,\NH}$. Applying Parseval's theorem again,
\begin{align}
\frac12
\int_{\TL^3}
\left(
|\nabla_x\phi_{\NF,\NH}|^2+\tau|\phi_{\NF,\NH}|^2
\right)\,\mathrm d^3x
=
\frac12
\sum_{\vn\in\mc{Z}_{\NF}^3}
\kappa_{\vn}|\hat\phi_{\vn}|^2=
\frac12
\sum_{\vn\in\mc{Z}_{\NF}^3}
\frac{1}{\kappa_{\vn}}|g_{\vn,\v 0}|^2 .
\end{align}
Combining the kinetic and field contributions gives
\begin{equation}
\mc F^{(2)}_{\NF,\NH}(\m u)
=
\frac12
\sum_{\substack{\vn\in\mc{Z}_{\NF}^3\\ \val\in\mc{N}_{\NH}^3}}
|g_{\vn,\val}|^2
+
\frac12
\sum_{\vn\in\mc{Z}_{\NF}^3}
\frac{1}{\kappa_{\vn}}|g_{\vn,\v 0}|^2 .
\end{equation}
By the definition of $\m P$, the right-hand side is exactly
$\frac12\lVert \m u\rVert_{2,\m P}^2$, which proves
Eq.~\eqref{eq:F2_truncated_equals_P_norm}.
\end{proof}

\section{Initial data analysis}
\label{appendix:initialdata}

This appendix evaluates the Lyapunov \(R_{\m P}\)-condition for the Fourier-mode and Gaussian initial data used in the main text.

\subsection{Proof of Proposition~\ref{prop:eta_avg_RP_necessary}}
\label{subsec:eta_avg_proof}

By Proposition~\ref{prop:P_quadratic_free_energy},
\begin{align}
\frac{1}{2}\lVert \m{u}_0\rVert_{2,\m P}^2
=&\;
\frac{1}{2}\int_{\TL^3}\int_{\mathbb{R}^3}
\lvert g(\vx,\vv,0)\rvert^2 w(\vv)\,\mathrm{d}^3v\,\mathrm{d}^3x+
\frac{1}{2}
\int_{\TL^3}
\big(
\lvert \nabla_{\vx}\phi(\vx,0)\rvert^2
+
\tau\lvert \phi(\vx,0)\rvert^2
\big)\,\mathrm{d}^3x
\nonumber\\
\geq&\;
\frac{\tau}{2}
\int_{\TL^3}
\lvert \phi(\vx,0)\rvert^2\,\mathrm{d}^3x
=
\frac{1}{2\tau}
\int_{\TL^3}
\lvert \varphi(\vx,0)\rvert^2\,\mathrm{d}^3x .
\end{align}
Hence,
\begin{equation}
\label{eq:u0_lower_bound_eta}
\lVert \m{u}_0\rVert_{2,\m P}^2
\geq
\frac{1}{\tau}
\int_{\TL^3}
\lvert \varphi(\vx,0)\rvert^2\,\mathrm{d}^3x .
\end{equation}
By the Cauchy--Schwarz inequality,
\begin{equation}
\varphi_{\mathrm{avg}}^2
=
\frac{1}{L^6}
\Bigg(
\int_{\TL^3}
\lvert \varphi(\vx,0)\rvert\,\mathrm{d}^3x
\Bigg)^2
\leq
\frac{1}{L^3}
\int_{\TL^3}
\lvert \varphi(\vx,0)\rvert^2\,\mathrm{d}^3x .
\end{equation}
Combining this with \eqref{eq:u0_lower_bound_eta}, we obtain
\begin{equation}
\label{eq:u0_lower_bound_eta_avg}
\lVert \m{u}_0\rVert_{2,\m P}
\geq
\frac{L^{3/2}}{\sqrt{\tau}}\varphi_{\mathrm{avg}} .
\end{equation}
Therefore,
\begin{equation}
R_{\m P}
=
\frac{\lVert \tilde{\m F}_2\rVert_2
      \lVert \m{u}_0\rVert_{2,\m{P}}}
     {\bar{\nu}}
\geq
\frac{
\lVert \tilde{\m F}_2\rVert_2 L^{3/2}
}{
\bar{\nu}\sqrt{\tau}
}
\varphi_{\mathrm{avg}} .
\end{equation}
Since \(R_{\m P}<1\), the result follows.

\subsection{Proof of Proposition~\ref{prop:Fourier_initial_data}}
\label{subsec:Fourier_proof}

Observe that the initial datum \eqref{eq:initialconditiontheorem} verifies \eqref{eq:eta_max} as $\cos(\v{\xi}_{\vn_0}\cdot\vx)=1$ is attained at $\v{x}=\v{0}\in \TL^3$. Moreover, the initial data \eqref{eq:initialconditiontheorem} satisfies \eqref{eq:phi_g_single_species} at $t=0$. Hence the initial data \eqref{eq:initialconditiontheorem} is consistent.

We compute the Fourier--Hermite coefficients of $g(\vx,\vv,0)$ as
\begin{equation}
g_{\vn,\val}(0)
=
\frac{\varphi_{\max}\kappa_{\vn_0}L^{3/2}}{2\tau}
\big(
\delta_{\vn,\vn_0}
+
\delta_{\vn,-\vn_0}
\big)
\delta_{\val,\v 0}.
\end{equation}
Therefore,
\begin{align}
\lVert \m{u}(0)\rVert_{2,\m{P}}^2
=
\sum_{\vn\in\mc{Z}_{\NF}^3}
\bigg(1+\frac{1}{\kappa_{\vn}}\bigg)
\lvert g_{\vn,\v 0}(0)\rvert^2
=2
\bigg(1+\frac{1}{\kappa_{\vn_0}}\bigg)
\bigg(
\frac{\varphi_{\max}\kappa_{\vn_0}L^{3/2}}{2\tau}
\bigg)^2
=
\frac{\varphi_{\max}^2}{\tau^2}
\frac{L^3}{2}
\kappa_{\vn_0}(1+\kappa_{\vn_0}).
\end{align}
This proves \eqref{eq:u0P_Fourier_mode}.

Inserting \eqref{eq:u0P_Fourier_mode} into \eqref{eq:RP_upper_bound}, we obtain
\begin{align}
R_{\m{P}}
\leq &\;
\frac{\sqrt{3}}{\bar{\nu}L^{3/2}}
\bigg(
\frac{1}{1+\tau}+\frac{13L^2}{2\pi^2}\NF
\bigg)^{1/2}
\big(1+\sqrt{3\NH}\big)
\frac{\varphi_{\max}}{\tau}
\bigg(
\frac{L^3}{2}
\kappa_{\vn_0}(1+\kappa_{\vn_0})
\bigg)^{1/2}
\nonumber\\
=&\;
\frac{\varphi_{\max}}{\tau\bar{\nu}}
\bigg(
\frac{3}{2}
\kappa_{\vn_0}(1+\kappa_{\vn_0})
\bigg)^{1/2}
\bigg(
\frac{1}{1+\tau}+\frac{13L^2}{2\pi^2}\NF
\bigg)^{1/2}
\big(1+\sqrt{3\NH}\big).
\end{align}
Hence the sufficient condition \eqref{eq:varphimax_Fourier_mode} implies $R_{\m{P}}<1$.

\subsection{Proof of Proposition~\ref{prop:Gaussian}}
\label{subsec:Gaussian_proof}

We first show that the initial data \eqref{eq:Gaussian_g0}--\eqref{eq:Gaussian_phi0} are consistent. 
Writing 
\begin{equation}\label{eq:g0_Fourier}
g(\vx,\vv,0)
=
\frac{C}{L^3}
\sum_{\vn\in\mathbb{Z}^3\setminus\{0\}}
\exp\bigg(
-\frac{\sigma^2}{2}\lvert\v{\xi}_{\vn}\rvert^2
\bigg)
\exp(i\v{\xi}_{\vn}\cdot\vx)
\end{equation}
and applying the Helmholtz operator to \eqref{eq:Gaussian_phi0}, we obtain
\begin{equation}\label{eq:phi_g_verification}
(-\nabla_{\vx}^2+\tau)\phi(\vx,0)
=
\frac{C}{L^3}
\sum_{\vn\in\mathbb{Z}^3\setminus\{0\}}
\exp\bigg(
-\frac{\sigma^2}{2}\lvert\v{\xi}_{\vn}\rvert^2
\bigg)
\exp(i\v{\xi}_{\vn}\cdot\vx).
\end{equation}
Comparing \eqref{eq:phi_g_verification} with \eqref{eq:g0_Fourier} gives the desired result. 

We next show that 
\begin{equation}\label{eq:varphi_max_sigma}
\varphi_{\max}=\tau \lvert C\rvert S_1(L,\sigma).
\end{equation}
Recall that
$\varphi=\tau\phi$. By the triangle inequality applied to
\eqref{eq:Gaussian_phi0}, and then using the definition of $S_1$ \eqref{eq:S1}, we have
\begin{equation}\label{eq:varphi_max_upper_bound}
\lvert \varphi(\vx,0)\rvert
\leq
\tau \lvert C\rvert S_1(L,\sigma).
\end{equation}
Moreover, the bound is saturated at $\vx=0$, 
\begin{equation}\label{eq:varphi_max_attained}
\lvert \varphi(0,0)\rvert
=
\tau \lvert C\rvert S_1(L,\sigma).
\end{equation}
Combining \eqref{eq:varphi_max_upper_bound} and
\eqref{eq:varphi_max_attained}, we obtain \eqref{eq:varphi_max_sigma}.

We are now prepared to prove the estimates \eqref{eq:Gaussian_u0P_exact} and \eqref{eq:Gaussian_RP}. From \eqref{eq:g0_Fourier}, and using the Fourier convention
\eqref{eq:Fourier_coeffs}, the Fourier--Hermite coefficients of the initial
datum are
\begin{equation}\label{eq:Gaussian_u0_coeffs}
g_{\vn,\val}(0)
=
\frac{C}{L^{3/2}}
\exp\bigg(
-\frac{\sigma^2}{2}\lvert\v\xi_{\vn}\rvert^2
\bigg)
\delta_{\val,\v{0}} \quad (\vn\neq \v{0}),
\end{equation}
The zero Fourier mode vanishes by \eqref{eq:Gaussian_g0}. Since the datum is
independent of $\vv$, all non-zero Hermite modes vanish.

Using \eqref{eq:Lyapunov_pullback} and \eqref{eq:Gaussian_u0_coeffs}, we get
\begin{equation}
\lVert \m u_0\rVert_{2,\m P}^2
=
\sum_{\vn\in \mc{Z}_{\NF}^3\setminus\{\v 0\}}
\bigg(
1+\frac{1}{\kappa_{\vn}}
\bigg)
\lvert g_{\vn,\v 0}(0)\rvert^2
=\lvert C\rvert^2 S_2(L,\sigma).
\end{equation}
Taking square roots proves \eqref{eq:Gaussian_u0P_exact}. Substituting \eqref{eq:Gaussian_u0P_exact} into
\eqref{eq:RP_upper_bound} gives
\begin{equation}\label{eq:Gaussian_RP_bound_exact}
R_{\m P}
\leq
\frac{
\sqrt{3}\lvert C\rvert
S_2(L,\sigma)^{1/2}}{
\bar{\nu}L^{3/2}
}
\bigg(
\frac{1}{1+\tau}
+
\frac{13L^2}{2\pi^2}\NF
\bigg)^{1/2}
\big(1+\sqrt{3\NH}\big)
.
\end{equation}
Therefore $R_{\m P}<1$ whenever the right-hand side of
\eqref{eq:Gaussian_RP_bound_exact} is strictly less than $1$. Using
\eqref{eq:varphi_max_sigma} to eliminate $\lvert C\rvert$ gives
\eqref{eq:Gaussian_RP}.

It remains to justify the asymptotic statements. Suppose that
$\sigma=cL$, with $c\in(0,1)$ fixed. Then
\begin{equation}
\label{eq:S1_sigma_cL_expanded}
S_1(L,cL)
=
\frac{1}{L^3}
\sum_{\vn\in\mathbb Z^3\setminus\{\v 0\}}
\frac{
\exp(-2\pi^2c^2|\vn|^2)
}{
4\pi^2|\vn|^2/L^2+\tau
}.
\end{equation}
The denominator in \eqref{eq:S1_sigma_cL_expanded}
is bounded below by $\tau$, while the Gaussian factor is summable on
$\mathbb Z^3$. Therefore
\begin{equation}
\label{eq:S1_sigma_cL_upper}
S_1(L,cL)=O(L^{-3}).
\end{equation}
Conversely, for $L\geq 1$, the contribution of the fixed nonzero mode
$\vn=(1,0,0)$ gives
\begin{equation}
\label{eq:S1_sigma_cL_lower}
S_1(L,cL)
\geq
\frac{1}{L^3}
\frac{
\exp(-2\pi^2c^2)
}{
4\pi^2+\tau
}.
\end{equation}
Combining \eqref{eq:S1_sigma_cL_upper} and
\eqref{eq:S1_sigma_cL_lower}, we obtain
\begin{equation}
\label{eq:S1_sigma_cL_theta}
S_1(L,cL)=\Theta(L^{-3}).
\end{equation}

Similarly,
\begin{equation}
\label{eq:S2_sigma_cL_expanded}
S_2(\NF,L,cL)
=
\frac{1}{L^3}
\sum_{\vn\in \mathcal{Z}_{\NF}^3\setminus\{\v 0\}}
\left(
1+\frac{1}{4\pi^2|\vn|^2/L^2+\tau}
\right)
\exp(-4\pi^2c^2|\vn|^2).
\end{equation}
The summand in \eqref{eq:S2_sigma_cL_expanded} is bounded above by
\begin{equation}
\label{eq:S2_sigma_cL_summand_upper}
\left(1+\frac{1}{\tau}\right)
\exp(-4\pi^2c^2|\vn|^2),
\end{equation}
which is summable uniformly in $L$ and $\NF$. Hence
\begin{equation}
\label{eq:S2_sigma_cL_upper}
S_2(\NF,L,cL)=O(L^{-3}).
\end{equation}
For $\NF\geq 1$ and $L\geq 1$, the mode $\vn=(1,0,0)$ gives
\begin{equation}
\label{eq:S2_sigma_cL_lower}
S_2(\NF,L,cL)
\geq
\frac{1}{L^3}
\left(
1+\frac{1}{4\pi^2+\tau}
\right)
\exp(-4\pi^2c^2).
\end{equation}
Therefore,
\begin{equation}
\label{eq:S2_sigma_cL_theta}
S_2(\NF,L,cL)=\Theta(L^{-3}),
\end{equation}
uniformly for $\NF\geq 1$.

It follows from \eqref{eq:S1_sigma_cL_theta} and
\eqref{eq:S2_sigma_cL_theta} that
\begin{equation}
\label{eq:S1_S2_ratio_sigma_cL}
\frac{L^{3/2}S_1(L,cL)}
{S_2(\NF,L,cL)^{1/2}}
=
O(1).
\end{equation}
Substituting \eqref{eq:S1_S2_ratio_sigma_cL} into
\eqref{eq:Gaussian_RP} gives
\begin{equation}
\label{eq:Gaussian_scaling_general_sigma_cL}
\varphi_{\max}
=
O\bigg(
\left(
\frac{1}{1+\tau}
+
\frac{13L^2}{2\pi^2}\NF
\right)^{-1/2}
\left(1+\sqrt{3\NH}\right)^{-1}
\bigg).
\end{equation}
For fixed $L$, the first factor in
\eqref{eq:Gaussian_scaling_general_sigma_cL} scales as $O(\NF^{-1/2})$,
while the second factor scales as $O(\NH^{-1/2})$. Thus
\begin{equation}
\label{eq:Gaussian_fixed_L_scaling}
\varphi_{\max}
=
O\!\left((\NF\NH)^{-1/2}\right),
\qquad
\NF,\NH\to\infty.
\end{equation}
Finally, under Debye-resolution scaling $\NF=\Theta(L)$, we have
\begin{equation}
\label{eq:Gaussian_Debye_resolution_factor}
\left(
\frac{1}{1+\tau}
+
\frac{13L^2}{2\pi^2}\NF
\right)^{-1/2}
=
O(L^{-3/2}).
\end{equation}
Substitution into \eqref{eq:Gaussian_scaling_general_sigma_cL} gives
\begin{equation}
\label{eq:Gaussian_Debye_resolution_scaling_strong}
\varphi_{\max}
=
O\!\left(L^{-3/2}\NH^{-1/2}\right).
\end{equation}
In particular, since $L^{-3/2}\NH^{-1/2}
=O\big((L\NH)^{-1/2}\big)$ for $L\geq 1$, this implies the stated bound
\begin{equation}
\label{eq:Gaussian_Debye_resolution_scaling}
\varphi_{\max}
=
O\!\left((L\NH)^{-1/2}\right).
\end{equation}

\subsection{Norm of the initial data}
\label{app:regularity}

\begin{proposition}[Weighted square integrability implies $\|\m{u}_0\|_2=O(1)$]
\label{prop:u0_O1_characterization}
For each pair $(\NF,\NH)\in\mathbb{N}^2$, let
\begin{equation}
\m{u}_0
\coloneqq
\big(g_{\vn,\v{\alpha}}(0)\big)_{\substack{
\vn\in \mc{Z}_{\NF}^3 \\
\v{\alpha}\in \mc{N}_{\NH}^3 }}.
\end{equation}
Assume that
\begin{equation}
\label{eq:square_integrability}
g(\cdot,\cdot,0)\in
L^2(\TL^3\times\mathbb{R}^3,\,
w(\vv)\,\mathrm{d}^3v\,\mathrm{d}^3x).
\end{equation}
Then,
\begin{equation}
\label{eq:uniform_boundedness}
\sup_{(\NF,\NH)\in\mathbb{N}^2}
\lVert\m{u}_0\rVert_2<\infty.
\end{equation}
and, moreover,
\begin{equation}
\label{eq:sup_identity_result}
\sup_{(\NF,\NH)\in\mathbb{N}^2}
\lVert\m{u}_0\rVert_2^2
=
\sum_{\substack{\vn\in\mathbb{Z}^3\\ \v{\alpha}\in\mathbb{N}_0^3}}
\lvert g_{\vn,\v{\alpha}}(0)\rvert^2
=
\int_{\TL^3}\int_{\mathbb{R}^3}
\lvert g(\vx,\vv,0)\rvert^2 w(\vv)\,\mathrm{d}^3v\,\mathrm{d}^3x.
\end{equation}
\end{proposition}

\begin{proof}
By the unitarity of the Fourier transform on $\TL^3$ and the orthonormality of
the Hermite basis in $L^2(\mathbb{R}^3,w(\vv)\,\mathrm{d}^3v)$, \eqref{eq:square_integrability} and Parseval's
theorem give
\begin{equation}
\label{eq:parseval_u0}
\sum_{\substack{\vn\in\mathbb{Z}^3\\ \v{\alpha}\in\mathbb{N}_0^3}}
\lvert g_{\vn,\v{\alpha}}(0)\rvert^2
=
\int_{\TL^3}\int_{\mathbb{R}^3}
\lvert g(\vx,\vv,0)\rvert^2 w(\vv)\,\mathrm{d}^3v\,\mathrm{d}^3x<\infty.
\end{equation}
For each $(\NF,\NH)\in\mathbb N^2$,
\begin{equation}
\label{eq:partial_sum_bound}
\lVert\m{u}_0\rVert_2^2
=S(\NF,\NH) \leq
\sum_{\substack{\vn\in\mathbb{Z}^3\\\v{\alpha}\in\mathbb{N}_0^3}}
\lvert g_{\vn,\v{\alpha}}(0)\rvert^2
<\infty
\end{equation}
where
\begin{equation}
\label{eq:partial_sums}
S(\NF,\NH)\coloneqq \sum_{\substack{\vn\in \mc{Z}_{\NF}^3\\ \v{\alpha}\in \mc{N}_{\NH}^3}}
\lvert g_{\vn,\v{\alpha}}(0)\rvert^2.
\end{equation}
Hence \eqref{eq:uniform_boundedness} holds. Since the partial sums \eqref{eq:partial_sums} are
monotone increasing in $\NF$ and $\NH$, we also obtain
\begin{equation}
\label{eq:sup_identity}
\sup_{(\NF,\NH)\in\mathbb{N}^2}
\lVert\m{u}_0\rVert_2^2
=
\sum_{\substack{\vn\in\mathbb{Z}^3\\ \v{\alpha}\in\mathbb{N}_0^3}}
\lvert g_{\vn,\v{\alpha}}(0)\rvert^2.
\end{equation}
Combining \eqref{eq:parseval_u0} and \eqref{eq:sup_identity} gives
\eqref{eq:sup_identity_result}.
\end{proof}

\subsection{Sobolev regularity}

Let $g(\vx,\vv,t)$ be a solution of \eqref{eq:gnonlinearintro}, and let
$g_{\vn,\val}(t)$ denote its Fourier-Hermite coefficients. For any subset
$\Omega\subseteq \mathbb{Z}^3\times\mathbb{N}_0^3$, define the
truncation operator $\m T_\Omega$ through its Fourier-Hermite coefficients,
by
\begin{equation}
(\m{T}_{\Omega}g)_{\vn,\val}\coloneqq
\begin{cases}
g_{\vn,\val} & (\vn,\val)\in\Omega,\\
0 & (\vn,\val)\notin\Omega.
\end{cases}
\end{equation}
For the chosen Fourier and Hermite cutoffs $\NF,\NH$, 
the corresponding Fourier-Hermite truncation is
\begin{equation}
\m{T}_{\NF,\NH}\coloneqq \m{T}_{\mc{Z}_{\NF}^3\times\mc{N}_{\NH}^3}.
\end{equation}

We measure truncation error in the weighted $L^2$ space
\begin{equation}
\mc{X}
\coloneqq
L^2\big(
\mathbb{T}_L^3\times\mathbb{R}^3,
w(\vv)\,\mathrm{d}^3 v\,\mathrm{d}^3 x
\big),
\end{equation}
with norm
\begin{equation}
\lVert g(t)\rVert_{\mc X}^2
\coloneqq
\int_{\mathbb T_L^3}
\int_{\mathbb R^3}
\lvert g(\vx,\vv,t)\rvert^2
w(\vv)\,\mathrm d^3v\,\mathrm d^3x =\sum_{\substack{\vn\in\mathbb Z^3\\ \val\in\mathbb N_0^3}}
\lvert g_{\vn,\val}(t)\rvert^2,
\end{equation}
using Parseval's theorem and the orthonormality of the Hermite basis. The associated Fourier-Hermite truncation error is
\begin{equation}
E_{\NF,\NH}(t)
\coloneqq
\lVert(\m{I}-\m{T}_{\NF,\NH})g(t)\rVert_{\mc X}.
\end{equation}

We also define the Fourier-only and Hermite-only truncations
\begin{equation}
\m{T}^{\mathrm{F}}_{\NF}
\coloneqq
\m{T}_{\mc{Z}_{\NF}^3\times\mathbb{N}_0^3},
\qquad
\m{T}^{\mathrm{H}}_{\NH}
\coloneqq
\m{T}_{\mathbb Z^3\times\mc{N}_{\NH}^3}.
\end{equation}
Then, $\m{T}_{\NF,\NH}
=
\m{T}^{\mathrm{F}}_{\NF}\m{T}^{\mathrm{H}}_{\NH}
=
\m{T}^{\mathrm{H}}_{\NH}\m T^{\mathrm{F}}_{\NF}$. Since $\m{T}^{\mathrm F}_{\NF}$ and $\m T^{\mathrm{H}}_{\NH}$ are orthogonal projections on $\mc{X}$, we have
\begin{equation}
\m{I}-\m{T}_{\NF,\NH}
=
(\m{I}-\m{T}^{\mathrm F}_{\NF})
+
\m{T}^{\mathrm F}_{\NF}(\m{I}-\m{T}^{\mathrm H}_{\NH}),
\end{equation}
and therefore
\begin{equation}
E_{\NF,\NH}(t)
\leq
\lVert(\m I-\m T^{\mathrm{F}}_{\NF})g(t)\rVert_{\mc{X}}
+
\lVert(\m I-\m T^{\mathrm{H}}_{\NH})g(t)\rVert_{\mc X}.
\end{equation}

We now state the spectral regularity assumptions used to control the
truncation error. For $\sof>0$, define the Fourier-Sobolev coefficient
norm
\begin{equation}
\lVert g(t)\rVert_{\mathrm{F},\sof}^2
\coloneqq
\sum_{\substack{\vn\in\mathbb Z^3\\ \val\in\mathbb N_0^3}}
\big(1+\lvert\v\xi_{\vn}\rvert^2\big)^{\sof}
\lvert g_{\vn,\val}(t)\rvert^2 .
\end{equation}
For $\soh>0$, define the Hermite-Sobolev coefficient norm
\begin{equation}
\lVert g(t)\rVert_{\mathrm{H},\soh}^2
\coloneqq
\sum_{\substack{\vn\in\mathbb Z^3\\ \val\in\mathbb N_0^3}}
\big(1+\lVert\val\rVert_1\big)^{2\soh}
\lvert g_{\vn,\val}(t)\rvert^2.
\end{equation}

\begin{ass}[Uniform Fourier-Hermite regularity]
\label{ass:FH_regular}
Fix a simulation interval $[0,T]$. We assume that there exist exponents
$\sof,\soh>0$ such that
\begin{align}
\sup_{0\le t\le T}\lVert g(t)\rVert_{\mathrm{F},\sof}<\infty,\quad 
\sup_{0\le t\le T}\lVert g(t)\rVert_{\mathrm{H},\soh}<\infty.
\end{align}
\end{ass}

\begin{prop}[Fourier-Hermite truncation error]
\label{prop:FH_truncation_error}
Under Assumption~\ref{ass:FH_regular}, the truncation error satisfies
\begin{align}\label{eq:E_truncation_bound}
\sup_{t\in[0,T]}E_{\NF,\NH}(t) 
\leq &\; 
\bigg(1+\bigg(\frac{2\pi}{L}\bigg)^2(\NF+1)^2\bigg)^{-\sof/2}\sup_{t\in [0,T]}\lVert g(t)\rVert_{\mathrm{F},\sof}
+
(\NH+2)^{-\soh}\sup_{t\in[0,T]}\lVert g(t)\rVert_{\mathrm{H},\soh}.
\end{align}
\end{prop}

\begin{proof}
By the tail decomposition established above,
\begin{equation}
\label{eq:E_triangle_bound}
E_{\NF,\NH}(t)
\leq
\lVert(\m I-\m T^{\mathrm F}_{\NF})g(t)\rVert_{\mc{X}}
+
\lVert(\m I-\m T^{\mathrm H}_{\NH})g(t)\rVert_{\mc{X}}.
\end{equation}

For the Fourier tail, if $\vn\notin \mc{Z}_{\NF}^3$, then
$\|\vn\|_\infty\ge \NF+1$, hence
\begin{equation}
\lvert\v\xi_{\vn}\rvert^2
=
\bigg(\frac{2\pi}{L}\bigg)^2\lvert\vn\rvert^2
\geq
\bigg(\frac{2\pi}{L}\bigg)^2(\NF+1)^2.
\end{equation}
Therefore,
\begin{equation}
\label{eq:Fourier_norm}
\lVert(\m I-\m T^{\mathrm F}_{\NF})g(t)\rVert_{\mc X}^2
=
\sum_{\substack{\vn\notin\mc{Z}_{\NF}^3\\ \val\in\mathbb N_0^3}}
\lvert g_{\vn,\val}(t)\rvert^2 \leq
\bigg(1+\bigg(\frac{2\pi}{L}\bigg)^2(\NF+1)^2\bigg)^{-\sof}
\|g(t)\|_{\mathrm{F},\sof}^2.
\end{equation}

Similarly, if \(\val\notin \mc{N}_{\NH}^3\), then
$\lVert\val\rVert_\infty\ge \NH+1$, and hence $1+\lVert\val\rVert_1\ge \NH+2$.
Thus,
\begin{equation}
\label{eq:Hermite_norm}
\lVert(\m I-\m T^{\mathrm H}_{\NH})g(t)\rVert_{\mc X}^2
=
\sum_{\substack{\vn\in\mathbb Z^3\\ \val\notin\mc{N}_{\NH}^3}}
\lvert g_{\vn,\val}(t)\rvert^2\leq
(\NH+2)^{-2\soh}
\lVert g(t)\rVert_{\mathrm{H},\soh}^2.
\end{equation}

Putting \eqref{eq:Fourier_norm} and \eqref{eq:Hermite_norm} into \eqref{eq:E_triangle_bound} and taking the supremum over $t\in[0,T]$ gives the result \eqref{eq:E_truncation_bound}.
 \end{proof}

\section{Quantum ODE solver}
\label{sec:quantumODEsolver}

This appendix gives the linear-system condition-number estimate used in the history-state generation step.

Let $\mathcal{H}_c$ denote the first clock register space, with basis $\{\ket{m}\}_{m=0}^{M-1}$. We define the subspace
\begin{equation}
\tilde{V}:= \mathcal{H}_c \otimes V \subseteq \mathcal{H}_c \otimes \mathcal{H}.
\end{equation}
 Since $\bar{\mathcal{A}}$ sends $V$ into $V$ and equals $\bar{\m{A}}$ on this space, it follows that $\m{C}\otimes g(\bar{\mathcal{A}})$ sends $\tilde{V}$ into itself for any operator $\m{C}$ sending $ \mathcal{H}_c $ to $ \mathcal{H}_c$ and any function $g$. Any linear combination of such operators also has the same property, and therefore we deduce that $\m{L}$ maps $\tilde{V}$ into $\tilde{V}$, and so does $\m{L}^{-1}$. If the data $\bar{\m{z}}^0$ is restricted to lie in $V$, then the linear system in Eq.~\eqref{eq:linearsystemembedded} gives the same solution (modulo padding with zeroes) as the linear system for $\m{L}_{\tilde{V}}$ ($\m{L}$ restricted to the $\tilde{V}$ subspace). But we have 
\begin{equation}
 \m{L}_{\tilde{V}} = \m{I} -\sum_{m=0}^{M-1} \ketbra{m+1}{m} \otimes \m{T}_k(\bar{\m{A}}\Delta t).
\end{equation}
So in the complexity analysis for the embedded linear system in Eq.~\eqref{eq:linearsystemembedded} we can use the condition number of the non-embedded  linear system, namely the condition number of $\m{L}$ restricted to $\tilde{V}$.

\begin{lemma}[Bound on $\|\m{L}^{-1}_{\tilde{V}}\|$]
\label{lemma:timeandconditionnumber}
We are given a linear ODE system $\dot{\bar{\m{z}}}(t) = \bar{\mathcal{A}} \bar{\m{z}}(t)$, with $\bar{\m{z}}(t) \in \mathcal{H}$, initial conditions $\bar{\m{z}}(t)=\bar{\m{z}}^0$, and $t\in [0,T]$. We assume a subspace $V \subseteq \mathcal{H}$ that is left invariant under $\bar{\mathcal{A}}$, that $\bar{\mathcal{A}}=\bar{\m{A}}$ when restricted to $V$, and $\mu(\bar{\m{A}}) <0$. We further assume that $\bar{\m{z}}^0\in V$. We consider the truncated Taylor series discretization in Eq.~\eqref{eq:recursive}, where $M$ is the total number of timesteps, $\Delta t$ the length of a single timestep and so $M= T/\Delta t$ (assumed to be integer for simplicity). The time discretization scheme defines a sequence of recursive relations that give a linear system of equations $\m{L} \bar{\m{y}} = \bar{\m{b}}$ as in Eq.~\eqref{eq:L}--\eqref{eq:Lextra}. Now fix
\begin{itemize}
    \item A timestep $\Delta t \leq 1/\alpha_{\bar{\m{A}}}$.
    \item A Taylor series truncation $k$ 
    \begin{align}
    (k+1)! \geq \frac{M \ee^3 \bar{z}_{\max}}{\epsilon_{\mathrm{TD}}},
\end{align}
where $\bar{z}_{\textrm{max}} \geq \sup_{t \in [0,T]} \| \bar{\m{z}}(t)\|$.
\end{itemize}
Then
\begin{align}
  \| \m{L}^{-1}_{\tilde{V}}\| \leq 1+ \bigg(1+ \frac{\epsilon_{\mathrm{TD}}}{\bar{z}_{\max}}\bigg) \frac{\ee^{\mu(\bar{\m{A}}) h} (1- \ee^{\mu(\bar{\m{A}}) h M})}{1- \ee^{\mu(\bar{\m{A}})  h}}
\end{align}

\end{lemma}
\begin{proof} 
    The inverse of $\m{L}_{\tilde{V}}$ can be expanded as
\begin{align}
   \m{L}_{\tilde{V}}^{-1} = \m{I} + \sum_{j=1}^{M} \sum_{m=0}^{M-j} \ketbra{m+j}{m} \otimes (\m{T}_k(\bar{\m{A}} h))^j
\end{align}

Furthermore, from Lemma~5 in Ref.~\cite{jennings2023cost},
\begin{align}
    \| \m{T}_k(\bar{\m{A}}\Delta t)^m\| \leq \| \ee^{\bar{\m{A}}\Delta t m} \| \bigg(1+ \frac{\epsilon_{\mathrm{TD}}}{\bar{z}_{\max}}\bigg) 
\end{align}
Let us now use the negativity of the log-norm of $\bar{\m{A}}$: 
\begin{align}
   \| \ee^{\bar{\m{A}}\Delta tm}\| \leq \ee^{\mu(\bar{\m{A}}) m h}
\end{align}
and so
\begin{align}
    \| \m{T}_k(\bar{\m{A}}\Delta t)^m\| \leq  \ee^{\mu(\bar{\m{A}}) m h} \bigg(1+ \frac{\epsilon_{\mathrm{TD}}}{\bar{z}_{\textrm{max}}}\bigg) 
\end{align}
Replacing in the upper bound for $\|\m{L}^{-1}_{\tilde{V}}\|$ we get, using subadditivity, submultiplicativity, and the fact that (partial) shifts have norm bounded by $1$:
\begin{align}
   \|\m{L}_{\tilde{V}}^{-1}\| &\; \leq 1 + \sum_{j=1}^{M}  \left\|\sum_{m=0}^{M-j} \ketbra{m+j}{m}\right\|  \left\|  \m{T}_k(\bar{\m{A}} h))^j\right\| \leq 1 + \sum_{j=1}^{M}  \left\|  \m{T}_k(\bar{\m{A}} h))^j\right\| \nonumber  \\ &\; \leq 1+ \sum_{j=1}^M \ee^{\mu(\bar{\m{A}}) j h} \bigg(1+ \frac{\epsilon_{\mathrm{TD}}}{\bar{z}_{\max}}\bigg) = 1+ \bigg(1+ \frac{\epsilon_{\mathrm{TD}}}{\bar{z}_{\textrm{max}}}\bigg)\frac{\ee^{\mu(\bar{\m{A}}) h} (1- \ee^{\mu(\bar{\m{A}}) h M})}{1- \ee^{\mu(\bar{\m{A}})  h}}.
\end{align}

\end{proof}

\section{Quadrature error for time-averaged kinetic energy}
\label{app:quadrature-error}
Here we present technical details for the time-averaged quadrature approximation.

Recall that the Carleman ODE is given by
\begin{align}
   \frac{\mathrm{d}}{\mathrm{d}t}\bar{\m{z}} = \bar{\m{A}} \bar{\m{z}}, \quad \bar{\m{z}}(0)= ( \bar{\m{u}}_0, \bar{\m{u}}^{\otimes 2}_0, \dots, \bar{\m{u}}^{\otimes N_\mathrm{C}}_0)^{\mathsf{T}}.
\end{align}
This implies the solution over a time-increment $s$ is given by
\begin{equation}
    \bar{\m{z}} (t_m+s) = \ee^{\bar{\m{A}} s}\bar{\m{z}} (t_m).
\end{equation}
For any quantity $\mc{K}(t) := \mathrm{Re } \,\bar{\ell}^\dagger_{\mc{K}} \bar{\m{z}} (t)$, obtained as a linear functional of the vector $ \bar{\m{z}}$, we can consider the time-average of an interval, and write it as
\begin{equation}
    \int_t^{t+h} \mc{K}(s)\,\mathrm{d}s = \mathrm{Re } \,\bar{\ell}^\dagger_{\mc{K}}\Bigg(\int_0^{h}  \ee^{\bar{\m{A}} s}\,\mathrm{d}s\Bigg)\bar{\m{z}} (t).
\end{equation}
We ignore the Carleman truncation error here, since it is exponentially suppressed in the regime $R_{\m{P}}<1$. Assuming that $\|\bar{\m{A}} h\| \le \|\alpha_{\bar{\m{A}}}\Delta t\|<1$, we have
\begin{equation}
    \int_0^{h}  \ee^{\bar{\m{A}} s} \,\mathrm{d}s= h \sum_{q=0}^{\infty} \frac{ (\bar{\m{A}}\Delta t)^q}{(q+1)!},
\end{equation}
and so
\begin{align}
    \langle \mc{K} \rangle &:= \frac{1}{T} \int_0^T  \mc{K}(t) \,\mathrm{d}t\nonumber \\
    &= \frac{1}{T}\sum_{m=0}^{M-1} \int_{m\Delta t}^{(m+1)\Delta t}  \mc{K}(t) \,\mathrm{d}t\nonumber \\
    &= \frac{1}{T}\sum_{m=0}^{M-1} \mathrm{Re } \,\bar{\ell}^\dagger_{\mc{K}} \Delta t \sum_{q=0}^{\infty} \frac{ (\bar{\m{A}}\Delta t)^q}{(q+1)!} \bar{\m{z}}(m\Delta t).
\end{align}
We now construct a truncated Taylor series approximation, and define
\begin{equation}
    \m{R}_k(\bar{\m{A}}\Delta t) := \sum_{q=0}^{k} \frac{ (\bar{\m{A}}\Delta t)^q}{(q+1)!}.
\end{equation}
We define the approximation
\begin{align}
    \overline{\mc{K}}_k &:=\frac{\Delta t}{T}\sum_{m=0}^{M-1} \mathrm{Re } \,\bar{\ell}^\dagger_{\mc{K}} \m{R}_k(\bar{\m{A}}\Delta t) \bar{\m{z}}(m\Delta t) \nonumber \\
    &=\mathrm{Re } \,\bar{\ell}^\dagger_{\mc{K},k}\frac{1}{M} \sum_{m=0}^{M-1}   \bar{\m{z}}(m\Delta t),
\end{align}
where
\begin{equation}
    \bar{\ell}_{\mc{K},k}:= \m{R}_k(\bar{\m{A}}\Delta t)^\dagger \bar{\ell}_{\mc{K}}.
\end{equation}
Now the truncated series has an error that decreases factorially, and in particular, 
\begin{equation}
    \left \|\int_0^h   \ee^{\bar{\m{A}}s}\mathrm{d}s\, -h \m{R}_k(\bar{\m{A}}\Delta t) \right \| =\left \|h \sum_{q=k+1}^\infty \frac{(\bar{\m{A}}\Delta t)^q}{(q+1)!} \right \| \le h\left \| \sum_{q=k+1}^\infty \frac{1}{(q+1)!} \right \| \le  h \frac{\ee}{(k+2)!}.
\end{equation}
This implies that the error is given by
\begin{align}
    | \langle \mc{K} \rangle - \overline{\mc{K}}_k| & = \left | \,\bar{\ell}^\dagger_{\mc{K}}  \sum_{q=k+1}^{\infty} \frac{ (\bar{\m{A}}\Delta t)^q}{(q+1)!} \frac{1}{M}\sum_{m=0}^{M-1} \bar{\m{z}}(m\Delta t) \right| \nonumber \\
    & \le \|\bar{\ell}_{\mc{K}} \| \left \|\sum_{q=k+1}^{\infty} \frac{ (\bar{\m{A}}\Delta t)^q}{(q+1)!} \right \|  \left \| \frac{1}{M}\sum_{m=0}^{M-1} \bar{\m{z}}(m\Delta t) \right \| \nonumber \\
        & \le \frac{\ee}{(k+2)!} \|\bar{\ell}_{\mc{K}} \| \max_{t\in [0,T]}\left \| \bar{\m{z}}(t) \right \|\le  \frac{\ee}{(k+2)!} \|\bar{\ell}_{\mc{K}} \| \max_{t\in [0,T]}\left \| \bar{\m{z}}(t) \right \| .
\end{align}
Therefore, if $\|\bar{\ell}_{\mc{K}} \|$ and $\|\bar{\m{z}}(t) \| $ are bounded, then this implies the Taylor truncation scale must obey
\begin{align}
    (k+1)! \ge  \frac{1}{\epsilon} \ee \|\bar{\ell}_{\mc{K}} \| \max_{t\in [0,T]}\left \| \bar{\m{z}}(t) \right \|.
\end{align}
The dynamics is dissipative in the Lyapunov-frame, and so it suffices to take
\begin{equation}
    k = O \left ( \frac{\log (1/\epsilon )}{\log\log (1/\epsilon)}\right ),
\end{equation}
provided the step-size obeys $\Delta t \alpha_{\bar{\m{A}}}  <1$.

\subsection{Averaged kinetic energy moment}
\label{app:average-kinetic-energy}
Let us now see how to extract the spacetime averaged kinetic energy moment of the perturbation. The spatial averaging selects only the zero Fourier mode:
\begin{align}
\frac{1}{L^3}\int_{\mathbb{T}_L^3} g(\v{x},\v{v},t)\,\mathrm{d}^3 x
=&\;
\frac{1}{L^3}\int_{\mathbb{T}_L^3}
\frac{1}{L^{3/2}}
\sum_{\v{n}\in\mathbb{Z}^3}
\hat g_{\v{n}}(\v{v},t)\ee^{i\v{n}\cdot\v{x}}
\,\mathrm{d}^3 x
\nonumber\\
=&\;
\frac{1}{L^{3/2}}\hat g_{\v{0}}(\v{v},t).
\end{align}
Therefore
\begin{align}
\langle \mc{K}\rangle
=
\frac{1}{T L^{3/2}}
\int_0^T \mathrm{d}t
\int_{\mathbb{R}^3} 
\frac{|\v{v}|^2}{2}
w(\v{v})\hat g_{\v{0}}(\v{v},t)\,\mathrm{d}^3 v.
\end{align}

We now expand the zero Fourier mode in the normalized Hermite basis:
\begin{align}
\hat g_{\v{0}}(\v{v},t)
=
\sum_{\v{\alpha}\in\mathbb{N}_0^3}
g_{\v{0},\v{\alpha}}(t)
\tilde H_{\v{\alpha}}(\v{v}).
\end{align}
Since $v_a^2=\sqrt{2}\tilde H_2(v_a)+1$, we have
\begin{align}
\frac{|\v{v}|^2}{2}
=
\frac{3}{2}\tilde H_{\v{0}}(\v{v})
+
\frac{1}{\sqrt{2}}
\sum_{a=1}^3
\tilde H_{2\v{e}_a}(\v{v}).
\end{align}
Substituting this identity and using $g_{\v{0},\v{0}}(t)=0$, and orthogonality, we get
\begin{align}
\langle \mc{K}\rangle
=
\frac{1}{\sqrt{2}L^{3/2}T}
\int_0^T 
\sum_{a=1}^3
g_{\v{0},2\v{e}_a}(t)\,\mathrm{d}t.
\end{align}
This specifies the components of the vector $\ell_{\mc{K}}$.

\subsection{State preparation of $\ket{\bar{\ell}_{\mc{K},k}}$ and state overlap estimate}
\label{app:state-prep-ell}
We estimate the time averaged kinetic energy by overlapping the history state with a normalized state vector $\ket{\overline{\mc{K}}_k}$ that involves a factor $ \ket{\bar{\ell}_{\mc{K},k}} = \bar{\ell}_{\mc{K},k}/\|\bar{\ell}_{\mc{K},k}\|$.  We must therefore construct a state preparation routine for the normalized state $\ket{\bar{\ell}_{\mc{K},k}}$. Recall, that the vector $\bar{\ell}_{\mc{K},k}$ is obtained by acting $\m{R}_k(\bar{\m{A}}\Delta t)^\dagger $ on a constant, normalized, sparse vector. This can be done by the successful application of $\m{R}_k(\bar{\m{A}}\Delta t)^\dagger$ via a block-encoding, provided the scale factor of the block-encoding is not too large. 

We have the following result that returns a block-encoding $\m{U}_{\m{R}_k}$ for $\m{R}_k(\bar{\m A}\Delta t)$, and thus $\m{U}^\dagger_{\m{R}_k}$ provides the block-encoding of its adjoint.
\begin{lemma}[Block-encoding of $\m{R}_k(\bar{\m A}\Delta t)$]
\label{lem:block-encode-integrated-taylor}
Let \(\m{U}_{\bar{\m{A}}}\) be an
\((\omega_{\bar{\m{A}}},a_{\bar{\m{A}}},\epsilon_{\bar{\m{A}}})\)-block-encoding of
\(\bar{\m A}\). Fix \(h>0\), \(k\in\mathbb N_0\), and define
\begin{equation}
    \m{R}_k(\bar{\m A}\Delta t)
    :=
    \sum_{q=0}^k
    \frac{(\bar{\m A}\Delta t)^q}{(q+1)!}.
\end{equation}
Assume $\Delta t \alpha_{\bar{\m{A}}} <1$.
Then, \(\m{R}_k(\bar{\m A}\Delta t)\) admits an $(\alpha_{\m{R}},\,ka_{\bar{\m{A}}}+\lceil\log(k+1)\rceil,\,\epsilon_{\m{R}})$-block-encoding with scale factor $\alpha_{\m{R}} \le \ee-1 = O(1)$ and $\epsilon_{\m{R}}
    \le
    k\epsilon_{\bar{\m{A}}}(1+\epsilon_{\bar{\m{A}}})^{k-1}$, where \(\epsilon_{\rm prep}\) is the error in preparing the LCU coefficient state. The construction uses at most $k$ controlled calls to $\m{U}_{\bar{\m{A}}}$.
\end{lemma}

\begin{proof}
We realise the block-encoding via LCU, weighted by the Taylor coefficients. For each \(q=0,\ldots,k\), let \(\m{V}_q\) denote the product of \(q\) copies of
\(\m{U}_{\bar{\m{A}}}\), with \(\m{V}_0=\m{I}\). In the exact case,
\begin{equation}
    \m{V}_q
    \quad\text{block-encodes}\quad
    \left(\frac{\bar{\m A}}{\alpha_{\bar{\m{A}}}}\right)^q,
\end{equation}
with unit scale factor. Prepare the LCU state
\begin{equation}
    \operatorname{PREP}_{\m{R}}|0\rangle
    =
    \sum_{q=0}^k
    \sqrt{\frac{\beta_q}{\alpha_{\m{R}}}}\,|q\rangle ,
\end{equation}
with
\begin{equation}
    \beta_q := \frac{(h \alpha_{\bar{\m{A}}})^q}{(q+1)!},
\end{equation}
and $\alpha_R = \sum_{q=0}^k \beta_q$. Now define
\begin{equation}
    \operatorname{SELECT}_{\m{R}}
    :=
    \sum_{q=0}^k |q\rangle\langle q|\otimes \m{V}_q .
\end{equation}
Equivalently, \(\operatorname{SELECT}_{\m{R}}\) may be implemented by applying the
\(r\)th copy of \(\m{U}_{\bar{\m{A}}}\) controlled on the condition \(q\geq r\), for
\(r=1,\ldots,k\). The standard LCU construction
\begin{equation}
    \m{U}_{\m{R}}
    :=
    (\operatorname{PREP}_{\m{R}}^\dagger\otimes \m{I})
    \operatorname{SELECT}_{\m{R}}
    (\operatorname{PREP}_{\m{R}}\otimes \m{I})
\end{equation}
has top-left block
\begin{equation}
    \frac1{\alpha_\m{R}}
    \sum_{q=0}^k
    \beta_q
    \left(\frac{\bar{\m A}}{\alpha_{\bar{\m{A}}}}\right)^q
    =
    \frac1{\alpha_\m{R}}
    \sum_{q=0}^k
    \frac{(\bar{\m A}\Delta t)^q}{(q+1)!}
    =
    \frac{\m{R}_k(\bar{\m A}\Delta t)}{\alpha_\m{R}}.
\end{equation}
This proves the exact block-encoding statement $\epsilon_{\bar{\m{A}}}=0$. For the approximate case,$\epsilon_{\bar{\m{A}}}>0$, let
\begin{equation}
    \m{B}:=\frac{\bar{\m A}}{\omega_{\bar{\m{A}}}},
    \qquad
    \widetilde{\m{B}}:=\m{B}+\m{E},
    \qquad
    \|\m{E}\|\leq\epsilon_{\bar{\m{A}}}.
\end{equation}
The \(q\)-fold product block-encodes \(\widetilde{\m{B}}^q\) rather than \(\m{B}^q\).
Using the telescoping identity
\begin{equation}
    \widetilde{\m{B}}^q-\m{B}^q
    =
    \sum_{r=0}^{q-1}
    \widetilde{\m{B}}^{\,q-1-r}(\widetilde{\m{B}}-\m{B})\m{B}^r
\end{equation}
and the bounds \(\|\m{B}\|\leq 1\), \(\|\widetilde{\m{B}}\|\leq 1+\epsilon_{\bar{\m{A}}}\),
we obtain
\begin{equation}
    \|\widetilde{\m{B}}^q-\m{B}^q\|
    \leq
    q\epsilon_{\bar{\m{A}}}(1+\epsilon_{\bar{\m{A}}})^{q-1}
    \leq
    k\epsilon_{\bar{\m{A}}}(1+\epsilon_{\bar{\m{A}}})^{k-1}.
\end{equation}
Averaging these errors with the positive LCU weights \(\beta_q/\alpha_{\m{R}}\)
gives the stated block-encoding error, up to the coefficient-state preparation
error \(\epsilon_{\rm prep}\). Finally, if \(h\omega_{\bar{\m{A}}}\leq 1\), then
\begin{equation}
    \alpha_R
    =
    \sum_{q=0}^k
    \frac{(h\omega_{\bar{\m{A}}})^q}{(q+1)!}
    \leq
    \sum_{q=0}^\infty
    \frac1{(q+1)!}
    =
    \ee-1.
\end{equation}
\end{proof}
This provides an approximate block-encoding of $\m{R}_k$, and we construct $\m{U}^\dagger_{\m{R}_k} \otimes \m{USP}(M)$, where the second factor is a uniform state preparation over the time registers at $s=0, 1,\dots, M-1$ for the left-endpoints of the intervals. Since $\alpha_R <\ee-1$ and we act on the state $\ket{\bar{\ell}_{\mc{K}}}\otimes \ket{0}$ the success probability of the block-encoding depends on the norm of $\m{R}_k$ on a particular state.

We now show that $\m{U}_{\m{R}_k}$ has success probability bounded away from zero for its realization on any quantum state. Note that
\begin{equation}
    \m{R}_k(\bar{\m A}\Delta t)=\m{I}+\sum_{q=1}^k \frac{(\bar{\m A}\Delta t)^q}{(q+1)!}.
\end{equation}
Hence
\begin{equation}
   \|\m{R}_k(\bar{\m A}\Delta t)-\m{I}\|_2 \leq \sum_{q=1}^k\frac{\|\bar{\m{A}}\Delta t\|^q}{(q+1)!} \leq \sum_{q=1}^{\infty} \frac1{(q+1)!}=\ee-2. 
\end{equation}
Thus, for any vector \(\ket{\psi}\),
\begin{equation}
    \|\m{R}_k(\bar{\m A}\Delta t)\ket{\psi}\|\geq \|\ket{\psi}\|-\|(\m{R}_k(\bar{\m A}\Delta t)-\m{I})\ket{\psi}\|\geq(3-\ee)\|\ket{\psi}\| = (3-\ee).
\end{equation}
This implies that with bounded probability we prepare $\ket{\overline{\mc{K}}_k}$. Moreover, since this probability is $\Omega(1)$ we can use fixed-point amplitude amplification~\cite{yoder2014fixed} to boost the state preparation success probability to $1-\delta$ using $O(\log(1/\delta))$ calls to the unitary block-encoding $\m{U}_{\m{R}_k}$.

\end{document}